\documentclass[11pt]{article}
\usepackage[margin=1in]{geometry}
\usepackage{amsmath,amssymb,amsthm,mathtools}
\usepackage{enumitem}
\usepackage{booktabs}
\usepackage{hyperref}
\newtheorem{definition}{Definition}[section]
\newtheorem{lemma}[definition]{Lemma}
\newtheorem{proposition}[definition]{Proposition}
\newtheorem{theorem}[definition]{Theorem}
\newtheorem{corollary}[definition]{Corollary}
\newtheorem{remark}[definition]{Remark}
\newtheorem{example}[definition]{Example}

\newcommand{\D}{\mathcal D}
\newcommand{\im}{\operatorname{im}}
\newcommand{\cmp}{\operatorname{cmp}}

\newcommand{\cd}{\operatorname{cd}}
\newcommand{\Front}{\operatorname{Front}}
\newcommand{\ifrag}{\operatorname{ifrag}}
\newcommand{\eps}{\varepsilon}
\newcommand{\B}{\mathcal B}
\newcommand{\RectCl}{\operatorname{RectCl}}
\newcommand{\frag}{\operatorname{frag}}
\newcommand{\tfrag}{\operatorname{tfrag}}
\newcommand{\Cells}{\operatorname{Cells}}
\newcommand{\Obs}{\mathsf{Obs}}
\newcommand{\Comp}{\mathsf{Comp}}
\newcommand{\Env}{\mathsf{Env}}

\title{Observer--Fragmentation--Exposure Tradeoffs\\
\large From rectangular CFG exposure to ordered MCFG scheduling}
\author{Takayuki Kuriyama\\
Independent Researcher, Tokyo, Japan\\
\texttt{growup.kuriyama@gmail.com}}
\date{}
\hypersetup{
  hidelinks,
  pdftitle={Observer--Fragmentation--Exposure Tradeoffs: From rectangular CFG exposure to ordered MCFG scheduling},
  pdfauthor={Takayuki Kuriyama},
  pdfkeywords={grammatical inference, positive data, context-free grammar, multiple context-free grammar, finite monoid, observer, characteristic sample, fan-out}
}

\begin{document}
\maketitle

\begin{abstract}
We study finite resources governing exact positive reconstruction in fixed-observation CFG and MCFG learning.  For an explicit rigid CFG family we compute safe observation size, internal residual fragmentation, and characteristic-data cost exactly.  The result is a two-point Pareto frontier: after compulsory rigidity witnesses are fixed, exact reconstruction reduces to connectivity inside observer fibers, and the variable part of every minimum characteristic sample is a spanning forest of complete bipartite fiber graphs.

For bounded-fan-out MCFG reconstruction, pure transition fragmentation multiplies across children.  An explicit fan-out-two family $X_{k,r}$ therefore has exact characteristic-data costs $2r[1+r(k-1)]$ and $2rk^r$ under two comparable observers.  An integral lattice invariant yields an affine-span lower bound and a unimodularity test.  In the binary-index subfamily, unimodularity suffices for minimum-cardinality samples through ranks two and three but not rank four.  Two distinct obstructions then appear: an order-independent laminar support conflict, and an order-sensitive occurrence-scheduling conflict.

For disjoint child requirements, the exact scheduling threshold is the largest monochromatic run count in the doubled reduced slot-colour word; in the two-colour case this is an alternation threshold.  Horn nonlocking and Cartesian locking certificates make these constraints explicit.  Hierarchical reuse can trade parent-root exposure for local fan-out: for the natural critical-module library of the mixed rank-four shapes, the exact width--anchor frontiers are $\{(2,1)\}$, $\{(2,2),(3,1)\}$, and $\{(4,1)\}$ for separated, nested, and crossing orders.  Thus observation, fragmentation, exposure, arithmetic span, laminar compatibility, ordered scheduling, and hierarchical reuse are genuinely distinct finite resources.
\end{abstract}

\section{Introduction}
\label{sec:introduction}

Positive-data grammar-learning proofs often package several finite resources
into one characteristic-sample argument.  This paper separates them.  The
organizing question is:
\[
 \boxed{\text{Which finite resources are actually required for exact positive reconstruction?}}
\]
We distinguish finite observation, semantic fragmentation, positive exposure,
arithmetic span, ordered scheduling, and hierarchical reuse.  The same
learner can behave very differently depending on how information is allocated
among these resources.

The first half gives an exact CFG laboratory.  A rigid-wrapper family
$R_{m,q}$ has a two-point global Pareto frontier relating observer size,
internal-live residual fragmentation, and characteristic-data cost.  After its
compulsory rigidity witnesses are fixed, positive exposure is governed by a
typed rectangular closure, so the variable parts of minimum locking samples
are spanning trees in observer fibers.

The second half shows why higher-rank reconstruction is not merely the same
graph theory with hyperedges.  For a fan-out-two family $X_{k,r}$, observer
refinement fragments one rank-$r$ transition into $k^r$ cells and causes an
exact exponential characteristic-data increase.  Integral span gives a sharp
arithmetic lower bound, but, among minimum-cardinality samples in the
binary-index family studied here, the first obstruction beyond unimodularity
appears at rank four.  Two layers must then be separated.  The sample
$K_{\diamond}$ has an order-independent laminar support conflict: the required
upward and downward corrections cannot be placed in disjoint sibling
occurrences at any fan-out.  The mixed rank-four shapes exhibit the genuinely
order-sensitive layer, where disjoint requirements may or may not be schedulable
depending on their physical order.  We give an exact run-count characterization
of the local fan-out needed to schedule several child requirements, and then
show how hierarchical reuse can trade additional exposed parent configurations
for smaller local fan-out.  The languages $X_{k,r}$ are finite test families: their
role is to isolate learner-relative resource geometry, not to separate MCFG
language classes.

The main results can be read as five packages.  First, the CFG family has an
exact observer--fragmentation--exposure Pareto frontier.  Second, rectangular
exposure is exactly a connectivity problem with spanning-tree minima.  Third,
higher-rank transition fragmentation multiplies across children and yields an
exact exponential observer-sensitive data gap.  Fourth, arithmetic span is
strictly weaker than derivational realizability; the latter separates into
laminar support compatibility and, when supports are disjoint, ordered
occurrence scheduling, both captured by the exact trace system.
Fifth, for a fixed natural critical-module library, hierarchical Cartesian
certificates give an exact width--anchor tradeoff on the critical rank-four
shapes.  Appendices collect weighted/refinement
extensions, abstract incidence interfaces, residual-core connections, the
complete rank-four census, witness-hypergraph machinery, and auxiliary
classification results.

\subsection*{Relation to prior work and scope}

The learning setting belongs to identification from positive data in the sense of
Gold and Angluin \cite{Gold1967,Angluin1980}, and more specifically to the
distributional/substitutability line of grammatical inference.  Clark and Eyraud proved polynomial
identification of substitutable context-free languages from positive data
\cite{ClarkEyraud2007}; Yoshinaka developed fixed-window variants and then
multidimensional substitutability for multiple context-free languages
\cite{Yoshinaka2008,Yoshinaka2011}.  A broader uniform view of such
distributional algorithms was given by Yoshinaka \cite{Yoshinaka2015}, while
Clark and Yoshinaka treated distributional learning for parallel MCFGs in a
richer setting that also uses membership queries \cite{ClarkYoshinaka2014}.
The present paper does not claim a new general learnability theorem for CFGs or
MCFGs.  It studies the finite \emph{resources inside} fixed-observation
positive reconstruction: how observation choice fragments semantic states and
rules, how many positive witnesses must be exposed, and when exposed local
operations can be scheduled coherently.

Finite typing as an inductive bias also has precedents in automata inference
\cite{CosteTyping2004}.  The immediate technical starting points here are the
fixed-finite-monoid CFG and MCFG reconstruction schemes
\cite{KuriyamaCFG,KuriyamaMCFG}.  We recall only the rule mechanisms needed for
the resource analysis below; their general soundness and identification
results are not reproved.  Canonical and residual structure is related to a
separate line on strong learning and canonical CFGs \cite{Clark2015}; our
objective here is weaker in one sense and more quantitative in another: we do
not identify a unique target presentation, but instead measure the finite
observation and exposure needed for exact language reconstruction.

Several combinatorial ingredients used below are classical.  In particular,
difunctional (rectangular) relations go back to Riguet, spanning-tree minima
are graphic-matroid facts, and lattice index/unimodularity are standard integer
linear algebra.  We therefore make no novelty claim for those facts in
isolation.  The contributions are their exact realization inside the stated
CFG/MCFG learners, the observer-sensitive cost formulas, and the ordered
occurrence-scheduling and hierarchical-reuse phenomena proved for the explicit
families below.  The Horn and Cartesian certificates are sound reusable
criteria; except where an exact family-specific characterization is stated, we
do not claim that they form a complete decision procedure for arbitrary MCFG
samples.

\section{Reconstruction setting and finite observation}

\subsection{Distributional and characteristic-data conventions}

For a language $L\subseteq\Sigma^*$ and a nonempty word $x\in\Sigma^+$, write
\[
  \D_L(x):=\{(u,v)\in\Sigma^*\times\Sigma^*:uxv\in L\}
\]
for its complete two-sided distribution.  Two nonempty live factors $x,y$ (that is,
$\D_L(x),\D_L(y)\ne\varnothing$) are
\emph{in conflict} when
\[
  \D_L(x)\cap\D_L(y)\ne\varnothing
  \qquad\text{and}\qquad
  \D_L(x)\ne\D_L(y).
\]
A finite-monoid morphism $h:\Sigma^*\to M$ is a \emph{safe observer for $L$}
when every conflicting pair is separated by $h$.  Equivalently, $L$ is
\emph{$h$-substitutable} when
\[
  h(x)=h(y),\quad
  \D_L(x)\cap\D_L(y)\ne\varnothing
  \quad\Longrightarrow\quad
  \D_L(x)=\D_L(y)
\]
for all nonempty $x,y$.  We write
\[
  \cmp_1(L):=
  \min\{|\im h|:h\text{ is a safe finite-monoid observer for }L\}
\]
whenever such observers exist.

For an ordered $d$-tuple $\mathbf x=(x_1,\ldots,x_d)$ of nonempty strings,
let $\D_L^{(d)}(\mathbf x)$ be the set of complete sentence contexts with
$d$ ordered holes that yield a word of $L$ when filled by $\mathbf x$.
For a finite-monoid observer $h$, put
\[
  h^{(d)}(\mathbf x):=(h(x_1),\ldots,h(x_d)).
\]
The language is \emph{$(f,h)$-tuple-substitutable} when, for every
$1\le d\le f$, equal componentwise $h$-type together with one common complete
$d$-tuple context implies equality of the complete tuple distributions.  This
is the convention used below for the MCFG learner.

A \emph{pure residual class} of arity $d$ is an equivalence class under
$\mathbf x\equiv_L^d\mathbf y$ iff
$\D_L^{(d)}(\mathbf x)=\D_L^{(d)}(\mathbf y)$.  For such a class $R$, define
\[
  \rho_h(R):=\{h^{(d)}(\mathbf x):\mathbf x\in R\}.
\]
The unary case $d=1$ is used in the CFG section.

Finally, for any set-driven reconstruction operator $\mathcal A$ sound on
positive subsets of $L$, a finite $C\subseteq L$ is \emph{characteristic for
$\mathcal A$} when
\[
  C\subseteq K\subseteq L,\ K\text{ finite}
  \quad\Longrightarrow\quad
  L(\mathcal A(K))=L.
\]
Thus ``characteristic'' throughout this paper means a locking sample for the
specified reconstruction operator, not for arbitrary learners.

Let $h:\Sigma^*\to M$ be a morphism into a finite monoid.  For a finite
sample $K\subseteq\Sigma^*$, write $\B_h(K)$ for the fixed-$h$ CFG
reconstruction grammar of \cite{KuriyamaCFG}, recalled here because its rule
geometry is analyzed quantitatively below.  Its observed states are
\[
 [x:u,v]\qquad(uxv\in K,\ x\in\Sigma^+),
\]
and its relevant rule families are
\[
\begin{array}{rll}
(R1)&[xy:u,v]\to[x:u,yv][y:ux,v],&x,y\in\Sigma^+,\\
(R2)&[x:u,v]\to[x:u',v'],&uxv,u'xv'\in K,\\
(R3)&[x:u,v]\to[x':u,v],&h(x)=h(x'),\\
(R4)&[a:u,v]\to a,&a\in\Sigma,\\
(R5)&\widehat S\to[w:\eps,\eps],&w\in K\cap\Sigma^+.
\end{array}
\]
We use the encoded sample size
\[
 \|K\|:=\sum_{w\in K}(|w|+1).
\]
For a target $L$ on which $\B_h$ is sound, define
\[
 \cd_h(L):=\min\{\|C\|: C\subseteq L\text{ is characteristic for }\B_h\}.
\]

For the MCFG part we use the following set-driven fixed-observation
occurrence learner, motivated by the sentence-interface construction of
\cite{KuriyamaMCFG} but specified here independently because the present
finite test families use arbitrary witness rank.  Fix a fan-out bound $f$.
For every sample word, enumerate all ordered tuple occurrences of arity at
most $f$ whose components are nonempty disjoint intervals listed left-to-right;
there is one state $[\mathbf x]$ for each observed tuple value.  The hypothesis
$\widehat G_h(K)$ has the following rules:
\begin{enumerate}[label=(M\arabic*)]
\item a start rule to every sampled full word and a rank-zero constant rule
      for every observed tuple value;
\item for every observed parent occurrence together with a segmentation of
      its components into fixed terminal gaps and nonempty labelled child
      intervals, emit the corresponding composition rule whenever the induced
      template is linear, nondeleting, nonpermuting, nonmerging, and has
      nonempty output components;
\item observed tuples of the same arity and componentwise $h$-type are joined
      by semantic unary identity-template rules in both directions whenever
      they occur in one shared concrete tuple context in the sample.
\end{enumerate}
A composition witness of rank $r$ uses $r$ nonempty child tuples with pairwise
disjoint terminal support inside its observed parent.  Hence, if the parent
occurs in a sample word $w$, then $r\le |w|\le\|K\|_+$; the rank-$r$ witnesses
used for $X_{k,r}$ therefore lie within the finite witness search automatically.
All characteristic-data statements for $X_{k,r}$ below refer to this explicitly
specified operator, not to arbitrary positive-data learners.  The nonempty,
left-to-right occurrence convention is important below: the exact laminar
trace system is defined to match it literally.  General learning results for
related fixed-observation MCFG constructions are given in
\cite{KuriyamaMCFG}; the local soundness property needed here is proved below.

\begin{definition}[Observer refinement]
For finite-monoid morphisms $h:\Sigma^*\to M$ and $g:\Sigma^*\to N$, write
$h\preceq g$ if there is a monoid morphism $\pi:N\to M$ such that
$h=\pi\circ g$.
\end{definition}

\begin{proposition}[Learner monotonicity under observer refinement]
\label{prop:learner-monotone}
If $h\preceq g$, then for every finite sample $K$,
\[
 L(\B_g(K))\subseteq L(\B_h(K)).
\]
If moreover $L$ is $h$-substitutable, then it is also $g$-substitutable and
\[
 \cd_h(L)\le \cd_g(L).
\]
\end{proposition}

\begin{proof}
The observed state set and Rules (R1), (R2), (R4), and (R5) are independent
of the chosen observer.  If Rule (R3) is available for $g$, then
$g(x)=g(x')$ and hence $h(x)=\pi(g(x))=\pi(g(x'))=h(x')$, so the same (R3)
instance is available for $h$.  Thus $\B_g(K)$ is a subgrammar of
$\B_h(K)$.

If $L$ is $h$-substitutable and $g(x)=g(y)$, then $h(x)=h(y)$, so the
$h$-substitution implication gives the $g$-substitution implication.  Let
$C$ be characteristic for $\B_g$.  For every $C\subseteq K\subseteq L$,
soundness gives
\[
 L=L(\B_g(K))\subseteq L(\B_h(K))\subseteq L,
\]
so $C$ is characteristic for $\B_h$ as well.
\end{proof}

\begin{remark}
Hence a genuine tradeoff in which a larger observer yields smaller semantic
fragmentation or smaller characteristic data cannot occur along a refinement
chain.  Such a tradeoff necessarily uses incomparable observers.
\end{remark}

\section{Exact CFG observer--fragmentation--exposure tradeoff}
\label{sec:cfg-tradeoff}

\subsection{The rigid-wrapper family}

Fix integers $m,q\ge2$.  Let
\[
 P_q=\{p_0,\ldots,p_{q-1}\},\qquad
 A_m=\{a_0,\ldots,a_{m-1}\},
\]
and let
\[
 Z_{m,q}=\{z_{s,i}:0\le s<q,\ 0\le i<m\}
\]
be pairwise distinct fresh symbols.  Put
\[
 \Sigma_{m,q}:=P_q\,\dot\cup\,A_m\,\dot\cup\,Z_{m,q}
\]
and define
\begin{equation}
\label{eq:Rmq}
 R_{m,q}
 :=
 \{p_sa_i a_j:0\le s<q,\ 0\le i,j<m\}
 \cup
 \{p_sa_i z_{s,i}:0\le s<q,\ 0\le i<m\}.
\end{equation}
The first block will be called the \emph{main block}; the second consists of
\emph{rigidity witnesses}.  Every target word has length three.

\subsection{Exact live distributions and conflicts}

\begin{lemma}[Live factor distributions]
\label{lem:distributions}
The nonempty live factors of $R_{m,q}$ have the following distributions.
\begin{align*}
\D(p_s)
 &=\{(\eps,a_i a_j):i,j<m\}
   \cup\{(\eps,a_i z_{s,i}):i<m\},\\
\D(a_i)
 &=\{(p_s,a_j):s<q,\ j<m\}
   \cup\{(p_s,z_{s,i}):s<q\}\\
 &\qquad\cup\{(p_sa_k,\eps):s<q,\ k<m\},\\
\D(z_{s,i})
 &=\{(p_sa_i,\eps)\},\\
\D(p_sa_i)
 &=\{(\eps,a_j):j<m\}\cup\{(\eps,z_{s,i})\},\\
\D(a_i a_j)
 &=\{(p_s,\eps):s<q\},\\
\D(a_i z_{s,i})
 &=\{(p_s,\eps)\},\\
\D(w)&=\{(\eps,\eps)\}\qquad(w\in R_{m,q}).
\end{align*}
Consequently:
\begin{enumerate}[label=(\roman*)]
\item the $p_s$ are pairwise conflicting;
\item the $a_i$ are pairwise conflicting;
\item every $a_j$ conflicts with every $z_{s,i}$;
\item the $qm$ prefixes $p_sa_i$ are pairwise conflicting;
\item all main suffixes $a_i a_j$ form one live residual class
\[
 F:=[a_i a_j];
\]
\item for each $s$, all $a_i z_{s,i}$ form one residual class $Z_s$, and
$F$ conflicts with every $Z_s$;
\item factors of different lengths do not conflict.
\end{enumerate}
\end{lemma}

\begin{proof}
The displayed formulae follow by inspection of the two blocks in
\eqref{eq:Rmq}.  Every accepting context for a factor of length $r$ has total
length $3-r$, so live factors of different lengths have disjoint context sets.
The conflict statements are then immediate from overlap and inequality of the
displayed distributions.
\end{proof}

\subsection{Observation complexity and forced live fragmentation}

For a safe observer $h$, define the fragmentation of the distinguished live
class $F$ by
\[
 \kappa_h
 :=
 |\rho_h(F)|
 =
 |\{h(a_i a_j):0\le i,j<m\}|.
\]

\begin{definition}[Internal-live residual fragmentation]
\label{def:ifrag}
Let $R$ be a unary syntactic residual class of a language $L$.  Its
\emph{internal part} is
\[
 R^{\circ}
 :=
 \{x\in R\cap\Sigma^+ :
    \exists(u,v)\in\D_L(x)\text{ with }|u|+|v|>0\}.
\]
Thus $R^{\circ}$ records representatives that occur in at least one proper
(non-root) cut of a target word.  For a safe unary observer $h$, define
\[
 \ifrag_h(L)
 :=
 \max_{R^{\circ}\ne\varnothing}
 \bigl|\{h(x):x\in R^{\circ}\}\bigr|.
\]
This is language-level and does not exclude a residual class merely because it
also contains an accepted word; only the representatives that actually occur
in nontrivial contexts contribute.  It is therefore the residual-fragmentation
quantity naturally aligned with the reusable (R1)--(R4) part of the
reconstruction architecture.
\end{definition}

\begin{theorem}[Exact minimum observation size]
\label{thm:cmp}
For every $m,q\ge2$,
\[
 \cmp_1(R_{m,q})=qm.
\]
\end{theorem}

\begin{proof}
By Lemma~\ref{lem:distributions}(iv), the $qm$ words $p_sa_i$ form a
conflict clique.  Every safe observer must therefore assign them pairwise
distinct values, giving $|\im h|\ge qm$.

For the upper bound, let $C_q$ be the cyclic group of order $q$ written
additively, and let
\[
 M_m=(\{0,\ldots,m-1\},\max)
\]
with identity $0$.  Define $H^-:\Sigma_{m,q}^*\to C_q\times M_m$ by
\[
 H^-(p_s)=(s,0),\qquad
 H^-(a_i)=(0,i),\qquad
 H^-(z_{s,i})=(1,m-1).
\]
Then $H^-(p_sa_i)=(s,i)$, so the image is all of $C_q\times M_m$ and has
size $qm$.  Lemma~\ref{lem:distributions} lists all conflicts.  The $p_s$ are
separated by the first coordinate, the $a_i$ by the second, the $a_j$ from
the $z_{s,i}$ by the first coordinate, the prefixes $p_sa_i$ by the pair of
coordinates, and $F$ from every $Z_s$ because the former has first coordinate
$0$ while the latter has first coordinate $1$.  Hence $H^-$ is safe.
\end{proof}

\begin{theorem}[Minimum observation forces $m$-fold live fragmentation]
\label{thm:min-frag}
If $h$ is safe for $R_{m,q}$ and $|\im h|=qm$, then
\[
 \kappa_h\ge m.
\]
Moreover the bound is attained by $H^-$, so
\[
 \min_{\substack{h\text{ safe}\\|\im h|=qm}}\kappa_h=m.
\]
\end{theorem}

\begin{proof}
Put $y_s=h(p_s)$ and $x_i=h(a_i)$.  The $qm$ prefix values
\[
 B:=\{y_sx_i:s<q,\ i<m\}
\]
are pairwise distinct by safety.  Since $|\im h|=qm$, we have
$\im h=B$.  In particular the monoid identity belongs to $B$, say
\[
 1=y_{s_0}x_{i_0}.
\]
The values $x_0,\ldots,x_{m-1}$ are pairwise distinct because the $a_i$ are
pairwise conflicting.  If
$x_{i_0}x_j=x_{i_0}x_k$, then left multiplication by $y_{s_0}$ gives
$x_j=x_k$, a contradiction.  Hence
\[
 x_{i_0}x_0,\ldots,x_{i_0}x_{m-1}
\]
are $m$ distinct values in $\rho_h(F)$.

For $H^-$,
\[
 H^-(a_i a_j)=(0,\max(i,j)),
\]
so exactly $m$ values occur on $F$.
\end{proof}

\begin{theorem}[Fragmentation gap]
\label{thm:gap}
If $h$ is safe for $R_{m,q}$ and $\kappa_h<m$, then
\[
 |\im h|\ge q(m+1).
\]
This bound is attained with $\kappa_h=1$.
\end{theorem}

\begin{proof}
Again put $y_s=h(p_s)$, $x_i=h(a_i)$ and
$B=\{y_sx_i:s<q,i<m\}$.  The set $B$ has cardinality $qm$.

We claim that no $y_s$ lies in $B$.  Otherwise
$y_s=y_tx_i$ for some $t,i$.  Then, as $j$ varies,
\[
 y_sx_j=y_t(x_ix_j).
\]
The right-hand side takes at most $\kappa_h<m$ values, whereas the $m$
left-hand values $y_sx_j$ are pairwise distinct because the corresponding
prefixes $p_sa_j$ are pairwise conflicting.  This is impossible.

The $y_s$ are themselves pairwise distinct because the $p_s$ conflict.
Hence $\im h$ contains the disjoint union of the $qm$-element set $B$ and
the $q$ row values $\{y_s\}$, proving $|\im h|\ge q(m+1)$.

For attainment, let
\[
 N_m=\{1,c_0,\ldots,c_{m-1}\}
\]
be the monoid with identity $1$ and multiplication
$c_ic_j=c_0$ for all $i,j$.  Define
\[
 H^+(p_s)=(s,1),\qquad
 H^+(a_i)=(0,c_i),\qquad
 H^+(z_{s,i})=(1,c_0)
\]
as a morphism into $C_q\times N_m$.  The same conflict check as in
Theorem~\ref{thm:cmp} shows safety; now
\[
 H^+(a_i a_j)=(0,c_0)
\]
for all $i,j$, so $\kappa_{H^+}=1$.  The image is all of
$C_q\times N_m$ and has size $q(m+1)$.
\end{proof}

\begin{lemma}[The critical class controls internal-live fragmentation]
\label{lem:ifrag-kappa}
For every safe observer $h$ of $R_{m,q}$,
\[
 \ifrag_h(R_{m,q})\ge\kappa_h.
\]
For the two explicit observers,
\[
 \ifrag_{H^-}(R_{m,q})=m,
 \qquad
 \ifrag_{H^+}(R_{m,q})=1.
\]
\end{lemma}

\begin{proof}
The first inequality is immediate because the main-suffix class $F$ is live,
proper, and has fiber size $\kappa_h$.  By
Lemma~\ref{lem:distributions}, every other proper live residual class is either
a singleton class or one of the classes
\[
 Z_s=[a_i z_{s,i}]\qquad(s<q).
\]
Under $H^-$ all words in $Z_s$ have value $(1,m-1)$, while under $H^+$ they
all have value $(1,c_0)$.  Hence these classes do not fragment under either
explicit observer.  The class $F$ has respectively $m$ and one observer
values, proving the equalities.
\end{proof}

\subsection{Rectangular exposure and spanning-tree structure}

The graph argument used below is an instance of a classical relation-theoretic
notion.  A binary relation $E\subseteq X\times Y$ is \emph{difunctional}
(or rectangular) when
\[
 (x,y),(x',y),(x',y')\in E
 \quad\Longrightarrow\quad
 (x,y')\in E.
\]
Equivalently, $EE^{-1}E\subseteq E$.  Difunctional relations go back to
Riguet and are precisely disjoint unions of rectangles; see, e.g.,
\cite{Riguet1950,BackhouseOliveira2023}.  What is specific here is that
fixed-observation reconstruction imposes this closure separately inside each
observer fiber.

\begin{definition}[Typed rectangular closure]
Let $P$ and $C$ be finite sets and let
\[
 C=C_1\dot\cup\cdots\dot\cup C_k
\]
be a fixed partition.  For $E\subseteq P\times C$, write
$E_t=E\cap(P\times C_t)$.  The \emph{typed rectangular closure}
$\RectCl_{\mathcal C}(E)$ is the least relation containing $E$ such that,
for every $t$, its restriction to $P\times C_t$ is difunctional.
\end{definition}

\begin{theorem}[Rectangular closure theorem]
\label{thm:rectangular-closure}
For each $t$, let $G_t(E)$ be the bipartite graph with vertex classes $P$ and
$C_t$ and edge set $E_t$.  If $\Gamma$ ranges over the nontrivial connected
components of $G_t(E)$ and $P_\Gamma,C_\Gamma$ denote their row and column
vertices, then
\[
 \boxed{
 \RectCl_{\mathcal C}(E)
 =
 \bigcup_{t=1}^k\ \bigcup_{\Gamma\in\pi_0(G_t(E))}
 P_\Gamma\times C_\Gamma .
 }
\]
Consequently,
\[
 \RectCl_{\mathcal C}(E)=P\times C
\]
if and only if every $G_t(E)$ is connected and spanning.
\end{theorem}

\begin{proof}
Inside a fixed block $C_t$, every connected component becomes a complete
bipartite graph under the three-corner rule: along a path, repeated
applications fill one missing corner at a time.  Hence the right-hand side is
contained in every difunctional relation containing $E_t$.  Conversely a
union of the displayed component rectangles is difunctional and contains
$E_t$, so it is the least such relation.  Different type blocks do not
interact.  The final statement is immediate.
\end{proof}

\begin{corollary}[Minimum rectangular exposure and matroid structure]
\label{cor:rectangular-minimum}
A set $E\subseteq P\times C$ satisfies
$\RectCl_{\mathcal C}(E)=P\times C$ with minimum cardinality exactly when,
for every $t$, $E_t$ is a spanning tree of the complete bipartite graph
$K_{|P|,|C_t|}$.  Hence
\[
 \boxed{
 |E|_{\min}
 =|C|+k(|P|-1).
 }
\]
Moreover all inclusion-minimal sufficient exposures already have this minimum
cardinality, and the minimum exposure sets are precisely the bases of the
direct sum of the graphic matroids of the graphs $K_{|P|,|C_t|}$.
Their number is
\[
 \boxed{
 \prod_{t=1}^k |P|^{|C_t|-1}|C_t|^{|P|-1}.
 }
\]
\end{corollary}

\begin{proof}
By Theorem~\ref{thm:rectangular-closure}, each block graph must be connected
and spanning.  A connected graph on $|P|+|C_t|$ vertices has at least
$|P|+|C_t|-1$ edges, with equality exactly for a tree.  Summing over $t$
gives the cardinality formula.  Inclusion-minimal connected spanning graphs
are trees, giving the matroid statement.  The counting formula is the standard
matrix--tree formula for $K_{r,s}$, namely $r^{s-1}s^{r-1}$, multiplied over
independent blocks.
\end{proof}

\paragraph{Weighted extension.}
The same rectangular decomposition admits an exact additive-weighted
minimum-spanning-tree formula.  Since the weighted statement is not used by
any later main theorem, we defer it, together with the exact fiber-splitting
penalties, to Appendix~\ref{app:weighted-rectangular}.  The unweighted
connectivity and spanning-tree structure above is the only rectangular input
needed in the sequel.

\subsection{Exact characteristic-data cost}

Let the $h$-fibers of the main suffix class $F$ be
\[
 \mathcal C_h=\{C_1,\ldots,C_{\kappa_h}\},
\]
where each $C_t$ is a subset of
\[
 A_mA_m=\{a_i a_j:0\le i,j<m\}
\]
and words in the same block have the same $h$-value.

For $K\subseteq R_{m,q}$ and each $C_t$, define the bipartite sample graph
$G_t(K)$ with left vertex set $P_q$, right vertex set $C_t$, and edge
$(p_s,c)$ exactly when $p_sc\in K$.

\begin{lemma}[Rigidity of the non-suffix parse]
\label{lem:rigid-parse}
Let $K\subseteq R_{m,q}$ and let $h$ be safe.  In a derivation of a main word
$p_sa_i a_j$ from a sampled root, the parse shape
\[
 (p_sa_i)\mid a_j
\]
cannot change either terminal factor.  More precisely, a hypothesis state
whose observed factor is a prefix $p_sa_i$ can derive no different length-two
prefix, and a one-letter state occurring in the corresponding last-letter
context can derive no different terminal.
\end{lemma}

\begin{proof}
Use the soundness invariant of the fixed-$h$ constructor.  A word $w$ derived
from $[x:u,v]$ satisfies $uwv\in R_{m,q}$ and $h(w)=h(x)$.  If $x=p_sa_i$
and the surrounding context has length one, then every possible $w$ has length
two and must itself be a prefix $p_ta_k$.  Such prefixes are pairwise
conflicting by Lemma~\ref{lem:distributions}, hence safety makes their
$h$-values pairwise distinct; therefore $w=x$.  The same argument applies to
the one-letter child: the only possible one-letter replacements in its fixed
accepting context are members of a conflict family, hence equal $h$-type
forces the original terminal.  In particular a private $z$-symbol cannot
replace an $a$-symbol, because every $a_j$ conflicts with every $z_{s,i}$.
\end{proof}

\begin{lemma}[Suffix-component invariant]
\label{lem:suffix-component}
Fix an $h$-fiber $C_t$ and a sampled main edge $(p_s,c)$ with $c\in C_t$.
Consider the suffix state $[c:p_s,\eps]$.  If this state derives a different
main suffix $c'\in C_t$, then the sampled edge carrying $c'$ lies in the same
connected component of $G_t(K)$ as $(p_s,c)$.
\end{lemma}

\begin{proof}
At the level of a length-two main-suffix state, an (R2) step preserves the
literal suffix and changes only between rows in which that same suffix has
been sampled; this moves between two edges of $G_t(K)$ sharing a right
vertex.  An (R3) step preserves the row context and changes between sampled
suffixes of the same $h$-type; this moves between two edges sharing a left
vertex.  Thus every unit (R2)/(R3) step stays in one connected component of
$G_t(K)$.

It remains to exclude a detour through a binary expansion.  Once a main suffix
$c=a_i a_j$ is expanded by (R1), the two one-letter children cannot change
their terminal yields: this follows from the same soundness-and-conflict
argument as Lemma~\ref{lem:rigid-parse}.  Hence such a branch yields exactly
$c$ and cannot reassemble as a different suffix.  Moreover an auxiliary
suffix $a_i z_{s,i}$ cannot be an equal-$h$ intermediate for a main suffix:
its residual class $Z_s$ conflicts with $F$, so safety separates every word of
$Z_s$ from every word of $F$.  Therefore every derivation that changes the
main suffix does so entirely through the top-level (R2)/(R3) unit graph, and
hence remains in the original component of $G_t(K)$.
\end{proof}

\begin{lemma}[Sample-graph characterization]
\label{lem:graph-characterization}
Let $h$ be safe for $R_{m,q}$ and $K\subseteq R_{m,q}$.  Then
$L(\B_h(K))=R_{m,q}$ if and only if both of the following hold:
\begin{enumerate}[label=(\roman*)]
\item every rigidity word $p_sa_i z_{s,i}$ belongs to $K$;
\item for every $C_t\in\mathcal C_h$, the graph $G_t(K)$ is connected and
spanning.
\end{enumerate}
\end{lemma}

\begin{proof}
Suppose first that (i)--(ii) hold.  Every rigidity word belongs to the
hypothesis by sample consistency.  Fix a desired main word $p_sc$ with
$c\in C_t$.  Since $G_t(K)$ is spanning, row $p_s$ is incident with some
sampled edge $(p_s,c_0)$.  Start from the sampled root $p_sc_0$ and apply
(R1) with the split $p_s\mid c_0$.  The right child is the state
$[c_0:p_s,\eps]$.  The one-letter left child cannot change to another $p_t$:
there is no nontrivial (R1) split, and an (R3) replacement would require equal
$h$-type, while the pairwise conflict of the $p_s$ and safety of $h$ separate
their types.  Thus the left child remains $p_s$.  Connectivity of $G_t(K)$ gives an alternating edge path
from $(p_s,c_0)$ to an edge carrying $c$.  Along this path, shared right
vertices are implemented by (R2) and shared left vertices by (R3).  The
resulting suffix state derives $c$, while the left child still derives $p_s$.
Hence the hypothesis generates $p_sc$.  This proves target inclusion, and
soundness gives equality.

Conversely assume $L(\B_h(K))=R_{m,q}$.  If a rigidity word
$p_sa_i z_{s,i}$ were absent from $K$, then the private terminal $z_{s,i}$
would occur in no sample word.  No observed state labelled by $z_{s,i}$ would
exist, so Rule (R4) could never emit that terminal.  Thus (i) is necessary.

Fix $C_t$ and suppose $G_t(K)$ is not connected and spanning.  If a right
vertex $c\in C_t$ or a row $p_s$ is isolated from all sampled edges, choose a
target main word $p_sc$ using that missing incidence.  More generally choose
$p_sc$ so that every sampled edge incident with row $p_s$ lies in a component
different from every sampled edge carrying $c$; such a pair exists whenever
the graph is not connected and spanning.  Any derivation of $p_sc$ must start
from a sampled full word.  A root-level (R3) step can only move to another
sampled full word, so it contributes no unsampled row--suffix incidence.
If the derivation uses the split $(p_ra_i)\mid a_j$, then
Lemma~\ref{lem:rigid-parse} shows that it cannot change the sampled prefix or
last terminal and therefore cannot create the chosen missing incidence.  The
only remaining productive mechanism is a split $p_r\mid c_0$ followed by a
change of the suffix child.  By Lemma~\ref{lem:suffix-component}, that child
can reach only suffixes in the same component of $G_t(K)$ as its sampled
starting edge.  It therefore cannot produce the chosen $p_sc$.  This
contradicts $L(\B_h(K))=R_{m,q}$, proving (ii).
\end{proof}

\begin{proposition}[Self-locking criterion for the fixed-$h$ constructor]
\label{prop:self-locking}
Let $L$ be $h$-substitutable and $C\subseteq L$ finite.  Then
\[
 C\text{ is characteristic for }\B_h
 \quad\Longleftrightarrow\quad
 L(\B_h(C))=L.
\]
\end{proposition}

\begin{proof}
The forward implication is the characteristic-sample definition with $K=C$.
For the converse, if $C\subseteq K\subseteq L$, every observed state and every
rule instance present in $\B_h(C)$ remains present in $\B_h(K)$.  Hence
$\B_h(C)$ is a subgrammar of $\B_h(K)$ and
\[
 L=L(\B_h(C))\subseteq L(\B_h(K)).
\]
Soundness on positive subsets of $L$ gives the reverse inclusion.
\end{proof}

\begin{theorem}[Exact characteristic-sample formula]
\label{thm:cd}
For every safe observer $h$,
\[
 \min\{|C|:C\text{ is characteristic for }\B_h\}
 =qm+m^2+\kappa_h(q-1).
\]
Equivalently, under the encoded sample measure
$\|K\|=\sum_{w\in K}(|w|+1)$,
\[
 \boxed{\cd_h(R_{m,q})
 =4\bigl(qm+m^2+\kappa_h(q-1)\bigr).}
\]
\end{theorem}

\begin{proof}
By Lemma~\ref{lem:graph-characterization}, all $qm$ rigidity words are
mandatory.  For a block $C_t$, a connected spanning bipartite graph on
$q+|C_t|$ vertices needs at least $q+|C_t|-1$ edges, and such a tree exists
because all row--column pairs are target words.  Hence the minimum number of
main examples is
\[
 \sum_{t=1}^{\kappa_h}(q+|C_t|-1)
 =\kappa_h(q-1)+\sum_t|C_t|
 =\kappa_h(q-1)+m^2.
\]
Adding the rigidity words gives the cardinality formula.  Every word has
length three, so encoded size is four times cardinality.

A minimum sample $C$ with $L(\B_h(C))=R_{m,q}$ is characteristic because the
constructor is syntactically monotone in the sample: if $C\subseteq K$, then
$\B_h(C)$ is a subgrammar of $\B_h(K)$; soundness for $K\subseteq R_{m,q}$
then yields equality with the target.
\end{proof}

\begin{corollary}[Minimum characteristic samples are spanning-tree bases]
\label{cor:tree-count}
Let $n_t:=|C_t|$.  After the compulsory $qm$ rigidity words are fixed, the
minimum characteristic samples are obtained independently by choosing a
spanning tree of the complete bipartite graph $K_{q,n_t}$ for each fiber
$C_t$.  Consequently their number is
\[
 \boxed{
 \prod_{t=1}^{\kappa_h} q^{\,n_t-1} n_t^{\,q-1}.
 }
\]
Equivalently, the variable part of a minimum characteristic sample is a basis
of the direct sum of the graphic matroids of the complete bipartite graphs
$K_{q,n_t}$.
\end{corollary}

\begin{proof}
Lemma~\ref{lem:graph-characterization} says that each fiber graph must be
connected and spanning, while Theorem~\ref{thm:cd} shows that minimum size is
attained exactly at $q+n_t-1$ edges, i.e. at spanning trees.  The classical
matrix--tree formula gives $q^{n_t-1}n_t^{q-1}$ spanning trees of
$K_{q,n_t}$, and the choices for distinct fibers are independent.
\end{proof}

\begin{corollary}[Rectangular exposure is the combinatorial skeleton of the family]
\label{cor:family-rectangular}
For the main block of $R_{m,q}$, after the compulsory rigidity words are fixed,
the words generated by the fixed-$h$ learner from a main sample are exactly
the typed rectangular closure of its row--suffix incidence relation, where the
column partition is the $h$-fiber partition $\mathcal C_h$.  Hence
Theorem~\ref{thm:cd} is the specialization of
Corollary~\ref{cor:rectangular-minimum} with $|P|=q$, $|C|=m^2$, and
$k=\kappa_h$.
\end{corollary}

\begin{remark}[Variable-length extension]
The equal-length family makes encoded sample cost a constant multiple of
cardinality.  In a rectangular family whose main words factor uniquely as
$pc$ with possibly varying lengths, the natural cost
$|pc|+1=|p|+|c|+1$ is additive.  Whenever the same rigidity argument reduces
reconstruction to typed rectangular closure,
Theorem~\ref{thm:weighted-rectangular} gives the exact minimum encoded
characteristic-data cost without an equal-length assumption.
\end{remark}

\subsection{The exact three-resource Pareto frontier}

For safe $h$, define the profile
\[
 \Pi_h(R_{m,q})
 :=
 \bigl(|\im h|,\ \kappa_h,\ \cd_h(R_{m,q})\bigr)
 \in\mathbb N^3.
\]
Use componentwise minimization.

\begin{theorem}[Observer--fragmentation--exposure tradeoff]
\label{thm:frontier}
The global Pareto frontier over all safe finite-monoid observers of
$R_{m,q}$ consists of exactly two points:
\[
 \boxed{
 P^-_{m,q}
 =
 \bigl(qm,\ m,\ 4(m^2+2qm-m)\bigr)
 }
\]
and
\[
 \boxed{
 P^+_{m,q}
 =
 \bigl(q(m+1),\ 1,\ 4(m^2+qm+q-1)\bigr).
 }
\]
\end{theorem}

\begin{proof}
The observer $H^-$ attains $P^-_{m,q}$ by
Theorems~\ref{thm:cmp}, \ref{thm:min-frag}, and \ref{thm:cd}; $H^+$ attains
$P^+_{m,q}$ by Theorems~\ref{thm:gap} and \ref{thm:cd}.

Let $h$ be any safe observer.  If $\kappa_h\ge m$, then
$|\im h|\ge qm$ by Theorem~\ref{thm:cmp}, while
Theorem~\ref{thm:cd} is increasing in $\kappa_h$; hence $P^-_{m,q}$
componentwise dominates $\Pi_h$.

If $\kappa_h<m$, Theorem~\ref{thm:gap} gives
$|\im h|\ge q(m+1)$, while $\kappa_h\ge1$ and
Theorem~\ref{thm:cd} gives the corresponding lower bound attained at
$\kappa_h=1$; hence $P^+_{m,q}$ dominates $\Pi_h$.
Thus no third Pareto-minimal point exists.
\end{proof}

\begin{corollary}[Constant observation overhead, unbounded savings]
For $q=2$, the two frontier points are
\[
 (2m,\ m,\ 4m(m+3))
 \qquad\text{and}\qquad
 (2m+2,\ 1,\ 4(m+1)^2).
\]
Thus adding only two observation values eliminates $m-1$ copies of the
critical live residual state and saves $4(m-1)$ units of encoded positive
characteristic data.
\end{corollary}

\begin{theorem}[Intrinsic proper-live frontier]
\label{thm:intrinsic-frontier}
The same two points form the global Pareto frontier when the distinguished
quantity $\kappa_h$ is replaced by the language-level internal-live
fragmentation $\ifrag_h(R_{m,q})$.  Explicitly, for
\[
 \Pi_h^{\mathrm{int}}(R_{m,q})
 :=
 \bigl(|\im h|,\ \ifrag_h(R_{m,q}),\ \cd_h(R_{m,q})\bigr),
\]
one has
\[
 \boxed{
 \Front^{\mathrm{int}}(R_{m,q})
 =\{P^-_{m,q},P^+_{m,q}\}.
 }
\]
\end{theorem}

\begin{proof}
The witnesses $H^-$ and $H^+$ attain the stated second coordinates by
Lemma~\ref{lem:ifrag-kappa}.  Let $h$ be arbitrary and write
$f_h:=\ifrag_h(R_{m,q})$.  If $\kappa_h\ge m$, then
$f_h\ge\kappa_h\ge m$, $|\im h|\ge qm$, and
Theorem~\ref{thm:cd} gives characteristic-data cost at least that of
$P^-_{m,q}$; hence $P^-_{m,q}$ dominates the intrinsic profile.  If
$\kappa_h<m$, Theorem~\ref{thm:gap} gives
$|\im h|\ge q(m+1)$, while $f_h\ge1$ and
Theorem~\ref{thm:cd} is minimized at $\kappa_h=1$; hence $P^+_{m,q}$
dominates.  No third Pareto-minimal profile exists.
\end{proof}

\section{Higher-rank fragmentation and exponential exposure}
\label{sec:higher-rank}

\subsection{Pure residual transitions and rank amplification}
\label{sec:transition-fragmentation}

Unary observation fragmentation measures how a finite observation splits
semantic states.  Higher-rank reconstruction has a second amplification mechanism: one
semantic operation may split into many typed rule cells even when the output
state itself has few observed types.  This section isolates that mechanism at
the language level.

Fix an arity $d$ and write
\[
  \mathbf x\equiv^{d}_L\mathbf y
  \quad\Longleftrightarrow\quad
  \D^{(d)}_L(\mathbf x)=\D^{(d)}_L(\mathbf y)
\]
for equality of complete $d$-tuple distributions, using the same nonempty
component convention as the MCFG learner.  A class is \emph{live} when its
distribution is nonempty.  For a finite monoid morphism
$h:\Sigma^*\to M$ and a live pure residual class $R$ of arity $d$, put
\[
  \frag_h(R):=|\rho_h(R)|,
\]
where $\rho_h(R)$ was defined in Section~2.  Thus $\frag_h(R)$ is the number
of learner-side type fibers into which $h$ splits one pure semantic state.

\begin{definition}[Witnessed pure residual transition]
Let $\rho$ be a good tuple template of rank $r$, with input arities
$d_1,\ldots,d_r$ and output arity $e$.  For live pure residual classes $R_j$
of the corresponding arities, we say that $(R_1,\ldots,R_r)$ is
\emph{purely witnessed} for $\rho$ if there are representatives
$\mathbf x_j\in R_j$ such that
$\rho(\mathbf x_1,\ldots,\mathbf x_r)$ has nonempty parent distribution.
\end{definition}

\begin{theorem}[Pure witnessed-operation theorem]
\label{thm:pure-witnessed-operation}
If $(R_1,\ldots,R_r)$ is purely witnessed for $\rho$, then every
choice of representatives $\mathbf y_j\in R_j$ is witnessed, and all parent
tuples
\[
  \rho(\mathbf y_1,\ldots,\mathbf y_r)
\]
lie in one and the same pure parent residual class $R$.  Consequently the
partial operation
\[
  [\rho]_{L}(R_1,\ldots,R_r):=R
\]
is well defined without choosing a finite observer.
\end{theorem}

\begin{proof}
Choose one witnessed tuple of representatives and one accepting parent context
$E$.  Replace the children one at a time.  Because $\rho$ is good,
nonpermutingness and nonmerging ensure that, after all other children are
fixed, $E$ and $\rho$ induce an admissible sentence context for the current
child: its named holes occur left-to-right and every internal separator is
nonempty.  Two representatives in the same pure residual class have identical
complete distributions, so replacing that child preserves acceptance.  After
$r$ replacements, $E$ still accepts the new parent tuple.

The same argument works for an arbitrary parent context $F$: if $F$ accepts
one parent tuple, the same good-template induced-context argument permits the
children to be replaced one at a time, and the reverse replacements give the
converse implication.  Hence all parent tuples have the same complete parent
distribution and therefore the same pure residual class.  The good-template
hypothesis is essential here: without nonmerging, consecutive components of
one child could induce adjacent holes with an empty separator, which is not a
sentence context in the reconstruction convention.
\end{proof}

\begin{proposition}[Soundness of the fixed-observation occurrence learner]
\label{prop:mcfg-local-soundness}
Let $L$ be $(f,h)$-tuple-substitutable and let $K\subseteq L$ be finite.
Then the occurrence learner specified in Section~2 is sound:
\[
  L(\widehat G_h(K))\subseteq L.
\]
More strongly, if an observed arity-$d$ state $[\mathbf x]$ derives a tuple
$\mathbf y$, then
$\D_L^{(d)}(\mathbf y)=\D_L^{(d)}(\mathbf x)$.
\end{proposition}

\begin{proof}
Induct on a derivation from an observed state.  A constant rule changes
nothing.  For a semantic unary rule
$[\mathbf x]\to[\mathbf z]$, the two observed tuples have equal
componentwise $h$-type and one shared complete sentence context; tuple
substitutability therefore gives
$\D_L^{(d)}(\mathbf x)=\D_L^{(d)}(\mathbf z)$, and the induction hypothesis
handles the continuation below $[\mathbf z]$.

For a composition rule witnessed in the sample, write its concrete parent as
\[
  \mathbf x=\rho(\mathbf u_1,\ldots,\mathbf u_s).
\]
The parent occurrence is live because its host sample word belongs to $L$.
By induction, every tuple $\mathbf v_j$ derived from $[\mathbf u_j]$ lies in
the same pure residual class as $\mathbf u_j$.  The
Pure witnessed-operation theorem then gives
\[
  \rho(\mathbf v_1,\ldots,\mathbf v_s)
  \equiv_L^d
  \rho(\mathbf u_1,\ldots,\mathbf u_s)=\mathbf x.
\]
This proves the stronger state invariant.

Finally every start rule is anchored at a sampled word $w\in K\subseteq L$.
If $[w]$ derives $y$, the stronger invariant gives
$\D_L(y)=\D_L(w)$; since $(\eps,\eps)\in\D_L(w)$, also $y\in L$.
\end{proof}

\begin{definition}[Typed transition cells]
Let
\[
  \theta:(R_1,\ldots,R_r)\xrightarrow{\rho}R
\]
be a witnessed pure residual transition.  Its set of $h$-typed child cells is
\[
 \Cells_h(\theta)
 :=\Bigl\{(\mathbf p_1,\ldots,\mathbf p_r):
     \begin{array}{l}
     \exists\mathbf x_j\in R_j\text{ with }
     h^{(d_j)}(\mathbf x_j)=\mathbf p_j\ (1\le j\le r),\\[-1mm]
     \rho(\mathbf x_1,\ldots,\mathbf x_r)
     \text{ has nonempty parent distribution}
     \end{array}
     \Bigr\},
\]
and its \emph{transition fragmentation} is
\[
  \tfrag_h(\theta):=|\Cells_h(\theta)|.
\]
The output type is not listed separately, because multiplicativity of $h$ and
the fixed template determine it from the child-type tuple.
\end{definition}

\begin{theorem}[Transition-fragmentation product law]
\label{thm:tfrag-product}
For every witnessed pure residual transition
$\theta:(R_1,\ldots,R_r)\xrightarrow{\rho}R$,
\[
 \boxed{
   \Cells_h(\theta)
   =\rho_h(R_1)\times\cdots\times\rho_h(R_r),
   \qquad
   \tfrag_h(\theta)
   =\prod_{j=1}^r\frag_h(R_j).
 }
\]
\end{theorem}

\begin{proof}
The inclusion from left to right is immediate.  Conversely choose arbitrary
$\mathbf p_j\in\rho_h(R_j)$ and representatives $\mathbf x_j\in R_j$ of those
types.  Theorem~\ref{thm:pure-witnessed-operation} says that this simultaneous
choice is witnessed and has the same pure parent residual class.  Hence every
Cartesian product cell occurs.
\end{proof}

\begin{corollary}[Rank amplification]
\label{cor:rank-amplification}
If a rank-$r$ witnessed operation has $\frag_h(R_j)=k$ for every child, then
\[
  \boxed{\tfrag_h(\theta)=k^r.}
\]
In particular a live-functional observer gives one typed transition cell,
whereas a $k$-fold fragmentation of every child produces $k^r$ cells.
\end{corollary}

\begin{corollary}[Refinement monotonicity for transition cells]
If $h\preceq g$, then for every live pure residual class
$\frag_h(R)\le\frag_g(R)$ and for every witnessed pure transition
\[
  \tfrag_h(\theta)\le\tfrag_g(\theta).
\]
Thus unary state fragmentation is additive at the state level, while
higher-rank rule fragmentation multiplies the child-fiber counts.
\end{corollary}

\begin{proposition}[Size of the full typed expansion]
\label{prop:typed-expansion-size}
Let $D=(Q,T,I)$ be a finite pure residual diagram, with every
$\theta\in T$ witnessed.  The literal full $h$-typed expansion that keeps one
state for every realized type fiber and one rule for every realized typed copy
has
\[
  |Q_h^\sharp|=\sum_{R\in Q}\frag_h(R)
\]
states and
\[
  \boxed{
  |T_h^\sharp|
  =\sum_{\theta:(R_1,\ldots,R_r)\to R\in T}
      \prod_{j=1}^r\frag_h(R_j)
  }
\]
transition copies, counting distinct target transitions separately.
\end{proposition}

\begin{proof}
The state formula is the definition of full type refinement.  For a fixed
transition the possible child-type tuples are exactly the cells of
Theorem~\ref{thm:tfrag-product}; the parent type is then forced by the template.
Summing over the finitely many target transitions gives the second formula.
\end{proof}

\begin{remark}[Relation to the current MCFG reconstruction proof]
The fixed-observation MCFG proof output-types the target presentation
and then exposes every surviving typed rule.  Proposition~\ref{prop:typed-expansion-size}
identifies a language-level source of the rule-copy count: a single pure
witnessed operation can induce a Cartesian product of child type fibers before
any presentation-specific redundancy reduction is attempted.  The sample-rank
bound controls which ranks need to be searched; it does not remove this
multiplicative type-cell effect.
\end{remark}

\subsection{Cell-isolated positive exposure}

\begin{definition}[Cell-isolated transition interface]
A witnessed pure transition $\theta$ is \emph{$h$-cell-isolated} for a local
reconstruction problem if (i) every typed cell in $\Cells_h(\theta)$ occurs as
a required local generator, (ii) one positive sample word exposes witnesses
from at most one such cell, and (iii) no composite simulation using other local
witness rules can replace a missing cell.  It is \emph{exactly cell-exposable}
if, in addition, every cell has a positive word exposing it.
\end{definition}

\begin{proposition}[Cell-isolated witness count]
\label{prop:cell-isolated-bound}
If $\theta$ is $h$-cell-isolated, every positive exposure set for the local
typed expansion contains at least
\[
  \tfrag_h(\theta)
\]
words devoted to these cells.  If it is exactly cell-exposable, the lower
bound is attained.
\end{proposition}

\begin{proof}
By (i) and (iii), each cell must be witnessed explicitly.  By (ii), one sample
word can discharge at most one cell, giving the lower bound.  Under exact
cell-exposability choose one exposing word per cell.
\end{proof}

\subsection{The fan-out-two family and exact observer gap}
\label{sec:fanout-two-family}

We now realize the preceding multiplication inside the standard MCFG learner.
The example is deliberately finite, so its role is to isolate the exposure
mechanism rather than to separate language classes.

Fix $k,r\ge2$.  For every slot $1\le j\le r$ and index $1\le i\le k$, let
$a_{j,i}$ and $b_{j,i}$ be distinct terminals.  Put
\[
  w_{\mathbf i}
  :=a_{1,i_1}a_{2,i_2}\cdots a_{r,i_r}
    b_{1,i_1}b_{2,i_2}\cdots b_{r,i_r},
  \qquad
  \mathbf i=(i_1,\ldots,i_r)\in[k]^r,
\]
and
\[
  X_{k,r}:=\{w_{\mathbf i}:\mathbf i\in[k]^r\}.
\]
Every word has length $2r$ and $|X_{k,r}|=k^r$.

Using the standard MCFG formalism of \cite{SekiEtAl1991}, a fan-out-two MCFG presentation has states $A_1,\ldots,A_r$ of fan-out two,
rank-zero rules
\[
  A_j\to(a_{j,i},b_{j,i})
  \qquad(1\le i\le k),
\]
and one rank-$r$ start rule
\[
 S\to
 x_1^1x_2^1\cdots x_r^1x_1^2x_2^2\cdots x_r^2
 (A_1,\ldots,A_r).
\]
It generates exactly $X_{k,r}$.

Let $\Gamma_r=\{A_1,\ldots,A_r,B_1,\ldots,B_r\}$ and let
$T_{2r}(\Gamma_r)$ be the finite truncation monoid consisting of all words of
length at most $2r$ together with an absorbing overflow element.  Define the
\emph{slot observer} $h_{\rm slot}$ by
\[
  a_{j,i}\mapsto A_j,
  \qquad
  b_{j,i}\mapsto B_j,
\]
followed by the truncation morphism.  Define the \emph{index observer}
$h_{\rm idx}$ as the corresponding truncation morphism on the full terminal
alphabet, so it is injective on every factor of every target word.  Erasing the
index $i$ gives a monoid quotient
\[
  h_{\rm slot}\preceq h_{\rm idx}.
\]

\begin{lemma}[Slot observer is safe through fan-out two]
\label{lem:slot-safe}
The language $X_{k,r}$ is $(2,h_{\rm slot})$-tuple-substitutable.
\end{lemma}

\begin{proof}
Every target word has the rigid slot skeleton
\[
 A_1A_2\cdots A_rB_1B_2\cdots B_r,
\]
and each skeleton symbol occurs exactly once.  Because the learner compares
only tuples with nonempty components, the componentwise
$h_{\rm slot}$-type of an observed tuple records the exact nonempty skeleton
interval occupied by each component.  Equal type therefore forces two
compared tuple fillings to occupy the same structural intervals.

Suppose two arity-$d$ tuples, $d\le2$, of equal slot type share an accepting
sentence context.  Consider one paired slot $(A_j,B_j)$.  If the context fixes
one of the two index-bearing terminals of that slot while the tuple fills the
other, shared acceptance forces the two compared tuples to carry the same
index there.  If the compared tuples differ at slot $j$, both index-bearing
terminals of that slot must therefore lie inside the tuple holes; each tuple
then carries a consistent but possibly different index for the complete slot.
In any other context accepting the first tuple, the same rigid skeleton places
the holes at the same structural intervals.  Hence a slot on which the tuples
differ cannot have an index-bearing terminal fixed outside the holes, and
replacing the complete slot index preserves membership.  Applying this slot by
slot shows that every accepting context for one tuple accepts the other, and
conversely.  Their complete tuple distributions are equal.
\end{proof}

\begin{corollary}
The finer observer $h_{\rm idx}$ is also safe through fan-out two.
\end{corollary}

For slot $j$, all tuples
\[
  u_{j,i}:=(a_{j,i},b_{j,i})
\]
lie in one pure live arity-two residual class $R_j$: replacing the complete
slot index in any accepting two-hole context preserves membership.  Hence
\[
  \frag_{h_{\rm slot}}(R_j)=1,
  \qquad
  \frag_{h_{\rm idx}}(R_j)=k.
\]
The rank-$r$ start operation is purely witnessed on
$(R_1,\ldots,R_r)$, so Theorem~\ref{thm:tfrag-product} gives
\[
  \boxed{
  \tfrag_{h_{\rm slot}}(\theta)=1,
  \qquad
  \tfrag_{h_{\rm idx}}(\theta)=k^r.
  }
\]

Let $\operatorname{mcd}_h(L)$ denote minimum characteristic-sample cost for
the standard fixed-observation MCFG learner, measured by
\[
  \|K\|_+:=\sum_{w\in K}\max(1,|w|).
\]
This is the native encoding convention of the MCFG companion construction;
the CFG part above keeps its companion convention
$\|K\|=\sum_{w\in K}(|w|+1)$.  No comparison below mixes the two numerical
cost scales directly.

\begin{lemma}[Rigidity under the index observer]
\label{lem:index-rigidity}
For every finite $K\subseteq X_{k,r}$,
\[
  L(\widehat G_{h_{\rm idx}}(K))=K.
\]
Consequently
\[
  \boxed{
  \operatorname{mcd}_{h_{\rm idx}}(X_{k,r})=2r\,k^r.
  }
\]
\end{lemma}

\begin{proof}
Every component of every observed tuple is a substring of a target word and
has length at most $2r$.  The truncation observer is injective on such strings,
so an equal-type semantic unary rule can only relate identical tuple values.
Induct on learner derivations.  Constant rules emit exactly their indexed
observed tuples.  In a composition witness the parent tuple is
exactly the fixed template applied to the concrete child tuples; by induction
the children cannot change value, so neither can the parent.  Thus every
observed tuple state derives only its own tuple.  A start rule can therefore
generate only the sampled word that indexes it.  Sample consistency gives the
reverse inclusion.

A characteristic sample must consequently equal the whole target language.
There are $k^r$ words, each of length $2r$.
\end{proof}

\begin{proposition}[Linear-size locking sample under the slot observer]
\label{prop:slot-upper}
There is a characteristic sample $C_{k,r}$ for $h_{\rm slot}$ with
\[
  |C_{k,r}|=1+r(k-1),
  \qquad
  \|C_{k,r}\|_+=2r\,[1+r(k-1)].
\]
Hence
\[
 \boxed{
  \operatorname{mcd}_{h_{\rm slot}}(X_{k,r})
  \le 2r\,[1+r(k-1)].
 }
\]
\end{proposition}

\begin{proof}
Let $w_0=w_{(1,\ldots,1)}$.  For every slot $j$ and $i\ge2$, let $w_{j,i}$
be the word obtained from $w_0$ by changing only the $j$th slot index from
$1$ to $i$, and put
\[
 C_{k,r}:=\{w_0\}\cup\{w_{j,i}:1\le j\le r,\ 2\le i\le k\}.
\]
For each $j,i$, the tuples $u_{j,1}$ and $u_{j,i}$ occur in the same concrete
arity-two sentence context in $w_0$ and $w_{j,i}$, respectively, and have the
same slot type.  The learner therefore has semantic unary paths between their
states.

The base word $w_0$ contains the rank-$r$ composition witness for the
start template with children $u_{1,1},\ldots,u_{r,1}$.  Starting from the
sampled parent $w_0$, apply that witness rule and then, independently in child
position $j$, move by semantic unary rules from $u_{j,1}$ to any
$u_{j,i_j}$ and use the corresponding observed constant rule.  This derives
$w_{\mathbf i}$ for every $\mathbf i\in[k]^r$.  Thus
$X_{k,r}\subseteq L(\widehat G(C_{k,r}))$.  Proposition~\ref{prop:mcfg-local-soundness}, together with
Lemma~\ref{lem:slot-safe}, gives the reverse inclusion.  Monotonicity in the
sample then makes $C_{k,r}$ characteristic.
\end{proof}

\section{Arithmetic information versus derivational realizability}
\label{sec:arith-realizability}

\subsection{An abstract additive locking principle}
\label{subsec:additive-locking}

The preceding proof uses only an additive local-increment mechanism.  We record
that mechanism separately because it can serve as a lower-bound template for
other positive reconstruction systems.  The statement is a conditional schema:
by a set-driven reconstruction system here we mean a map from finite positive
samples to hypotheses whose successful derivations are finite rooted trees built
from the start, unary, constant/lexical, and composition steps listed below.
No claim is made for systems lacking such an increment interpretation.

\begin{theorem}[Integral additive-locking schema]
\label{thm:additive-locking}
Let a set-driven positive reconstruction system act on a target set $L$, and
let $\Psi:L\to a+\Lambda$ map completed outputs to an affine lattice over a
free abelian group $\Lambda$.  For each finite nonempty sample $K\subseteq L$,
put
\[
 \Lambda_K^{\Psi}
 :=\langle\Psi(x)-\Psi(y):x,y\in K\rangle_{\mathbb Z}.
\]
Assume that its learner derivations admit a local increment interpretation with
the following properties:
\begin{enumerate}[label=\textup{(\roman*)},leftmargin=*]
\item every start derivation is anchored at some sample element;
\item every sample-induced unary replacement has a context-independent
increment in $\Lambda_K^{\Psi}$;
\item constant or lexical steps contribute zero increment; and
\item every composition step has total increment equal to the sum of its child
increments.
\end{enumerate}
Then every generated completed output $z$ satisfies
\[
 \boxed{
   \Psi(z)\in \Psi(x_0)+\Lambda_K^{\Psi}
 }
\]
for some $x_0\in K$.  Consequently, if $K$ is characteristic, then the subgroup
generated by its differences contains the full target difference lattice
\[
 \Lambda_L^{\Psi}:=
 \langle\Psi(x)-\Psi(y):x,y\in L\rangle_{\mathbb Z}.
\]
In particular
\[
 |K|\ge 1+\operatorname{rank}\Lambda_L^{\Psi}.
\]
\end{theorem}

\begin{proof}
Induct on the learner derivation.  Conditions (ii)--(iv) imply that the local
increment accumulated below any state lies in $\Lambda_K^{\Psi}$: unary steps
add a generator-group element, constants add zero, and composition adds child
increments.  By (i), a completed derivation starts from a sample anchor
$x_0$, giving the affine-lattice containment.  If $K$ is characteristic, every
target output is generated, so every target difference belongs to
$\Lambda_K^{\Psi}$.  A subgroup generated by $|K|-1$ anchor differences has
rank at most $|K|-1$, giving the cardinality bound.
\end{proof}

\begin{remark}
This theorem is only a lower-bound principle.  In the family below it is
specialized to the index-pattern map $\Phi$; completeness still depends on
positive witnesses and on whether the required local decompositions are
actually schedulable.
\end{remark}

\subsection{Affine lower bound and exact coarse-observer cost}

The preceding sample is not merely a convenient upper bound.  The standard
learner cannot escape the affine span of the index patterns already present in
the positive sample.  This gives an exact lower bound even though the learner
may use arbitrary observed tuples and arbitrary composition
witnesses of fan-out at most two.

Let $V_{k,r}$ be the real vector space with basis
$e_{j,i}$, $1\le j\le r$, $1\le i\le k$, and encode an index vector by
\[
  \Phi(i_1,\ldots,i_r):=\sum_{j=1}^r e_{j,i_j}.
\]
All such points lie in the affine subspace determined by
$\sum_i x_{j,i}=1$ for each $j$, and their affine hull has dimension
\[
  r(k-1).
\]
For nonempty $K\subseteq X_{k,r}$ put
\[
 U_K:=\operatorname{span}\{\Phi(\mathbf i)-\Phi(\mathbf j):
       w_{\mathbf i},w_{\mathbf j}\in K\},
\]
and let
$\operatorname{Aff}(K)=\Phi(\mathbf i_0)+U_K$ for any
$w_{\mathbf i_0}\in K$.

\begin{lemma}[Context-independent slot difference]
\label{lem:slot-difference}
Let $\vec x,\vec y$ be tuples of the same arity with equal
$h_{\rm slot}$-type and suppose they share an accepting sentence context.
There is a vector $\Delta(\vec x,\vec y)\in V_{k,r}$ such that, for every
sentence context $E$ accepting both tuples,
\[
  \Phi(\operatorname{ind}(E[\vec y]))
  -\Phi(\operatorname{ind}(E[\vec x]))
  =\Delta(\vec x,\vec y),
\]
where $\operatorname{ind}(w_{\mathbf i})=\mathbf i$.
Moreover these differences are additive whenever the terms are defined.
If the standard learner has a semantic unary rule between observed tuples
$\vec x$ and $\vec y$ on sample $K$, then
\[
  \Delta(\vec x,\vec y)\in U_K.
\]
\end{lemma}

\begin{proof}
Equal slot type fixes the exact skeleton intervals occupied by the tuple
components.  Consider one paired slot $(A_j,B_j)$.  If one of its two terminals
is fixed by the surrounding context, then shared acceptance forces the terminal
inside the tuple to carry the unique matching index; hence that index cannot
change between $\vec x$ and $\vec y$.  Thus an index can change only when both
terminals of the paired slot lie inside the tuple holes.  In that case the
change contributes exactly
$e_{j,i'}-e_{j,i}$, independently of the accepting outer context.  Summing over
the changed slots defines $\Delta(\vec x,\vec y)$ and proves context
independence.  Additivity is coordinatewise.

A semantic unary rule is witnessed by a concrete context $E_0$ with
$E_0[\vec x],E_0[\vec y]\in K$.  Hence its difference is
\[
 \Delta(\vec x,\vec y)
 =\Phi(\operatorname{ind}(E_0[\vec y]))
  -\Phi(\operatorname{ind}(E_0[\vec x]))\in U_K.
\]
\end{proof}

\begin{theorem}[Affine-span invariant for the slot learner]
\label{thm:slot-affine-invariant}
For every nonempty finite $K\subseteq X_{k,r}$,
\[
 \boxed{
   w_{\mathbf i}\in L(\widehat G_{h_{\rm slot}}(K))
   \quad\Longrightarrow\quad
   \Phi(\mathbf i)\in\operatorname{Aff}(K).
 }
\]
\end{theorem}

\begin{proof}
We use the stronger derivation invariant that, whenever an observed tuple state
with concrete value $\vec x$ derives a tuple $\vec y$, replacing $\vec x$ by
$\vec y$ in any accepting sentence context changes the full-word index vector
by an element of $U_K$.

For a constant rule the change is zero.  For a semantic unary step
$[\vec x]\to[\vec z]$, Lemma~\ref{lem:slot-difference} puts the observed
endpoint difference $\Delta(\vec x,\vec z)$ in $U_K$; the induction hypothesis
for the continuation from $\vec z$ adds another vector in $U_K$.

For a composition witness
\[
 [\vec z]\to\rho([\vec x_1],\ldots,[\vec x_s]),
\]
the concrete witness satisfies
$\vec z=\rho(\vec x_1,\ldots,\vec x_s)$.  Replace the children one at a time by
their derived tuples.  Linearity and nondeletion make the child occurrences
disjoint.  By the slot-rigidity argument in
Lemma~\ref{lem:slot-difference}, a nonzero child replacement can change a slot
index only when both terminals of that slot lie inside that child occurrence;
therefore different child changes contribute additively.  Each contribution
lies in $U_K$ by induction, so their sum does too.  Rank-zero witnesses are the
constant case.

Finally a start rule begins at some sampled word $w_{\mathbf i_0}\in K$.
Applying the invariant to a complete derivation of $w_{\mathbf i}$ gives
\[
  \Phi(\mathbf i)-\Phi(\mathbf i_0)\in U_K,
\]
which is exactly the asserted affine containment.
\end{proof}

\begin{theorem}[Exact slot-observer characteristic data]
\label{thm:slot-exact}
For the standard fixed-observation MCFG learner,
\[
 \boxed{
  \operatorname{mcd}_{h_{\rm slot}}(X_{k,r})
  =2r\,[1+r(k-1)].
 }
\]
Equivalently, every characteristic sample has at least
$1+r(k-1)$ words, and the sample $C_{k,r}$ of
Proposition~\ref{prop:slot-upper} attains this bound.
\end{theorem}

\begin{proof}
Let $K$ be characteristic.  Then already
$L(\widehat G_{h_{\rm slot}}(K))=X_{k,r}$, so
Theorem~\ref{thm:slot-affine-invariant} forces every point
$\Phi(\mathbf i)$, $\mathbf i\in[k]^r$, to lie in
$\operatorname{Aff}(K)$.  The affine hull of all these points has dimension
$r(k-1)$, whereas the affine hull of $|K|$ points has dimension at most
$|K|-1$.  Hence
\[
 |K|\ge1+r(k-1).
\]
Every target word has length $2r$, so
$\|K\|_+\ge2r[1+r(k-1)]$.  Proposition~\ref{prop:slot-upper} gives equality.
\end{proof}

\begin{corollary}[Exact exponential rank sensitivity]
\label{cor:exponential-rank-sensitivity}
Along the observer refinement
$h_{\rm slot}\preceq h_{\rm idx}$,
\[
 \boxed{
 \frac{\operatorname{mcd}_{h_{\rm idx}}(X_{k,r})}
      {\operatorname{mcd}_{h_{\rm slot}}(X_{k,r})}
 =
 \frac{k^r}{1+r(k-1)}.
 }
\]
For every fixed $k\ge2$ this ratio grows exponentially in the rule rank $r$.
\end{corollary}

\begin{proof}
Combine Lemma~\ref{lem:index-rigidity} with
Theorem~\ref{thm:slot-exact}.
\end{proof}

\begin{remark}[Product-of-simplices geometry]
The vectors $\Phi([k]^r)$ are the vertices of the product of simplices
$(\Delta_{k-1})^r$.  The sample $C_{k,r}$ is an affine basis: one base vertex
and, for each of the $r$ coordinates, the $k-1$ elementary deviations from
that base.  Thus the coarse-observer locking problem in this family has exact
cost ``affine dimension plus one''.  This is a different finite geometry from
the CFG rectangular interfaces above, where minimum exposure is governed by
graphic-matroid bases (spanning trees).
\end{remark}

\begin{remark}[What this example proves, and what it does not]
The languages $X_{k,r}$ are finite and therefore regular.  Their purpose is to
isolate a genuinely fan-out-two reconstruction mechanism: the same pure
rank-$r$ residual operation has one typed cell under the slot observer and
$k^r$ typed cells under the finer index observer, while the standard MCFG
learner has exact characteristic-data costs
\[
 2r[1+r(k-1)]
 \quad\text{and}\quad
 2rk^r,
\]
respectively.  The exponential gap is therefore exact, not merely an upper/lower comparison.  This is not a non-context-free language separation.  An infinite non-CFL lift preserving
the same cell geometry is a separate construction problem.
\end{remark}

\subsection{Integral strengthening: affine lattices and unimodularity}
\label{subsec:lattice-locking}

The affine-span invariant admits a strictly stronger integral form.  The
learner does not average sample differences: semantic unary steps contribute
integer differences of observed index patterns, and linear nondeleting
composition adds such contributions.  Consequently the relevant invariant is
an affine lattice rather than only a real affine subspace.

Let
\[
 \Lambda_{k,r}
 :=
 \left\{
   \sum_{j=1}^r\sum_{i=1}^k n_{j,i}e_{j,i}
   : n_{j,i}\in\mathbb Z,\
     \sum_{i=1}^k n_{j,i}=0\text{ for every }j
 \right\}.
\]
This is the difference lattice generated by the vertices
$\Phi([k]^r)$; equivalently,
\[
 \Lambda_{k,r}\cong A_{k-1}^{\oplus r}.
\]
For nonempty $K\subseteq X_{k,r}$ define
\[
 \Lambda_K
 :=
 \left\langle
   \Phi(\mathbf i)-\Phi(\mathbf j)
   :w_{\mathbf i},w_{\mathbf j}\in K
 \right\rangle_{\mathbb Z}
 \subseteq \Lambda_{k,r}.
\]

\begin{theorem}[Integral lattice-span invariant]
\label{thm:slot-lattice-invariant}
For every nonempty finite $K\subseteq X_{k,r}$ and every sampled base word
$w_{\mathbf i_0}\in K$,
\[
 \boxed{
 w_{\mathbf i}\in L(\widehat G_{h_{\rm slot}}(K))
 \quad\Longrightarrow\quad
 \Phi(\mathbf i)-\Phi(\mathbf i_0)\in\Lambda_K.
 }
\]
In particular, Theorem~\ref{thm:slot-affine-invariant} follows after tensoring
with $\mathbb R$.
\end{theorem}

\begin{proof}
Repeat the derivation invariant from
Theorem~\ref{thm:slot-affine-invariant}, replacing the real vector space
$U_K$ by the integer subgroup $\Lambda_K$.
For a semantic unary rule between observed tuples, the final assertion of
Lemma~\ref{lem:slot-difference} identifies its increment with
\[
 \Phi(\operatorname{ind}(E_0[\vec y]))
 -\Phi(\operatorname{ind}(E_0[\vec x])),
\]
where both completed words belong to $K$.  Hence the increment is one of the
integer generators of $\Lambda_K$ (up to sign).  Constant rules contribute
zero.  At a composition witness, linearity and nondeletion
make the child contributions add without coefficients other than $1$; thus a
sum of elements of $\Lambda_K$ again lies in $\Lambda_K$.  A start rule anchors the derivation at some sampled word
$w_{\mathbf j_0}\in K$, so the induction first gives
$\Phi(\mathbf i)-\Phi(\mathbf j_0)\in\Lambda_K$.  For the arbitrarily
chosen base word $w_{\mathbf i_0}\in K$ in the theorem statement,
$\Phi(\mathbf j_0)-\Phi(\mathbf i_0)\in\Lambda_K$ by definition.
Adding the two differences yields the displayed containment for every sampled
base word.
\end{proof}

\begin{definition}[Lattice defect]
For nonempty $K\subseteq X_{k,r}$ define
\[
 \delta_{k,r}(K)
 :=
 \begin{cases}
  [\Lambda_{k,r}:\Lambda_K],
    &\text{if }\Lambda_K\text{ has full rank }r(k-1),\\
  \infty,&\text{otherwise.}
 \end{cases}
\]
\end{definition}

\begin{corollary}[Arithmetic obstruction to locking]
\label{cor:lattice-defect}
If $K$ is characteristic for the slot observer, then
\[
 \boxed{\delta_{k,r}(K)=1.}
\]
Thus full real affine span is necessary but not sufficient even for passing
this arithmetic obstruction.
\end{corollary}

\begin{proof}
If $K$ is characteristic, every $w_{\mathbf i}\in X_{k,r}$ is generated.
Theorem~\ref{thm:slot-lattice-invariant} therefore places every vertex
difference in $\Lambda_K$.  These differences generate
$\Lambda_{k,r}$, so $\Lambda_K=\Lambda_{k,r}$.
\end{proof}

A concrete computational form uses the standard root-lattice basis
\[
 f_{j,i}:=e_{j,i}-e_{j,1},
 \qquad 1\le j\le r,\quad 2\le i\le k.
\]
Put $d=r(k-1)$.  If a sample has the minimum possible cardinality
$|K|=d+1$, choose one anchor $w_{\mathbf i_0}\in K$ and write the remaining
$d$ differences in this basis.  Let $M_K\in\mathbb Z^{d\times d}$ be the
resulting integer matrix.

\begin{corollary}[Unimodularity test for minimum-size samples]
\label{cor:unimodular-test}
If $|K|=1+r(k-1)$ and $K$ is characteristic under $h_{\rm slot}$, then
\[
 \boxed{|\det M_K|=1.}
\]
More generally, for an arbitrary sample the Smith normal form of its difference
matrix computes the finite index $\delta_{k,r}(K)$ whenever the rank is full.
\end{corollary}

\begin{proof}
For a full-rank square difference matrix, the subgroup generated by its columns
has index $|\det M_K|$ in the ambient lattice written in the basis
$\{f_{j,i}\}$.  Apply Corollary~\ref{cor:lattice-defect}.  The Smith-normal-form
statement is the standard structure theorem for finitely generated subgroups
of $\mathbb Z^d$.
\end{proof}

\begin{example}[An affine basis that cannot lock]
\label{ex:index-two-affine-basis}
Take $k=2$ and $r=3$, abbreviating index $1$ by $0$ and index $2$ by $1$.
Consider the four sample patterns
\[
 000,\qquad 110,\qquad 101,\qquad 011.
\]
They are affinely independent over $\mathbb R$: relative to $000$, their
three difference vectors in the root-lattice basis are
\[
 (1,1,0),\qquad(1,0,1),\qquad(0,1,1),
\]
whose determinant has absolute value $2$.  Hence these four words form a real
affine basis of $(\Delta_1)^3$, but their difference lattice has index two.
Corollary~\ref{cor:unimodular-test} therefore proves that this sample is not
characteristic.  For instance, every generated pattern remains in the even
parity coset of the sample-generated lattice, so a pattern such as $100$ cannot
be generated.
\end{example}

\begin{remark}[Affine matroid versus arithmetic data]
The real-rank lower bound sees only the affine matroid of the vertex set
$\Phi([k]^r)$.  The lattice invariant retains additional arithmetic
information: two minimum-size samples may both be affine bases while only one
passes the index-one condition.  The canonical sample $C_{k,r}$ from
Proposition~\ref{prop:slot-upper} is unimodular, since its $d$ nonbase
differences are exactly the standard root-lattice basis vectors.  We do not
claim that unimodularity alone is sufficient for an arbitrary minimum-size
sample to be characteristic; it is an exact necessary arithmetic test and a
strict strengthening of the affine-span obstruction.
\end{remark}

\subsection{A hierarchical obstruction beyond unimodularity}
\label{subsec:hierarchical-envelope}

The lattice obstruction is not the whole story.  The learner can add
sample-derived increments only along an actual derivation tree: a semantic
unary move is attached to one observed tuple occurrence, and after a
composition rule is chosen, the mutable child occurrences are disjoint and
strictly smaller than the parent occurrence.  Thus an arbitrary integer
combination of sample differences need not be schedulable inside one learner
derivation.  We make this distinction explicit by a finite over-approximation
of the fan-out-two learner.

Write the terminal positions of a word of $X_{k,r}$ in the surface order
\[
  A_1<\cdots<A_r<B_1<\cdots<B_r,
\]
where $A_j$ denotes the position occupied by $a_{j,i_j}$ and $B_j$ the
position occupied by $b_{j,i_j}$.  A \emph{two-interval mask} is a nonempty
set of terminal positions that is the union of at most two intervals in this
surface order.  Every concrete tuple occurrence of fan-out at most two has
such a mask, after forgetting only the division into named tuple components.

For a target word $w$ and a mask $M$, write $w|_M$ for the filling of the
positions of $M$ in surface order.  For a sample $K$, let
$\Obs_K(M):=\{w|_M:w\in K\}$.  On $\Obs_K(M)$ define the elementary mobility
relation
\[
 \alpha\leftrightarrow_{K,M}\beta
\]
when there are $u,v\in K$ with
$u|_M=\alpha$, $v|_M=\beta$, and $u|_{M^c}=v|_{M^c}$; let
$\sim_{K,M}$ be its equivalence closure.  This deliberately forgets tuple
component boundaries, so it can only enlarge the mobility available to the
actual learner.

Let $\Comp(M)$ be the set of fillings of $M$ that occur as restrictions of
some word of $X_{k,r}$.  Recursively on $|M|$, define the \emph{hierarchical
two-interval envelope} $\Env_K(M,\alpha)\subseteq\Comp(M)$ as follows.  A
filling $\gamma$ lies in $\Env_K(M,\alpha)$ if there are
$\beta\sim_{K,M}\alpha$, pairwise disjoint proper two-interval submasks
$M_1,\ldots,M_s\subsetneq M$ (possibly $s=0$), and fillings
\[
  \gamma_t\in\Env_K(M_t,\beta|_{M_t})
\]
such that $\gamma$ is obtained from $\beta$ by replacing
$\beta|_{M_t}$ by $\gamma_t$ for all $t$ and leaving every other position of
$M$ fixed.  Only target-compatible resulting fillings are retained.  The
recursion is well founded because every mutable child mask is proper.  A
rank-one identity witness whose child mask equals its parent changes nothing
and may be deleted before this normalization.

\begin{theorem}[Hierarchical envelope bound]
\label{thm:hierarchical-envelope}
Let $K\subseteq X_{k,r}$ be nonempty and let $\widehat G_{h_{\rm slot}}(K)$
be the standard fan-out-two learner hypothesis.  Suppose an observed tuple
state has a concrete occurrence in a sample word $u$ with two-interval mask
$M$ and filling $\alpha=u|_M$.  Every target-compatible filling induced by a
tuple derivable from that state belongs to
$\Env_K(M,\alpha)$.  Consequently, with $P$ the full $2r$-position mask,
\[
  \boxed{
  L(\widehat G_{h_{\rm slot}}(K))\cap X_{k,r}
  \subseteq
  \{w:\ w|_P\in\Env_K(P,u|_P)\}
  }
\]
for every $u\in K$.
\end{theorem}

\begin{proof}
Group, at every derived state, the initial chain of semantic unary rules before
the first constant or composition rule, and induct on the remaining derivation
tree.  Suppose the observed state is represented by a tuple occurrence with
mask $M$.  If a derived tuple is target-compatible, soundness gives it the
same componentwise $h_{\rm slot}$-type as the observed tuple.  On factors of a
target word the slot observer records the exact rigid skeleton word, so the
derived tuple occupies the same skeleton positions and therefore determines a
filling of the same mask $M$.

A constant rule leaves that filling unchanged.  For a semantic unary rule,
the two endpoint tuples occur in one concrete sentence context and have the
same slot type.  Hence their target-word realizations agree outside $M$;
after forgetting component boundaries their fillings are related by
$\leftrightarrow_{K,M}$.  Thus every initial unary chain stays inside one
$\sim_{K,M}$-class.

Now consider the first composition witness after such a unary chain, and let
$\beta$ be the resulting observed parent filling.  The witness segmentation
places every child occurrence inside the parent occurrence, and distinct
children cover disjoint terminal positions.  After forgetting component
partitions, child $t$ therefore determines a two-interval submask $M_t\subseteq
M$.  If $M_t=M$, then no terminal position or other nonempty child remains;
the good-template conditions force a rank-one identity witness, which is a
self-loop and may be deleted.  Every nontrivial child mask is therefore proper.
By induction, the tuple derived at child $t$ induces a filling in
$\Env_K(M_t,\beta|_{M_t})$.  The composition template leaves all terminal
positions outside the child masks equal to the observed parent filling, so the
result is exactly one of the recursive replacements allowed in the definition
of $\Env_K(M,\alpha)$.  Retaining only $\Comp(M)$ is harmless: any derived
tuple considered here is target-compatible by hypothesis (and, in a successful
complete derivation, also by learner soundness).

Finally, start rules are anchored at sampled full words.  For the full mask
$P$ its complement is empty, so all sampled full words lie in one
$\sim_{K,P}$-class.  Hence the same full-mask envelope may be based at any
$u\in K$, which proves the displayed inclusion for every sampled base word.
\end{proof}

The envelope retains derivation-tree locality but is still an over-approximation:
it forgets named component boundaries and the exact good-template
partition.  Hence failure of the envelope to contain a target word is a valid
non-locking certificate for the real learner.

\begin{lemma}[Half-slot rigidity in the envelope]
\label{lem:half-slot-rigidity}
Fix a slot $j$.  If a mask $M$ contains exactly one of the two positions
$A_j,B_j$, then every filling in $\Env_K(M,\alpha)$ has at that included
position the same index as $\alpha$.
\end{lemma}

\begin{proof}
Induct on $|M|$.  Two sample words that agree outside $M$ cannot differ at the
included terminal of slot $j$: target words change the two terminals of a slot
together, while the other terminal lies outside $M$.  Hence mobility does not
change that index.  In a composition step, any child mask is still missing the
other terminal of slot $j$, so the induction hypothesis prevents a change
inside every child; fixed positions are unchanged as well.
\end{proof}

We now show that the index-one lattice condition can nevertheless fail to
lock.  Put $k=2$, $r=4$, and write the two indices as $0,1$.  Consider
\[
 K_{\diamond}
 :=\{0000,0001,0010,0111,1100\}
 \subseteq[2]^4,
\]
where, as before, a four-bit pattern abbreviates the corresponding word of
$X_{2,4}$.  Let
\[
 s_0=0000,\quad s_1=0001,\quad s_2=0010,
 \quad s_3=0111,\quad s_4=1100.
\]
The only sample row with first index $1$ is $s_4$, and the only rows with
second index $1$ are $s_3,s_4$.

\begin{lemma}[The $01$--$11$ incomplete-slot geometry]
\label{lem:diamond-geometry}
Let $P$ be the full eight-position mask for $X_{2,4}$.
\begin{enumerate}[label=\textup{(\roman*)},leftmargin=*]
\item If a two-interval mask $M$ contains both positions of slot $1$ but not
both positions of slot $2$, then for $q\in\{0,1,2\}$ every filling in
\[
  \Env_{K_{\diamond}}(M,s_q|_M)
\]
has first index $0$.
\item Under the same mask hypothesis, if
\(\Env_{K_{\diamond}}(M,s_3|_M)\) contains a filling with first index $1$, then
\[
   M=P\setminus\{A_2\}
   \quad\text{or}\quad
   M=P\setminus\{B_2\}.
\]
\item If a two-interval mask $M$ contains both positions of slot $2$ but not
both positions of slot $1$, then every filling in
\[
  \Env_{K_{\diamond}}(M,s_4|_M)
\]
has second index $1$.
\end{enumerate}
\end{lemma}

\begin{proof}
We prove the three assertions simultaneously by induction on $|M|$.  First
record the corresponding top-level mobility facts.

For a mobility path to reach first index $1$, its final edge must enter the
unique first-$1$ sample row $s_4$.  Under the hypotheses of \textup{(i)} and
\textup{(ii)}, at least one of $A_2,B_2$ lies outside $M$, so every row in one
mobility component has the same second index.  The only possible predecessor
of $s_4$ is therefore $s_3$.  Since $s_3$ and $s_4$ differ exactly in slots
$1,3,4$, a mask supporting that edge must contain all six terminal positions
of those three slots.  In the surface order
\[
 A_1,A_2,A_3,A_4,B_1,B_2,B_3,B_4,
\]
a union of at most two intervals with those six required positions and with an
incomplete slot $2$ is necessarily
\[
 P\setminus\{A_2\}\quad\text{or}\quad P\setminus\{B_2\}.
\]
For either seven-position mask, the sole outside terminal has second index
$1$ on $s_3,s_4$ and second index $0$ on $s_0,s_1,s_2$.  Thus a mobility
component rooted at a $00$ row cannot enter the $s_3$--$s_4$ component, proving
the top-level part of \textup{(i)}; a component rooted at $s_3$ can reach first
index $1$ only on one of the displayed masks, proving the top-level part of
\textup{(ii)}.

For \textup{(iii)}, at least one of $A_1,B_1$ lies outside $M$.  The row $s_4$
is the unique sample row with first index $1$, so its mobility component is a
singleton.  In particular top-level mobility cannot lower its second index.

Now consider recursive child replacements.  In \textup{(i)}, if a
 target-compatible filling were to change the complete first-slot index, then
Lemma~\ref{lem:half-slot-rigidity} forces some responsible child mask $N$ to
contain both $A_1,B_1$.  Since $N\subsetneq M$, it still has incomplete slot
$2$.  The top mobility stage leaves the root in one of the $00$ sample rows,
so the induction hypothesis for \textup{(i)} applied to $N$ forbids the
change.

For \textup{(ii)}, if the first-index change were first produced strictly
below $M$, half-slot rigidity again gives a proper child $N$ containing both
positions of slot $1$ and still missing a complete slot $2$.  Before entering
that child, top mobility has not changed the first index; because the outside
half of slot $2$ fixes second index $1$, the child is rooted in the projection
of the unique $01$ sample row $s_3$.  The induction hypothesis for
\textup{(ii)} would therefore force
\(N=P\setminus\{A_2\}\) or \(N=P\setminus\{B_2\}\).  Such a seven-position
mask has no proper supermask with incomplete slot $2$, contradicting
$N\subsetneq M$.  Hence any first-index change must occur at top mobility, and
the two displayed masks are necessary.

Finally, in \textup{(iii)}, any target-compatible lowering of the complete
second slot must occur inside one child containing both $A_2,B_2$; splitting
the two halves cannot change the index by half-slot rigidity.  That child is a
proper submask and still has incomplete slot $1$.  Top mobility from the
$s_4$ root is trivial, so the induction hypothesis for \textup{(iii)} forbids
the lowering inside the child.  This completes the simultaneous induction.
\end{proof}

\begin{lemma}[Masked Horn preservation for $K_{\diamond}$]
\label{lem:diamond-masked-horn}
Let $M$ be a two-interval mask containing both positions of slots $1$ and $2$.
For every sample row $s_q$ and every
\(
\gamma\in\Env_{K_{\diamond}}(M,s_q|_M)
\), the two complete slot indices in $\gamma$ satisfy
\[
  i_1(\gamma)\le i_2(\gamma).
\]
\end{lemma}

\begin{proof}
Induct on $|M|$.  In one recursive envelope step, the filling after top
mobility is the projection of a sample row, so its first two indices are one of
$00,01,11$.

Suppose first that this parent pair is $00$.  If the final filling had first
index $1$, target compatibility and half-slot rigidity force one child $N$ to
contain the complete first slot.  If $N$ has incomplete slot $2$,
Lemma~\ref{lem:diamond-geometry}\textup{(i)} forbids that change.  If $N$
contains both complete slots, the induction hypothesis on the proper child
forbids the pair $10$ inside $N$; because $N$ already contains both positions
of slot $2$, no disjoint child can subsequently restore second index $0$.
Thus a $00$ parent cannot produce $10$.

The $11$ case is dual.  To obtain $10$, some child must lower the complete
second slot.  If that child has incomplete slot $1$,
Lemma~\ref{lem:diamond-geometry}\textup{(iii)} forbids the lowering; if it
contains both complete slots, the induction hypothesis forbids $10$ inside the
child, and no disjoint child can alter the already covered first slot.

Finally suppose the parent pair is $01$.  Producing $10$ requires raising
slot $1$ and lowering slot $2$.  If one child contains both complete slots,
the induction hypothesis rules this out.  Otherwise two distinct disjoint
children must perform the two changes.  The child raising slot $1$ has
incomplete slot $2$ and is rooted in the projection of the unique $01$ row
$s_3$; by Lemma~\ref{lem:diamond-geometry}\textup{(ii)} its mask must be
$P\setminus\{A_2\}$ or $P\setminus\{B_2\}$.  Either seven-position mask
intersects every mask containing the complete second slot, so it cannot be
disjoint from the child that lowers slot $2$.  Hence $10$ is impossible here
as well.
\end{proof}

\begin{proposition}[Unimodularity is not sufficient]
\label{prop:unimodular-not-sufficient}
The sample $K_{\diamond}$ has the minimum possible cardinality
$1+r(k-1)=5$ and has lattice defect one, but it is not characteristic for
$X_{2,4}$ under $h_{\rm slot}$.
\end{proposition}

\begin{proof}
Relative to the base vertex $0000$, the four difference columns are
\[
 (0,0,0,1)^T,\quad
 (0,0,1,0)^T,\quad
 (0,1,1,1)^T,\quad
 (1,1,0,0)^T.
\]
Thus in the standard root-lattice basis the difference matrix is
\[
 M_{\diamond}=
 \begin{pmatrix}
 0&0&0&1\\
 0&0&1&1\\
 0&1&1&0\\
 1&0&1&0
 \end{pmatrix},
 \qquad
 \det M_{\diamond}=1.
\]
Hence $K_{\diamond}$ is a unimodular affine basis and
$\delta_{2,4}(K_{\diamond})=1$.

Apply Lemma~\ref{lem:diamond-masked-horn} to the full mask $P$.  Every filling
in every sampled full-mask envelope satisfies $i_1\le i_2$, so in particular
no such envelope contains the target pattern $1000$.  The hierarchical
envelope bound therefore gives
\[
  w_{1000}\notin L(\widehat G_{h_{\rm slot}}(K_{\diamond})).
\]
Thus $K_{\diamond}$ is not characteristic.
\end{proof}

\begin{remark}[Arithmetic span versus derivational realizability]
Corollary~\ref{cor:unimodular-test} is therefore an exact arithmetic
obstruction but not a characterization.  The sample differences of
$K_{\diamond}$ generate the full root lattice, yet the derivation tree cannot
schedule the corresponding integer combinations: unary replacements live at
specific tuple occurrences, while sibling replacements must occupy disjoint
suboccurrences and descendants are strictly nested.  The hierarchical envelope
records this missing scheduling information.  Thus minimum locking sets for
$X_{k,r}$ carry at least two layers of structure:
\[
 \boxed{
   \text{integral span}
   \quad+\quad
   \text{hierarchical occurrence realizability}.
 }
\]
For the binary-index family considered here, the second layer is first witnessed explicitly at $(k,r)=(2,4)$.
The envelope proof above establishes nonlocking at the native fan-out two.
Section~\ref{subsubsec:alternation-conflict} gives a stronger interpretation:
for $K_{\diamond}$ the mandatory upward support $\{1,3,4\}$ meets every
mandatory downward support $\{2,3\}$, $\{2,4\}$, and $\{2,3,4\}$.  Hence
this particular obstruction is laminar and order-independent; increasing
fan-out cannot make the two required sibling supports disjoint.  The mixed
rank-four shapes studied later exhibit the genuinely order-sensitive layer.
\end{remark}

\subsection{Graphic sufficiency and the sharp low-rank boundary}
\label{subsec:exchange-threshold-main}

For $K\subseteq X_{k,r}$ identify sample words with their index vectors and
let $H(K)$ be the graph joining two sample vectors when they differ in one
slot.  The changed slot labels the edge, and $\pi_j(K)$ denotes the set of
indices observed in slot $j$.

\begin{theorem}[Graphic sufficiency and sharp binary threshold]
\label{thm:graphic-sharp-package}
The following hold.
\begin{enumerate}[label=\textup{(\roman*)},leftmargin=*]
\item If $H(K)$ is connected and $\pi_j(K)=[k]$ for every slot $j$, then
$K$ is characteristic for $X_{k,r}$ under $h_{\rm slot}$.
\item If moreover $|K|=1+r(k-1)$, then the sample-difference lattice is the
full root lattice $A_{k-1}^{\oplus r}$; in particular $K$ is unimodular.
\item For the binary family $X_{2,r}$ at minimum size $|K|=r+1$,
unimodularity is sufficient for locking for $r=2,3$ but not for $r=4$.
Thus rank four is the first binary rank at which hierarchical occurrence
realizability yields an obstruction strictly beyond the lattice test.
\end{enumerate}
\end{theorem}

\begin{proof}[Proof sketch]
A one-slot sample edge exposes a semantic unary exchange on the corresponding
complete slot tuple.  Connectivity of $H(K)$ therefore makes every observed
index value movable independently in each slot, proving \textup{(i)}.
At minimum cardinality a spanning tree has exactly $k-1$ edges of each slot
label; their root differences form a basis of $A_{k-1}^{\oplus r}$, giving
\textup{(ii)}.  For \textup{(iii)}, the binary ranks two and three reduce to
three low-dimensional unimodular cube types, all of which lock, while the
rank-four sample $K_{\diamond}$ from
Proposition~\ref{prop:unimodular-not-sufficient} is unimodular and nonlocking.
The complete low-rank classification and the explicit rank-three higher-arity
derivation are given in Appendix~\ref{app:low-rank-threshold}.
\end{proof}

The theorem separates three notions already at low rank:
\[
 \text{single-slot exchange basis}
 \Longrightarrow
 \text{characteristic basis}
 \Longrightarrow
 \text{unimodular affine basis}.
\]
Within the binary-index family, the first implication is strict at rank three; the second becomes strict at
rank four.

\subsection{The exact middle layer: laminar occurrence-trace reachability}
\label{subsec:laminar-trace}

The preceding results leave a genuine gap between the graphic sufficient
criterion and the arithmetic/geometric necessary conditions.  We now isolate
the finite combinatorial object that occupies this gap.  It is not a matroid:
it is a laminar, template-labelled reachability system extracted from concrete
tuple occurrences.

Write $P_r$ for the ordered $2r$ terminal positions
$A_1<\cdots<A_r<B_1<\cdots<B_r$.  A \emph{fan-out-two occurrence trace}
$\tau$ is the finite positional data of a concrete tuple occurrence of arity
$d\le2$ in this skeleton: one or two nonempty component intervals, listed
left-to-right, with a nonempty separator between two components.  This is
exactly the occurrence convention used by the witness enumeration.
For an index vector $\mathbf i\in[k]^r$,
write
\[
  \operatorname{val}_{\mathbf i}(\tau)
\]
for the ordered tuple of terminal strings read on the trace and
$\operatorname{ctx}_{\mathbf i}(\tau)$ for the concrete sentence context obtained by replacing
those component intervals by their named holes.

For a finite sample $K\subseteq X_{k,r}$ define the \emph{laminar trace system}
$\mathfrak L(K)$ as follows.
\begin{enumerate}[label=\textup{(L\arabic*)},leftmargin=*]
\item Its states are all tuple values $\operatorname{val}_{\mathbf i}(\tau)$ observed in
      sample words $w_{\mathbf i}\in K$ on traces of arity at most two.
\item Every observed state has its rank-zero constant rule.
\item Two observed states of the same arity are joined by a semantic unary
      edge whenever their componentwise $h_{\rm slot}$-types agree and they
      admit occurrences with the same concrete sentence context in two sample
      words.
\item Fix a sample word $w_{\mathbf i}$ and a parent trace $\tau$.  Let
      $\tau_1,\ldots,\tau_s$ be nonempty child traces whose component
      intervals lie inside parent components and whose supports are pairwise
      disjoint.  Require, in addition, that their left-to-right component order
      and the intervening terminal gaps satisfy the nonpermuting and nonmerging
      conditions of a composition witness.  Reading the remaining parent
      positions as fixed terminals then determines the corresponding good
      template $\rho$.  Insert the labelled hyperedge
      \[
        \operatorname{val}_{\mathbf i}(\tau)
        \xrightarrow{\rho}
        \bigl(\operatorname{val}_{\mathbf i}(\tau_1),\ldots,
              \operatorname{val}_{\mathbf i}(\tau_s)\bigr).
      \]
\item Every sampled full word is a start state through its full arity-one
      trace.
\end{enumerate}
The child supports in (L4) are disjoint inside one parent support; recursively,
therefore, the supports occurring in a successful derivation are nested or
disjoint.  This is the laminar scheduling information that the lattice
invariant forgets and the hierarchical envelope deliberately relaxes.

Let $\operatorname{LReach}(K)\subseteq[k]^r$ be the set of index vectors whose full words are
generated by the finite hypergraph grammar $\mathfrak L(K)$.

\begin{theorem}[Exact laminar-trace characterization]
\label{thm:laminar-trace-exact}
For every nonempty $K\subseteq X_{k,r}$,
\[
  \boxed{
  \operatorname{LReach}(K)
  =
  \{\mathbf i\in[k]^r:
      w_{\mathbf i}\in L(\widehat G_{h_{\rm slot}}(K))\}.
  }
\]
Consequently
\[
  \boxed{
    K\text{ is characteristic for }X_{k,r}
    \iff
    \operatorname{LReach}(K)=[k]^r.
  }
\]
Thus minimum characteristic samples in the slot family are exactly the
minimum-cardinality laminar trace locking sets.
\end{theorem}

\begin{remark}[Alignment with the reconstruction operator]
The exactness statement uses the nonempty, left-to-right tuple occurrences and
the good composition witnesses of the fixed-observation MCFG learner recalled in Section~2.  In particular, (L4) is not the larger class of all
disjoint interval segmentations: the nonpermuting and nonmerging witness guards
are part of the trace rule.  This restriction is what makes the two inclusions
in Theorem~\ref{thm:laminar-trace-exact} literal rather than merely an
over-approximation.
\end{remark}

\begin{proof}
The two systems have the same start and constant rules.  Rule (L3) is exactly
the learner's semantic unary rule: the shared concrete sentence context is the
positive witness required by the learner, and equal componentwise slot type is
its typing guard.

For (L4), the selected child traces are concrete nonempty tuple
occurrences inside one observed parent occurrence.  Pairwise disjointness,
their left-to-right component order, and the explicit nonpermuting/nonmerging
admissibility requirement are exactly the data of a current composition
witness.  Hence the standard witness enumeration emits exactly the displayed
good rule.  This proves that every $\mathfrak L(K)$ derivation is a learner
derivation.

Conversely, every learner composition rule is emitted from a composition
witness.  Its parent and child occurrences supply exactly the trace data of
(L4): every component interval is nonempty, child components occur
left-to-right, different child supports are disjoint, and the witness itself
satisfies the nonpermuting/nonmerging conditions.  Every semantic unary learner
rule supplies (L3), and every constant or start rule supplies (L2) or (L5).
Reading a learner derivation tree node by node therefore yields a derivation in
$\mathfrak L(K)$ with the same concrete tuple at every node.  Hence the
generated full words coincide.  The characteristic-sample equivalence follows from sample monotonicity and
Proposition~\ref{prop:mcfg-local-soundness}.
\end{proof}

The exact system sits between the earlier certificates.  Put, for a fixed base
$\mathbf i_0\in K$,
\[
  \operatorname{Lat}(K)
  :=\{\mathbf i\in[k]^r:
       \Phi(\mathbf i)-\Phi(\mathbf i_0)\in\Lambda_K\}.
\]
The lattice invariant and hierarchical envelope theorem give
\[
  \boxed{
  \operatorname{LReach}(K)
  \subseteq
  \operatorname{Lat}(K)
  \cap
  \bigcup_{u\in K}
  \{\mathbf i:
      w_{\mathbf i}|_{P_r}\in\Env_K(P_r,u|_{P_r})\}.
  }
\]
On the sufficient side, Theorem~\ref{thm:exchange-locking} identifies a purely
graphic condition under which $\operatorname{LReach}(K)=[k]^r$.  Neither relaxation is the
exact middle layer: $K_3$ has disconnected single-slot exchange graph but full
laminar reachability, whereas $K_{\diamond}$ has full lattice span but proper
laminar reachability.  The trace system records precisely the missing
hierarchical placement information.

\section{Ordered scheduling and analytic locking certificates}
\label{sec:ordered-scheduling}

\subsection{Analytic critical-matching theorem}
\label{subsec:critical-matching-main}

The two order-sensitive rank-four simplex shapes have the same intrinsic
perfect matching of critical two-slot exchanges.  Relative to the physical
slot order, that matching is either separated, nested, or crossing.  This
single ordered datum explains the entire mixed behaviour.

\begin{theorem}[Analytic critical-matching classification]
\label{thm:critical-matching-package}
Let $K$ be either of the two mixed rank-four sample shapes and let $K^\pi$
be any physical ordering of its four slots.  Then:
\begin{enumerate}[label=\textup{(\roman*)},leftmargin=*]
\item $K^\pi$ is characteristic exactly when its critical matching is
noncrossing (separated or nested).
\item In a crossing ordering the reachable full words form exactly a proper
Horn half-cube.  For the two canonical representatives these are
\[
 \{x:x_0\le x_1\}
 \qquad\text{and}\qquad
 \{x:x_1\le x_2\},
\]
respectively.
\item The eight nonlocking orderings are precisely the stabilizer of the
crossing perfect matching, hence form a copy of the dihedral group $D_4$
inside $S_4$.
\end{enumerate}
\end{theorem}

\begin{proof}[Proof sketch]
A root unary change can modify only slots completed by its occurrence trace,
and a half-completed slot is rigid below that trace.  In crossing order the
mandatory support for the protected upward change and the mandatory support
for the protected downward change cannot be placed as disjoint fan-out-two
siblings; induction on the derivation therefore preserves the displayed Horn
implication.  Conversely, in separated order both critical pairs saturate as
disjoint children of one parent.  In nested order the inner pair saturates and
two exposed outer parent configurations cover the two remaining selector
slices.  These constructions yield the full cube.  The trace lemmas, exact
half-cube reverse inclusions, and group-theoretic stabilizer calculation are
collected in Appendix~\ref{app:critical-matching-details}.
\end{proof}

Thus the rank-four order effect is analytic rather than an artifact of the
finite census: crossing preserves a scheduling invariant, while noncrossing
order admits constructive positive saturation.

\subsection{Dual structural certificates beyond the rank-four census}
\label{subsec:dual-certificates}

The preceding mixed-shape proof contains two mechanisms that do not depend on
that particular five-sample census.  The nonlocking argument uses only a
\emph{support conflict}: an upward correction and a downward correction cannot
be placed in disjoint children.  The locking argument uses only a
\emph{Cartesian cover}: a finite set of sampled parent anchors, together with
independently saturable disjoint child blocks, covers the whole target cube.
We isolate these two certificates because they provide a reusable ordered
higher-rank interface between the exact laminar trace system and coarser graph
or lattice invariants.

For a positive fan-out-two trace $\tau$ on the ordered skeleton $P_r$, recall
that $C(\tau)$ is the set of slots completed by $\tau$.  For nonempty slot sets
$U,D\subseteq[r]$, write
\[
 U\perp_2 D
\]
if there do not exist two positive fan-out-two traces $\tau,\sigma$ with
pairwise disjoint supports such that
\[
 U\subseteq C(\tau),\qquad D\subseteq C(\sigma).
\]
Thus $\perp_2$ is a purely ordered occurrence-geometric incompatibility
relation.  It is stronger than ordinary set intersection: $U\cap D\ne\varnothing$
implies $U\perp_2 D$, but disjoint sets can also be incompatible because each
complete slot contributes one terminal in the $A$-half and one in the
$B$-half and a fan-out-two trace has only two interval components.

\begin{lemma}[Crossing pairs are trace-incompatible]
\label{lem:crossing-pairs-incompatible}
Let $a<b<c<d$ be four physical slots.  Then
\[
 \{a,c\}\perp_2\{b,d\}.
\]
The same statement holds with the two pairs interchanged.
\end{lemma}

\begin{proof}
This is Lemma~\ref{lem:crossing-traces-intersect} rewritten in support language.
A trace completing $a,c$ while avoiding the two complete slots $b,d$ would
need two separated components in the $A$-half and two more in the $B$-half,
which exceeds fan-out two.
\end{proof}

\subsubsection{Exact alternation characterization of support conflict}
\label{subsubsec:alternation-conflict}

The relation $\perp_2$ has an elementary exact description that removes trace
enumeration from the negative certificate.  For a slot set $U\subseteq[r]$ put
\[
 \widehat U:=\{A_j,B_j:j\in U\}\subseteq P_r.
\]
For two nonempty disjoint slot sets $U,D$, scan the slots in their physical
order $1<\cdots<r$, erase the slots outside $U\cup D$, and write the letter
$\mathsf U$ on a slot of $U$ and $\mathsf D$ on a slot of $D$.  Let
\[
 \operatorname{alt}(U,D)
\]
be the number of maximal constant blocks in this reduced two-letter word.  If
$U\cap D\ne\varnothing$, set $\operatorname{alt}(U,D):=\infty$.

For $f\ge1$, write $U\perp_fD$ if there do not exist disjoint supports
$S_U,S_D\subseteq P_r$, each a union of at most $f$ nonempty intervals, with
\[
 \widehat U\subseteq S_U,
 \qquad
 \widehat D\subseteq S_D.
\]
Thus the previously defined $\perp_2$ is exactly the fan-out-two instance.

\begin{theorem}[Alternation equals the exact fan-out threshold]
\label{thm:alternation-fanout-threshold}
For nonempty slot sets $U,D\subseteq[r]$,
\[
 \boxed{
 U\not\perp_fD
 \quad\Longleftrightarrow\quad
 U\cap D=\varnothing
 \ \text{and}\ 
 \operatorname{alt}(U,D)\le f.
 }
\]
Equivalently, for disjoint $U,D$ the least fan-out at which the two complete
slot sets can be scheduled as disjoint supports is exactly
\[
 \boxed{
 f_{\min}(U,D)=\operatorname{alt}(U,D).
 }
\]
Hence
\[
 \boxed{
 U\perp_fD
 \quad\Longleftrightarrow\quad
 \operatorname{alt}(U,D)>f,
 }
\]
with the convention $\operatorname{alt}=\infty$ on overlapping sets.
\end{theorem}

\begin{proof}
Assume first that $U\cap D=\varnothing$.  Let $c=c_1\cdots c_m$ be the reduced
slot-colour word over $\{\mathsf U,\mathsf D\}$ and put
$a:=\operatorname{alt}(U,D)$.  The required coloured terminal positions in the
full skeleton
\[
 A_1<\cdots<A_r<B_1<\cdots<B_r
\]
form the doubled colour word $cc$: the $A$-half sees the slots in the order
$c$, and the $B$-half sees the same order again.

Let $r_{\mathsf U}$ and $r_{\mathsf D}$ be the numbers of monochromatic
$\mathsf U$- and $\mathsf D$-runs in $cc$.  If the first and last letters of
$c$ differ, then $a$ is even and each colour has exactly $a$ runs in $cc$.  If
they agree, then $a$ is odd; the boundary between the two copies merges the
two end-runs of that colour, giving $a$ runs for the end colour and $a-1$ for
the other.  In either case
\[
 \max\{r_{\mathsf U},r_{\mathsf D}\}=a.                 \tag{$\ast$}
\]

Now suppose disjoint supports $S_U,S_D$ exist.  Since $S_D$ contains every
required $\mathsf D$-position, every interval component of $S_U$ must avoid
those positions.  It therefore cannot bridge two distinct
$\mathsf U$-runs of $cc$.  Hence $S_U$ needs at least
$r_{\mathsf U}$ interval components, and symmetrically $S_D$ needs at least
$r_{\mathsf D}$.  By ($\ast$), any common fan-out bound satisfies
$f\ge a$.

Conversely, for every monochromatic run in $cc$ take the interval from its
first required terminal position to its last.  Such an interval contains no
required terminal of the opposite colour; positions belonging to slots outside
$U\cup D$ may occur in its interior and are harmless.  The run intervals of
opposite colours are ordered and disjoint.  Taking the union of the
$\mathsf U$-run intervals for $S_U$ and the $\mathsf D$-run intervals for
$S_D$ gives disjoint supports with respectively
$r_{\mathsf U}$ and $r_{\mathsf D}$ components, both at most $a$.  Thus
fan-out $a$ suffices.

If $U\cap D\ne\varnothing$, some required terminal position belongs to both
$\widehat U$ and $\widehat D$, so disjoint supports are impossible at every
finite fan-out.  This is exactly the convention
$\operatorname{alt}(U,D)=\infty$.
\end{proof}

\begin{corollary}[The three four-slot matching thresholds]
\label{cor:matching-fanout-thresholds}
For physical slots $a<b<c<d$,
\[
\begin{aligned}
 f_{\min}(\{a,b\},\{c,d\})&=2,
 &&\text{(separated matching)},\\
 f_{\min}(\{a,d\},\{b,c\})&=3,
 &&\text{(nested matching)},\\
 f_{\min}(\{a,c\},\{b,d\})&=4,
 &&\text{(crossing matching)}.
\end{aligned}
\]
In particular, at fan-out two only the separated perfect matching can be
placed as two disjoint complete sibling supports.  The nested matching is
already trace-incompatible at fan-out two even though the mixed rank-four
sample can still lock in nested order by the different two-anchor construction
of Theorem~\ref{thm:noncrossing-constructive}.
\end{corollary}

\begin{proof}
The reduced slot-colour words are respectively
$\mathsf U\mathsf U\mathsf D\mathsf D$,
$\mathsf U\mathsf D\mathsf D\mathsf U$, and
$\mathsf U\mathsf D\mathsf U\mathsf D$, with two, three, and four constant
blocks.  Apply Theorem~\ref{thm:alternation-fanout-threshold}.
\end{proof}

\begin{corollary}[Alternation form of the Horn-conflict test]
\label{cor:alternation-horn-conflict}
In Theorem~\ref{thm:abstract-horn-conflict}, condition \textup{(C3)} is
equivalent to the purely ordered test
\[
 \boxed{
 \operatorname{alt}(U,D)>2
 \qquad
 \text{for every }U\in\mathcal U,
 D\in\mathcal D.
 }
\]
Thus the negative certificate can be checked from the ordered slot sets alone,
without enumerating tuple traces.  Each value
$\operatorname{alt}(U,D)$ is computable by one left-to-right scan of the slot
order after deleting neutral slots.
\end{corollary}

\begin{remark}[Fan-out as an alternation budget]
Theorem~\ref{thm:alternation-fanout-threshold} gives a direct quantitative
meaning to fan-out in the scheduling layer: fan-out $f$ can separate two
mandatory complete-slot supports exactly up to $f$ alternating blocks in the
one-half slot order.  Crossing chords are therefore only the first visually
obvious instance of a more general obstruction.  For example, three physical
slots with colours $\mathsf U\mathsf D\mathsf U$ already require fan-out three.
Conversely, increasing fan-out removes these support conflicts at the exact
alternation threshold.  This statement concerns occurrence scheduling, not
semantic substitutability; the Horn theorem needs the separate hypothesis that
the conflicting supports are mandatory for the protected up/down changes.
\end{remark}

\subsubsection{Multiway support scheduling and the positive block-cover number}
\label{subsubsec:multiway-support-scheduling}

The two-colour alternation theorem has a direct multiway extension.  This is
useful on the positive side, because one parent composition may expose several
independently variable children at once.

Let $B_1,\ldots,B_t\subseteq[r]$ be pairwise disjoint nonempty slot sets.  Scan
the slots in physical order, erase the neutral slots outside
$B_1\cup\cdots\cup B_t$, and colour every remaining slot by the unique label
$j$ with slot in $B_j$.  Write
\[
 c=c_1\cdots c_m\in[t]^m
\]
for the resulting reduced colour word.  For a colour $j$, let
$r_j(c)$ be the number of maximal $j$-runs in $c$.  Since the full terminal
skeleton contains an $A$-half and then a $B$-half in the same slot order, the
required coloured terminal positions have colour word $cc$.  Define
\[
 d_j(B_1,\ldots,B_t):=r_j(cc)
\]
and the \emph{multiway block-cover number}
\[
 \boxed{
 \operatorname{bc}(B_1,\ldots,B_t)
 :=\max_{1\le j\le t} d_j(B_1,\ldots,B_t).
 }
\]
Equivalently,
\[
 d_j(B_1,\ldots,B_t)
 =2r_j(c)-\mathbf{1}[c_1=c_m=j].                  \tag{$\dagger$}
\]

\begin{theorem}[Exact multiway support-scheduling theorem]
\label{thm:multiway-support-scheduling}
Let $B_1,\ldots,B_t$ be pairwise disjoint nonempty slot sets and let
$f_1,\ldots,f_t\ge1$.  There exist pairwise disjoint supports
$S_1,\ldots,S_t\subseteq P_r$ such that
\[
 \widehat B_j\subseteq S_j
 \qquad(1\le j\le t)
\]
and $S_j$ is a union of at most $f_j$ nonempty intervals if and only if
\[
 \boxed{
 d_j(B_1,\ldots,B_t)\le f_j
 \qquad(1\le j\le t).
 }
\]
Consequently the least common fan-out bound permitting all $t$ complete-slot
requirements to occur as pairwise disjoint sibling supports is exactly
\[
 \boxed{
 f_{\min}(B_1,\ldots,B_t)
 =\operatorname{bc}(B_1,\ldots,B_t).
 }
\]
\end{theorem}

\begin{proof}
The required terminals of colour $j$ occur in exactly
$d_j(B_1,\ldots,B_t)$ maximal monochromatic runs of the doubled colour word
$cc$.  Suppose pairwise disjoint supports $S_1,\ldots,S_t$ exist.  An interval
component of $S_j$ cannot bridge two distinct $j$-runs, because between them
lies a required terminal of some other colour $\ell\ne j$, and that terminal
belongs to $S_\ell$.  Disjointness would be violated.  Hence $S_j$ needs at
least $d_j$ interval components, proving necessity.

Conversely, for every monochromatic run of $cc$, take the interval from its
first required terminal position to its last.  Such an interval contains no
required terminal of another colour.  Intervals belonging to distinct runs are
therefore pairwise disjoint across colours.  For each colour $j$, let $S_j$ be
the union of its run intervals.  Then $S_j$ contains $\widehat B_j$ and has
exactly $d_j$ components.  Thus the coordinatewise bounds $f_j\ge d_j$ are
sufficient.  Taking a common cap gives the displayed maximum formula.
Formula~\textup{($\dagger$)} is the usual concatenation count: the two copies of
$c$ contribute $2r_j(c)$ runs, except that the last run of the first copy and
the first run of the second merge exactly when both boundary colours are $j$.
\end{proof}

\begin{corollary}[Alternation is the two-colour specialization]
\label{cor:alternation-is-two-colour}
For two disjoint nonempty slot sets $U,D$,
\[
 \operatorname{bc}(U,D)=\operatorname{alt}(U,D).
\]
Hence Theorem~\ref{thm:alternation-fanout-threshold} is precisely the
$t=2$ common-fan-out instance of
Theorem~\ref{thm:multiway-support-scheduling}.
\end{corollary}

\begin{proof}
The proof of Theorem~\ref{thm:alternation-fanout-threshold} already shows that
for a two-colour reduced word with $a$ constant blocks, the larger of the two
run counts in $cc$ is exactly $a$.
\end{proof}

\begin{corollary}[Fan-out-two no-return criterion]
\label{cor:fanout-two-no-return}
Assume $t\ge2$.  The complete slot sets $B_1,\ldots,B_t$ can be placed as
pairwise disjoint fan-out-two sibling supports if and only if every colour
occurs in a single maximal block of the reduced one-half word $c$.
Equivalently, after neutral slots are erased, no child label may disappear,
be followed by a different child label, and later reappear.
\end{corollary}

\begin{proof}
If some colour $j$ occurs in at least two runs of $c$, then
\textup{($\dagger$)} gives $d_j\ge3$, so fan-out two is impossible.  Conversely,
if every colour occurs in one run and $t\ge2$, the first and last colours of
$c$ are different.  Hence \textup{($\dagger$)} gives $d_j=2$ for every colour,
and Theorem~\ref{thm:multiway-support-scheduling} applies.
\end{proof}

\begin{corollary}[The rank-four matching vectors]
\label{cor:matching-scheduling-vectors}
For physical slots $a<b<c<d$, the two children of the three perfect matchings
have exact component requirements
\[
\begin{array}{c|c|c}
\text{matching type}&\text{reduced colour word}
 & (d_{\mathsf U},d_{\mathsf D})\\
\hline
\text{separated}&\mathsf U\mathsf U\mathsf D\mathsf D&(2,2)\\
\text{nested}&\mathsf U\mathsf D\mathsf D\mathsf U&(3,2)\\
\text{crossing}&\mathsf U\mathsf D\mathsf U\mathsf D&(4,4).
\end{array}
\]
Thus at fan-out two a single parent can expose both critical children only in
the separated order.  In particular, the nested locking proof must use its
two-anchor hierarchical construction rather than one simultaneous sibling
placement, while the crossing order is still more expensive because both
children individually require four components in any disjoint simultaneous
placement.
\end{corollary}

\begin{remark}[One scheduling invariant underlies both certificates]
The negative and positive scheduling primitives are therefore not unrelated
phenomena.  The Horn-conflict test asks whether two mandatory child
requirements have block-cover number greater than the available fan-out.  A
one-parent Cartesian construction asks the multiway version of the same
question for all intended children simultaneously.  What differs is the
logical use of the answer: incompatibility can certify preservation of a Horn
barrier, whereas compatibility only supplies the positional part of a positive
certificate.  To obtain locking one still needs the relevant observed child
states and sufficiently large fill sets.  Thus
$\operatorname{bc}$ is an exact occurrence-scheduling invariant, not by itself
a learnability criterion.
\end{remark}

For a binary sample $K\subseteq\{0,1\}^r$ and two protected coordinates
$p,q$, call a family $\mathcal U$ an \emph{up-support family} if every sample
pair $u,v$ with
\[
 u_q=v_q=1,\qquad u_p\ne v_p
\]
has difference support containing some $U\in\mathcal U$.  Dually,
$\mathcal D$ is a \emph{down-support family} if every sample pair with
\[
 u_p=v_p=0,\qquad u_q\ne v_q
\]
has difference support containing some $D\in\mathcal D$.

\begin{theorem}[Abstract Horn-conflict certificate]
\label{thm:abstract-horn-conflict}
Let $K\subseteq\{0,1\}^r$ be used as a sample for $X_{2,r}$ under
$h_{\rm slot}$, and fix protected coordinates $p,q$.  Assume
\begin{enumerate}[label=\textup{(C\arabic*)},leftmargin=*]
\item every $u\in K$ satisfies $u_p\le u_q$;
\item $\mathcal U$ is an up-support family and $\mathcal D$ is a
      down-support family for $(p,q)$;
\item every $U\in\mathcal U$ and $D\in\mathcal D$ are trace-incompatible:
      $U\perp_2 D$.
\end{enumerate}
Then
\[
 \boxed{
 \operatorname{LReach}(K)\subseteq\{x\in\{0,1\}^r:x_p\le x_q\}.
 }
\]
In particular, if the target cube contains a word with $(x_p,x_q)=(1,0)$,
then $K$ is not characteristic.
\end{theorem}

\begin{proof}
Repeat the simultaneous induction of
Theorem~\ref{thm:crossing-scheduling-horn}.  Lemmas
\ref{lem:unary-diff-completed} and \ref{lem:exact-half-slot-rigidity} are
independent of the number of ambient slots.  They imply that a child which
raises $p$ from a protected root pair $01$ while not completing $q$ must
complete some $U\in\mathcal U$, and a distinct child which lowers $q$ while
not completing $p$ must complete some $D\in\mathcal D$.

The cases with root protected pair $00$ or $11$, and the case in which one
child completes both protected slots, are excluded exactly as in the earlier
Horn induction.  In the remaining $01\to10$ case two distinct children are
required.  Their supports would have to be disjoint by positive-support
composition, but (C3) forbids disjoint traces completing the mandatory sets
$U$ and $D$.  Hence $10$ cannot be derived.
\end{proof}

The theorem can be visualized by a finite \emph{support-conflict graph}.  Put
one vertex for every mandatory support in $\mathcal U\cup\mathcal D$ and join
$U\in\mathcal U$ to $D\in\mathcal D$ whenever $U\perp_2 D$.  A protected Horn
implication is certified whenever this bipartite conflict graph is complete.
The crossing matching of the rank-four mixed shapes is the canonical four-slot
perfect-matching example: its two mandatory supports are disjoint as coordinate sets but are
adjacent in the conflict graph by
Lemma~\ref{lem:crossing-pairs-incompatible}.

We now record the complementary positive certificate.  Let $B\subseteq[r]$
be a slot block whose exact complete-slot trace $\tau_B$ is legal at fan-out
two, and let
$s=\operatorname{val}_{u}(\tau_B)$ be an observed state.  Define
\[
 \operatorname{Fill}_K(s;B)
 \subseteq \{0,1\}^{B}
\]
to be the set of binary assignments on $B$ whose corresponding tuple value is
derivable from $s$ in the laminar trace grammar $\mathfrak L(K)$.  Thus a
single-slot switch gives a two-point fill set, while
Lemma~\ref{lem:adjacent-pair-saturation} gives a four-point fill set for the
appropriate adjacent pair states.

\begin{definition}[Cartesian cover certificate]
\label{def:cartesian-cover-certificate}
A \emph{Cartesian cover certificate} for $K\subseteq X_{2,r}$ consists of a
finite family of sampled full-word anchors $w_{u^{(a)}}\in K$ such that, for
each anchor $a$, one chooses pairwise support-disjoint legal child traces
\[
 \tau_{B_{a,1}},\ldots,\tau_{B_{a,t_a}}
\]
inside that sampled full occurrence and a composition witness
from the full anchor state to the corresponding observed child states
$s_{a,j}$.  Let
\[
 R_a=[r]\setminus\bigcup_{j=1}^{t_a}B_{a,j}
\]
be the coordinates left as explicit terminal material by the parent template,
and define the associated box
\[
 Q_a:=\Bigl\{x\in\{0,1\}^r:
      x|_{R_a}=u^{(a)}|_{R_a},\quad
      x|_{B_{a,j}}\in\operatorname{Fill}_K(s_{a,j};B_{a,j})
      \ \forall j\Bigr\}.
\]
The certificate is \emph{covering} if
\[
 \bigcup_a Q_a=\{0,1\}^r.
\]
\end{definition}

\begin{theorem}[Cartesian cover locking certificate]
\label{thm:cartesian-cover-locking}
If $K\subseteq X_{2,r}$ admits a covering Cartesian cover certificate, then
\[
 \boxed{\operatorname{LReach}(K)=\{0,1\}^r,}
\]
so $K$ is characteristic for the slot observer.
\end{theorem}

\begin{proof}
Fix a certificate anchor $a$.  Its composition witness is a learner rule
whose children are precisely the states $s_{a,j}$ and whose remaining terminal
material fixes the coordinates in $R_a$ to the anchor values.  Because the
child supports are pairwise disjoint, their derivations are independent below
that rule.  Hence every choice of one filling from each
$\operatorname{Fill}_K(s_{a,j};B_{a,j})$ produces the corresponding full word
in $Q_a$.  Therefore $Q_a\subseteq\operatorname{LReach}(K)$ for every anchor.
If the boxes cover the full cube, every target word is generated.  Soundness from
Proposition~\ref{prop:mcfg-local-soundness} gives the reverse inclusion.
\end{proof}

\begin{corollary}[The mixed rank-four theorem has dual certificates]
\label{cor:mixed-dual-certificates}
For each ordered copy of the mixed shapes $K_{10}$ and $K_{13}$:
\begin{enumerate}[label=\textup{(\roman*)},leftmargin=*]
\item in the crossing order, the mandatory up/down supports give a complete
      support-conflict graph, and
      Theorem~\ref{thm:abstract-horn-conflict} certifies the proper Horn
      half-cube;
\item in the separated order, one sampled full parent together with the two
      saturated adjacent critical pairs gives a single Cartesian box equal to
      the whole four-cube;
\item in the nested order, the two endpoints of the outer critical exchange
      give two Cartesian boxes of size eight, using the saturated inner pair
      and the directly switchable outer collateral coordinate, and the two
      boxes cover the four-cube.
\end{enumerate}
Thus the analytic crossing/noncrossing classification is an instance of two
finite certificate schemes rather than an isolated rank-four calculation.
\end{corollary}

\begin{proof}
Part (i) is the support-language reformulation of
Corollary~\ref{cor:mixed-crossing-horn-analytic}: intersecting mandatory
supports are automatically incompatible, and the remaining disjoint pair is
the crossing matching, hence incompatible by
Lemma~\ref{lem:crossing-pairs-incompatible}.  Part (ii) is the separated
construction in Theorem~\ref{thm:noncrossing-constructive}, with two saturated
pair fill sets of size four.  Part (iii) is its nested construction: each outer
anchor fixes the unswitched outer coordinate and independently varies the
inner saturated pair and the direct-switch outer coordinate, producing one
$8$-vertex box; the two opposite outer-anchor values cover all sixteen
vertices.
\end{proof}

\begin{remark}[What the certificates do and do not claim]
The Horn-conflict and Cartesian-cover conditions are sound certificates, not a
claimed complete decision procedure for arbitrary $X_{2,r}$.  A sample may
lock through a more entangled laminar derivation than any convenient Cartesian
cover, and failure may be witnessed by an invariant more complicated than one
binary Horn implication.  Their value is structural: the negative certificate
reduces nonlocking to an ordered support incompatibility, while the positive
certificate reduces locking to finitely many independent child fill sets.
Both are strictly richer than the single-slot exchange graph and retain the
scheduling information that the affine and lattice relaxations discard.
\end{remark}

\section{Hierarchical reuse and width--anchor tradeoffs}
\label{sec:hierarchy-tradeoffs}

The block-cover number measures the positional width needed by one simultaneous
sibling decomposition.  A derivation tree may do better by reusing different
exposed parent configurations at different branches.  To separate these
resources, let $\operatorname{LReach}_f(K)$ denote the occurrence-trace
reachability set when traces are allowed at most $f$ interval components.
For $f=2$ this is the exact learner geometry of
Theorem~\ref{thm:laminar-trace-exact}; larger $f$ is used only as a scheduling
relaxation.

A \emph{hierarchical Cartesian certificate} is a finite tree whose leaves are
previously established local fill sets, whose unary nodes route between
unary-equivalent observed states, and whose product nodes place pairwise
disjoint child blocks under one exposed parent.  The width of the certificate
is the maximum local fan-out used at any node.  Several sampled parent roots
may be joined by a start-cover union.  Relative to a fixed finite module
library, we record the pair
\[
  (\text{certificate width},\ \text{number of parent-root branches}).
\]

\begin{theorem}[Hierarchy, max-node width, and the mixed-shape Pareto frontiers]
\label{thm:hierarchy-package}
For the two mixed rank-four shapes the following statements hold.
\begin{enumerate}[label=\textup{(\roman*)},leftmargin=*]
\item The exact flat-versus-hierarchical fan-out thresholds for separated,
nested, and crossing orders are
\[
\begin{array}{c|c|c}
\text{order}& f_{\rm flat}& f_{\rm hier}\\ \hline
\text{separated}&2&2\\
\text{nested}&3&2\\
\text{crossing}&4&4.
\end{array}
\]
Thus nested order has a strict one-unit hierarchy dividend.
\item The width of a hierarchical Cartesian certificate is a
\emph{maximum-over-nodes} resource, not a sum over depth.  Every product node
pays at least the block-cover number of its child blocks.
\item For the natural library consisting of the exposed critical-pair
saturations and collateral singleton switches, the exact width--anchor Pareto
frontiers for the full four-cube are
\[
 \boxed{
 \begin{array}{c|c}
 \text{separated}&\{(2,1)\}\\
 \text{nested}&\{(2,2),(3,1)\}\\
 \text{crossing}&\{(4,1)\}.
 \end{array}}
\]
In particular, in nested order reducing local width from three to two forces
one additional exposed parent branch.
\end{enumerate}
\end{theorem}

\begin{proof}[Proof sketch]
Separated order realizes the two saturated critical pairs as width-two
siblings under one parent.  Nested order cannot place the two pair freedoms
flatly below width three, but at width two the saturated inner pair can be
combined with a switchable collateral singleton under each endpoint of the
outer critical exchange; the two resulting eight-point boxes cover the full
cube.  Crossing order is protected below width four by the fan-out-$f$
Horn-conflict certificate, Theorem~\ref{thm:fanout-f-horn-conflict}.  For certificate trees,
induction shows that composition depth does not add widths: only the largest
local product-node requirement matters.  Finally, at nested width two no
available critical module varies the remaining outer selector coordinate, so
a single root branch covers only one selector slice; two roots are necessary
and sufficient.  Full definitions and proofs are in
Appendix~\ref{app:hierarchy-details}.
\end{proof}

This theorem gives the second Pareto phenomenon of the paper.  The CFG family
trades finite observation against positive exposure; the MCFG mixed family
trades positional width inside one composition node against hierarchical reuse
of exposed parent configurations.

\section{Discussion and research interpretation}

The results isolate a sequence of finite resources that are usually bundled
inside one characteristic-sample argument:
\[
 \boxed{\begin{gathered}
 \text{observation}\to\text{fragmentation}\to\text{exposure}
 \to\text{arithmetic span}\\
 \to\text{ordered scheduling}\to\text{hierarchical reuse}
 \end{gathered}}
\]
No arrow is merely terminological.  The CFG family $R_{m,q}$ gives an exact
two-point Pareto frontier between observer size, internal-live fragmentation,
and characteristic-data cost.  Its exposure problem is rectangular: after the compulsory rigidity witnesses
are fixed, observer fibers induce bipartite connectivity requirements, and the
variable parts of minimum samples are spanning-tree bases.  Comparable observer refinement can only split fibers and
increase this rectangular cost, so the nontrivial tradeoff comes from
incomparable ways of allocating finite algebraic information.

Higher rank introduces genuinely different phenomena.  In $X_{k,r}$,
transition fragmentation multiplies across children, and a comparable finer
observer changes the characteristic-data cost from
$2r[1+r(k-1)]$ to $2rk^r$.  The coarse learner obeys an integral additive
locking invariant, but lattice generation is not enough: a minimum-size
unimodular sample first fails to lock at rank four within the binary-index family.  Exact laminar
trace reachability identifies the missing resource as derivational
realizability in a linear terminal order.

The ordered scheduling theorem then turns this geometric restriction into an
exact scalar cost.  For several intended siblings, the least common fan-out is
the largest monochromatic run count in the doubled reduced child-colour word;
in the two-colour case this is the alternation threshold.  Horn-conflict and
Cartesian-cover certificates are dual uses of the same run geometry.  The
rank-four critical matching makes the distinction visible analytically:
crossing order preserves a proper Horn half-cube, whereas separated and nested
orders lock constructively.

Finally, hierarchical reuse is a resource distinct from local fan-out.  A
certificate tree pays the maximum local block-cover cost rather than a sum over
depth.  At the critical-module assembly level, the natural module library gives
in nested order the exact exchange
\[
  (3,1)\quad\longleftrightarrow\quad(2,2)
\]
between local width and exposed parent-root branches, while the same library
gives the single Pareto points $(2,1)$ and $(4,1)$ in separated and crossing
orders.  Thus finite positive
reconstruction is governed not only by how much semantic information is
observed, but also by how that information fragments, where it is exposed, and
whether the exposed operations can be scheduled coherently in a derivation.

The appendices retain weighted refinements, abstract incidence formulations,
residual-core connections, the complete rank-four census, witness-hypergraph
machinery, and the detailed low-rank and scheduling proofs.  These support the
main theorem packages without changing the paper's central resource hierarchy.

\appendix

\section{Observer-fiber refinement penalties and weighted extensions}
\label{app:weighted-rectangular}

The rectangular theorem quantifies not only absolute exposure cost but also the
exact price of splitting an observer fiber.

\subsection{Additive-weighted rectangular exposure}

\begin{theorem}[Weighted rectangular exposure]
\label{thm:weighted-rectangular}
Give each possible exposure $(p,c)\in P\times C$ a nonnegative cost
$w(p,c)$.  The minimum total cost of an exposure $E$ satisfying
$\RectCl_{\mathcal C}(E)=P\times C$ is
\[
 \boxed{
 \sum_{t=1}^k
 \operatorname{MST}_w(K_{P,C_t}),
 }
\]
where $\operatorname{MST}_w$ is the minimum spanning-tree weight.

If the weights are additive,
\[
 w(p,c)=\alpha_p+\beta_c+\gamma,
 \qquad \alpha_p,\beta_c,\gamma\ge0,
\]
then the contribution of block $C_t$ is explicitly
\[
 \boxed{
 \sum_{p\in P}\alpha_p
 +\sum_{c\in C_t}\beta_c
 +( |C_t|-1)\alpha_{\min}
 +( |P|-1)\beta_{t,\min}
 +( |P|+|C_t|-1)\gamma,
 }
\]
where $\alpha_{\min}=\min_{p\in P}\alpha_p$ and
$\beta_{t,\min}=\min_{c\in C_t}\beta_c$.
\end{theorem}

\begin{proof}
The first assertion is Corollary~\ref{cor:rectangular-minimum} with edge
weights: the blocks are independent and a minimum sufficient exposure in each
block is exactly a minimum spanning tree.

For additive weights, let $T$ be any spanning tree of $K_{P,C_t}$.  Writing
$d_T(v)$ for degrees,
\[
 w(T)=
 \sum_{p\in P}d_T(p)\alpha_p
 +\sum_{c\in C_t}d_T(c)\beta_c
 +( |P|+|C_t|-1)\gamma.
\]
Every degree is at least one, while
\[
 \sum_{p\in P}(d_T(p)-1)=|C_t|-1,
 \qquad
 \sum_{c\in C_t}(d_T(c)-1)=|P|-1.
\]
This gives the stated lower bound.  It is attained by the double-star formed
from a row $p_*$ minimizing $\alpha$ and a column $c_*$ minimizing $\beta$:
use every edge $(p_*,c)$ and every edge $(p,c_*)$ with $p\ne p_*$.  Hence the
bound is exact.
\end{proof}

\begin{remark}[Relation to substitution graphs]
Clark--Eyraud style substitutability learning already organizes observed
substrings into connected components of a substitution graph.  The present
closure is a typed bipartite specialization: (R2) moves vertically along a
literal column and (R3) horizontally inside one observer-type block.  Thus the
classical difunctional closure explains why the resulting components complete
to rectangles; the observer partition determines which horizontal moves are
available.
\end{remark}

\begin{theorem}[Unweighted fiber-splitting penalty]
\label{thm:split-penalty}
Let $|P|=r$ and let $\mathcal C$ and $\mathcal D$ be partitions of $C$ with
$\mathcal D$ refining $\mathcal C$.  Write
\[
 e_P(\mathcal C):=|C|+(r-1)|\mathcal C|
\]
for the minimum rectangular exposure cardinality.  Then
\[
 \boxed{
 e_P(\mathcal D)-e_P(\mathcal C)
 =(r-1)(|\mathcal D|-|\mathcal C|).
 }
\]
In particular, splitting one type block into $s$ nonempty blocks costs exactly
$(s-1)(r-1)$ additional positive incidences.
\end{theorem}

\begin{proof}
Apply Corollary~\ref{cor:rectangular-minimum} to the two partitions and
subtract.  The column term $|C|$ cancels.
\end{proof}

\begin{corollary}[Local observer-refinement penalty]
\label{cor:local-refinement-penalty}
Let $h\preceq g$ and let $C$ be a finite set of representatives contained in
one pure semantic class.  The partition of $C$ into $g$-fibers refines its
partition into $h$-fibers.  On an $r$-row rectangularly isolated interface,
refining $h$ to $g$ changes the minimum unweighted exposure by exactly
\[
 \boxed{
 (r-1)\bigl(|g[C]|-|h[C]|\bigr).
 }
\]
Thus observer refinement has a quantitatively exact local data penalty whenever
the interface remains rectangularly isolated.
\end{corollary}

\begin{theorem}[Additive weighted splitting penalty]
\label{thm:weighted-split-penalty}
Let $B\subseteq C$ be one type block, let $|P|=r$, and suppose
\[
 w(p,c)=\alpha_p+\beta_c+\gamma,
 \qquad \alpha_p,\beta_c,\gamma\ge0.
\]
If $B$ is split into nonempty blocks $B_1,\ldots,B_s$, set
\[
 A:=\sum_{p\in P}\alpha_p,
 \qquad a:=\min_{p\in P}\alpha_p,
 \qquad b:=\min_{c\in B}\beta_c,
 \qquad b_j:=\min_{c\in B_j}\beta_c.
\]
Then the exact increase in minimum exposure cost caused by this split is
\[
 \boxed{
 (s-1)(A-a)
 +(r-1)\left(\sum_{j=1}^s b_j-b+(s-1)\gamma\right).
 }
\]
In particular the penalty is always nonnegative.
\end{theorem}

\begin{proof}
Apply the additive closed form of Theorem~\ref{thm:weighted-rectangular} once
to $B$ and once to each $B_j$, then subtract.  The terms
$\sum_{c\in B}\beta_c$ cancel.  Since
$\sum_j(|B_j|-1)=|B|-s$, the row-minimum contribution changes by
$-(s-1)a$, while the repeated row base cost contributes $(s-1)A$.  The
column-minimum term changes by $(r-1)(\sum_j b_j-b)$, and the edge-count term
changes by $(s-1)(r-1)\gamma$.  Combining them gives the displayed formula.
Nonnegativity follows from $A\ge a$, $b_j\ge0$, one $b_j$ equaling $b$ for a
block containing a global minimizer, and $\gamma\ge0$.
\end{proof}

\begin{remark}[Where the Pareto tradeoff comes from]
Along an observer-refinement chain, local fiber partitions only split, so
Theorem~\ref{thm:split-penalty} and
Corollary~\ref{cor:local-refinement-penalty} force exposure cost upward.  The
$R_{m,q}$ tradeoff therefore cannot be a ``more information helps'' phenomenon
along one refinement chain.  Rather, the two Pareto-optimal observers are
algebraically incomparable: the larger global observer spends additional
states elsewhere so that its partition of the critical semantic class $F$ is
\emph{coarser}.  In this precise sense the tradeoff is a reallocation of finite
algebraic information between global separation requirements and local
semantic fragmentation.
\end{remark}

\section{Abstract transport--substitution interfaces}

The preceding rectangular closure can be separated completely from grammar
syntax.  This gives the common local mechanism behind the CFG reconstruction
rules and the semantic-unary layer of the MCFG learner.

\begin{definition}[Occurrence graph and transport--substitution closure]
Let $P,C$ be finite sets, let $\tau:C\to T$ be a finite typing, and let
$E\subseteq P\times C$ be an observed incidence relation.  The
\emph{occurrence graph} $\Lambda_\tau(E)$ has vertex set $E$ and joins two
vertices $(p,c)$ and $(p',c')$ whenever either
\[
 c=c' \qquad\text{(transport)}
\]
or
\[
 p=p'\quad\text{and}\quad \tau(c)=\tau(c')
 \qquad\text{(same-row type substitution)}.
\]
Define
\[
 \operatorname{TSCl}_\tau(E)
 :=
 \{(p,c'):\exists c,p'\ ((p,c),(p',c')\in E
 \text{ and }(p,c)\sim_{\Lambda_\tau(E)}(p',c'))\}.
\]
Thus the first coordinate of an observed root is kept fixed while the mutable
column may travel through transport and type-substitution moves.
\end{definition}

\begin{theorem}[Transport--substitution closure theorem]
\label{thm:ts-closure}
Let $\mathcal C_\tau$ be the partition of $C$ into the nonempty fibers of
$\tau$.  Then
\[
 \boxed{
 \operatorname{TSCl}_\tau(E)
 =
 \RectCl_{\mathcal C_\tau}(E).
 }
\]
Equivalently, inside each type fiber the occurrence graph is the line graph of
the corresponding bipartite incidence graph, and every one of its connected
components generates the complete rectangle on the rows and columns occurring
in that component.
\end{theorem}

\begin{proof}
Fix a type $t$ and write $E_t=E\cap(P\times\tau^{-1}(t))$.  Two vertices of
$E_t$ are adjacent in $\Lambda_\tau(E)$ exactly when the corresponding edges
of the bipartite graph $G_t(E)$ share a row or share a column.  Hence the
restriction of $\Lambda_\tau(E)$ to $E_t$ is the line graph of $G_t(E)$.
Nontrivial connected components of a graph and of its line graph have the same
edge support.  Therefore, if $\Gamma$ is a component of $G_t(E)$, any root row
$p\in P_\Gamma$ can start from one of its incident observed edges and the
mutable column can reach every edge, hence every column, in $C_\Gamma$.
This gives $P_\Gamma\times C_\Gamma$.  Conversely a transport or substitution
move never leaves the same component or the same type fiber.  Comparing with
Theorem~\ref{thm:rectangular-closure} gives the equality.
\end{proof}

\begin{definition}[Rectangularly isolated reconstruction interface]
A local positive-reconstruction interface is \emph{rectangularly isolated}
when its observed root configurations are indexed by an incidence relation
$E\subseteq P\times C$, its only column-changing unary moves are the transport
and same-row type-substitution moves of
the preceding definition, and after a root is opened its row output is
immutable while the mutable column state is read out independently.  No other
rule of the local interface may change the column readout.
\end{definition}

\begin{corollary}[Abstract rectangular reconstruction principle]
\label{cor:abstract-rectangle}
For every rectangularly isolated interface, the set of row--column outputs
obtainable from the observed roots is exactly
$\RectCl_{\mathcal C_\tau}(E)$.  Consequently full reconstruction of
$P\times C$ is equivalent to connected spanning incidence in every type
fiber, and all conclusions of Corollary~\ref{cor:rectangular-minimum} and
Theorem~\ref{thm:weighted-rectangular} apply without reference to CFG syntax.
\end{corollary}

\begin{remark}[CFG instance]
For the fixed-$h$ CFG constructor, take rows to be observed two-sided contexts
$(u,v)$, columns to be observed nonempty factors $x$, and
$\tau(x)=h(x)$.  Rule (R2) is transport along a fixed column and Rule (R3) is
same-row type substitution.  Thus the unary occurrence-state part of the CFG
constructor is literally the occurrence graph above.  The rigidity lemmas for
$R_{m,q}$ are exactly what verify rectangular isolation for the main-suffix
interface.
\end{remark}

\begin{theorem}[MCFG incidence-connectivity theorem]
\label{thm:mcfg-incidence}
Fix a finite positive sample $K$, an arity $d$, and a componentwise observation
type $\mu\in M^d$ in the standard fixed-observation MCFG learner.  Let
$C_\mu(K)$ be the observed arity-$d$ tuples $\mathbf x$ with
$h^{(d)}(\mathbf x)=\mu$, let $P_\mu(K)$ be the concrete arity-$d$ sentence
contexts occurring with at least one such tuple in $K$, and put
\[
 E_\mu(K)
 :=
 \{(E,\mathbf x):E[\mathbf x]\in K,\ \mathbf x\in C_\mu(K)\}.
\]
Then for observed tuples $\mathbf x,\mathbf y\in C_\mu(K)$,
\[
 \boxed{
 \begin{aligned}
 [\mathbf x]\Rightarrow_{\rm unary}^*[\mathbf y]
 \quad\Longleftrightarrow\quad
 &\mathbf x,\mathbf y\text{ lie in the same connected}\\
 &\text{component of }E_\mu(K).
 \end{aligned}
 }
\]
Here $\Rightarrow_{\rm unary}^*$ uses only the learner's semantic unary
identity-template rules.
\end{theorem}

\begin{proof}
A direct semantic unary rule $[\mathbf x]\to[\mathbf y]$ exists exactly when
$\mathbf x$ and $\mathbf y$ have the same type $\mu$ and some concrete context
$E$ is incident with both, i.e. when the two tuple vertices are joined by a
length-two path through a context vertex in the bipartite incidence graph.
Concatenating unary rules concatenates such paths.  Conversely any alternating
incidence path between tuple vertices yields, two edges at a time, a chain of
semantic unary rules.  Hence unary reachability is exactly connectedness.
\end{proof}

\begin{remark}[Scope of the MCFG corollary]
Theorem~\ref{thm:mcfg-incidence} is an exact statement about the semantic-unary
layer, not by itself an exact characteristic-sample theorem for arbitrary
MCFGs.  Positive-rank composition witnesses and the sample-dependent rank cap remain
additional structure.
For a rectangularly isolated arity-$d$ exposure interface, however, the same
spanning-tree exposure theorem applies verbatim.  Thus the rectangular theory
extends locally to every arity; what remains formalism-specific is proving
isolation from alternative composition routes.
\end{remark}

\section{Positive residual cores and residual sections}

For this section, a \emph{finite pure residual diagram}
$D=(Q,T,I)$ consists of finitely many live pure residual classes $Q$, a finite
set $T$ of witnessed residual transitions between them, and a set $I\subseteq Q$
of start classes.  It is \emph{$L$-complete} when the residual presentation
obtained from $(Q,T,I)$ generates exactly $L$.  For a finite positive sample
$K\subseteq L$, write
\[
  D\preceq H_K
\]
when $D$ has a yield-preserving realization inside the hypothesis
$H_K:=\B_h(K)$: start classes are start-reachable and every transition of $D$
is simulated by a derivation schema with the same concrete template, allowing
(R2)/(R3) unit paths around the local rule instance.  A finite sample
$S\subseteq L$ \emph{positively forces} $D$ when
\[
  S\subseteq K\subseteq L,\ K\text{ finite}
  \quad\Longrightarrow\quad D\preceq H_K.
\]
An $L$-complete finite pure residual diagram that is positively forceable is
called a \emph{positive residual core}.

\begin{lemma}[Diagram locking]
\label{lem:diagram-locking}
Assume the fixed-$h$ constructor is sound on positive subsets of $L$.  If
$D$ is $L$-complete and $D\preceq H_K$, then $L(H_K)=L$.
\end{lemma}
\begin{proof}
The realization of the complete residual diagram gives
$L\subseteq L(H_K)$, while learner soundness gives
$L(H_K)\subseteq L$.
\end{proof}

The positive-core formalism distinguishes two statements that should not be
conflated.  A forcing witness for an $L$-complete residual diagram is always
characteristic, but a characteristic sample need not automatically choose one
coherent representative of every canonical residual state.  We therefore
separate the unconditional lower bound from the additional section condition
needed for equality.

For a positive residual core $D$, write
\[
 e(D):=\min\{\|S\|:S\subseteq L\text{ positively forces }D\},
\]
and put
\[
 e_h^+(L):=\min\{e(D):D\text{ is an $L$-complete positive residual core}\}
\]
whenever such cores exist.

\begin{proposition}[Characteristic data versus positive-core exposure]
\label{prop:core-lower}
For every nonempty $h$-substitutable target for which the positive-core
formalism is defined,
\[
 \boxed{\cd_h(L)\le e_h^+(L).}
\]
\end{proposition}

\begin{proof}
A positive residual core is $L$-complete and positively forceable.  By
Lemma~\ref{lem:diagram-locking}, every forcing witness for such a diagram is a
characteristic sample for the same sound learner.  Taking minima gives the
inequality.
\end{proof}

\begin{definition}[Residual section at a sample]
\label{def:residual-section}
Let $C\subseteq L$ and $H_C:=\B_h(C)$.  An $L$-complete finite residual
diagram $D=(Q,T,I)$ has a \emph{residual section in $H_C$} if there is an
assignment
\[
 \sigma:Q\to\operatorname{State}(H_C)
\]
such that each $\sigma(R)$ has semantic residual label $R$, every start class
in $I$ is reachable from the hypothesis start object, and every transition of
$T$ is simulated between the chosen states by a derivation schema in $H_C$
with the same concrete template.  Unit (R2)/(R3) paths may occur inside the
schema; they are not required to appear as separate transitions of $D$.
\end{definition}

\begin{theorem}[Residual-section criterion]
\label{thm:section-criterion}
Let $L$ be nonempty and $h$-substitutable.  Let $C$ be a minimum
characteristic sample for $\B_h$.  If some $L$-complete finite residual
diagram $D$ has a residual section in $H_C=\B_h(C)$, then $D$ is positively
forceable by $C$ and
\[
 \boxed{e(D)=\cd_h(L)=\|C\|.}
\]
Consequently $e_h^+(L)=\cd_h(L)$ whenever a minimum characteristic sample
admits such an $L$-complete residual section.
\end{theorem}

\begin{proof}
Because the constructor is syntactically monotone, $C\subseteq K$ implies
$H_C\subseteq H_K:=\B_h(K)$.  Hence every chosen state, every unit path, and
every transition-simulation schema witnessing the residual section in $H_C$
remains present in $H_K$.  Thus
\[
 C\subseteq K\subseteq L\quad\Longrightarrow\quad D\preceq H_K,
\]
so $C$ positively forces $D$ and $e(D)\le\|C\|$.  Since $D$ is complete,
Proposition~\ref{prop:core-lower} gives
$\cd_h(L)\le e(D)$.  Minimality of $C$ gives $\|C\|=\cd_h(L)$, so all three
quantities are equal.
\end{proof}

\begin{remark}[Why a section hypothesis is genuinely needed]
The semantic map from hypothesis states to canonical residual classes need not
have a coherent section for an arbitrary characteristic sample: distinct local
rule occurrences may require different hypothesis representatives of the same
residual class.  This is exactly the coherence issue that motivates diagram
realization rather than fact-by-fact forceability.  Thus
\(e_h^+(L)=\cd_h(L)\) is not asserted here without the residual-section
hypothesis.
\end{remark}

\begin{theorem}[Exact positive-core exposure for the rigid-wrapper family]
\label{thm:family-core-equality}
For every safe observer $h$ of $R_{m,q}$, a minimum characteristic sample
constructed from all rigidity words and one spanning tree in each graph
$G_t$ admits an $R_{m,q}$-complete residual section.  Consequently
\[
 \boxed{
 e_h^+(R_{m,q})
 =\cd_h(R_{m,q})
 =4\bigl(qm+m^2+\kappa_h(q-1)\bigr).
 }
\]
\end{theorem}

\begin{proof}
Let $C_h$ contain all $qm$ rigidity words and, for every $h$-fiber $C_t$, the
main words corresponding to the edges of a spanning tree of $G_t$.  By
Lemma~\ref{lem:graph-characterization} and
Proposition~\ref{prop:self-locking}, $C_h$ is a minimum characteristic sample.

Map every hypothesis state $[x:u,v]$ to the canonical $(R_{m,q},h)$ residual
class determined by $x$.  This semantic label is well defined for terminal
yields of that state: the fixed-$h$ soundness invariant gives the same
$h$-type and one shared accepting context with $x$, so substitutability gives
equality of distributions.  Rules (R2) preserve the literal factor and Rules
(R3) connect equal-$h$ factors observed in one context; hence both kinds of
unit rule stay inside a single canonical residual class.  Rules (R1) and (R4)
therefore supply the genuine concatenation and terminal transition templates
in the residual image.

Choose one hypothesis representative for every residual state used below.
All rigidity factors have canonical sampled occurrences.  For each main
suffix fiber $C_t$, choose a root edge of its spanning tree and use its suffix
state as the representative of the corresponding refined residual state.
Tree connectivity supplies an (R2)/(R3) path from this representative to every
sampled occurrence of every suffix in $C_t$.  Hence whenever a binary
concatenation transition is needed, the representative can first move by unit
rules to the local sampled occurrence, apply (R1), and then move each child by
unit rules to its chosen representative.  The same construction at the
start/full-word level uses the sampled root words.  Taking the finite union of
the residual transitions needed by these schemas for all words of
$R_{m,q}$ yields a finite residual diagram $D_h$ that is realized by the
chosen representatives in $H_{C_h}$ and derives every target word.  Intrinsic
residual soundness gives the reverse inclusion, so $D_h$ is
$R_{m,q}$-complete.

Thus $D_h$ has a residual section in $H_{C_h}$, and
Theorem~\ref{thm:section-criterion} gives
$e(D_h)=\cd_h(R_{m,q})$.  Proposition~\ref{prop:core-lower} then shows that
no positive residual core has smaller forcing cost.  The numerical formula is
Theorem~\ref{thm:cd}.
\end{proof}

\section{Unary mobility quotients and witness hypergraphs}
\label{sec:witness-hypergraph}

The preceding rectangular theory isolates the mobility generated by two local
moves: transport of the same observed object between contexts and substitution
between same-type objects in one context.  For the MCFG learner a second layer
is present.  Composition witnesses generate genuine
higher-rank rules.  It is therefore tempting to replace the bipartite exposure
graph by a single hypergraph spanning-tree problem.  This is not quite the
correct abstraction.  The unary mobility remains graph-theoretic; higher-rank
witnesses form a second, directed and template-labelled hypergraph layer on top
of the unary connected components.

\subsection{The unary mobility quotient}

Fix a finite positive sample $K$ and a fixed observation $h$.  Let
$\mathsf V_K$ denote the finite set of observed tuple-value states of the
current learner.  For $v=[\vec x]$ put
\[
  \operatorname{ar}(v)=|\vec x|,
  \qquad
  \operatorname{tp}(v)=h^{(\operatorname{ar}(v))}(\vec x).
\]
Write
\[
  v\leftrightarrow_K w
\]
when the learner contains semantic unary rules in both directions between
$v$ and $w$.  Concretely, this requires the same arity, equal componentwise $h$-type,
and a shared observed sentence context.  The
condition is symmetric, so $\leftrightarrow_K$ may be regarded as an undirected
edge relation.  Let
\[
  \mathsf Q_K:=\mathsf V_K/{\leftrightarrow_K^*}
\]
be the set of unary mobility components, and write $[v]_U$ for the component
of $v$.

\begin{lemma}[Representative mobility]
\label{lem:representative-mobility}
If $v,w\in\mathsf V_K$ lie in the same unary mobility component, then the
learner derives $v\Rightarrow^* w$ and $w\Rightarrow^* v$ using semantic unary
rules only.
\end{lemma}
\begin{proof}
A path in the undirected graph $\leftrightarrow_K$ is, by definition, a
sequence of pairs for which both directed unary rules are present.  Reading the
path in either direction gives the two derivations.
\end{proof}

Thus the choice of a concrete representative inside one component is
irrelevant for local generator simulation: a chosen state can be moved to any
other representative before firing a local witness rule and moved back after
that rule at each child.

\subsection{The quotient witness multihypergraph}

Every composition witness in $K$ gives a learner rule
\[
  z\longrightarrow
  \rho(x_1,\ldots,x_r),
\]
where $z,x_1,\ldots,x_r\in\mathsf V_K$ are the observed parent and child
states and $\rho$ is the induced good template.  Rank zero is treated as a labelled constant edge.

\begin{definition}[Unary-quotient witness multihypergraph]
\label{def:quotient-witness-hypergraph}
The \emph{unary-quotient witness multihypergraph}
$\overline{\mathcal H}_h(K)$ has vertex set $\mathsf Q_K$.  For every
positive-rank observed witness rule above it contains the ordered directed
hyperedge
\[
  [z]_U
  \xrightarrow{\ \rho\ }
  ([x_1]_U,\ldots,[x_r]_U).
\]
For every observed rank-zero constant $z\to\vec c$ it contains a rank-zero
edge $[z]_U\xrightarrow{\vec c}()$.  A component is marked start-reachable
when it contains a state reached by a learner start rule.
\end{definition}

The object is a multihypergraph rather than an ordinary hypergraph: child order,
child arities, and the concrete good-template label $\rho$ are part of the
edge data.  Equivalently, it is the finite generator-level shadow of the
learner's derivation multicategory after contracting semantic-unary mobility.

\begin{definition}[Local quotient-hypergraph realization]
\label{def:local-hypergraph-realization}
Let $P=(V,I,R)$ be a finite typed tuple presentation.  A \emph{local
quotient-hypergraph realization} of $P$ in
$\overline{\mathcal H}_h(K)$ is a sort-preserving map
\[
  \phi:V\to\mathsf Q_K
\]
such that:
\begin{enumerate}[label=(\roman*)]
\item for every $X\in I$, $\phi(X)$ is start-reachable;
\item for every rank-zero generator $X\to\vec c$, there is a rank-zero edge
      $\phi(X)\xrightarrow{\vec c}()$;
\item for every positive-rank generator
\[
  R:X\to\rho(Y_1,\ldots,Y_r),
\]
there is a quotient witness edge
\[
  \phi(X)
  \xrightarrow{\ \rho\ }
  (\phi(Y_1),\ldots,\phi(Y_r)).
\]
\end{enumerate}
\end{definition}

\begin{theorem}[Quotient-hypergraph local simulation]
\label{thm:quotient-hypergraph-simulation}
If $P$ admits a local quotient-hypergraph realization in
$\overline{\mathcal H}_h(K)$, then the learner hypothesis on $K$ contains a
coherent yield-preserving generator simulation of $P$.  More precisely, one
may choose a single concrete representative $s_X\in\phi(X)$ for every target
state $X$, and use those same representatives simultaneously for all target
generators.
\end{theorem}
\begin{proof}
Choose one observed state $s_X$ in each component $\phi(X)$.  If $X$ is a
start state, start-reachability gives a learner start state $t_X$ in the same
component; Lemma~\ref{lem:representative-mobility} connects $t_X$ to $s_X$.

For a positive-rank generator
$R:X\to\rho(Y_1,\ldots,Y_r)$, choose one witnessing hyperedge and concrete
representatives
\[
  z\to\rho(x_1,\ldots,x_r)
\]
that induce it.  By construction, $z$ lies in $\phi(X)$ and each $x_j$ lies in
$\phi(Y_j)$.  Unary mobility therefore gives
\[
  s_X\Rightarrow^* z,
  \qquad
  x_j\Rightarrow^* s_{Y_j}
  \quad(1\le j\le r).
\]
Insert the observed witness rule between these unary paths.  This yields the
derivation schema
\[
  s_X
  \Rightarrow^*
  z
  \Rightarrow
  \rho(x_1,\ldots,x_r)
  \Rightarrow^*
  \rho(s_{Y_1},\ldots,s_{Y_r}),
\]
where the last step means applying the unary paths independently inside the
child nonterminals.  The concrete template is unchanged, so yield is
preserved.  The rank-zero case is identical with a constant edge in place of
the witness rule.  The same chosen $s_X$ is used for every occurrence of $X$,
so the simulation is coherent across all generators.
\end{proof}

\begin{corollary}[Hypergraph certificate for exact reconstruction]
\label{cor:hypergraph-exact}
Assume $L(P)=L$, $K\subseteq L$, and the learner is sound on positive subsets
of $L$.  If $P$ locally realizes in $\overline{\mathcal H}_h(K)$, then
\[
  L(\widehat G(K))=L.
\]
\end{corollary}
\begin{proof}
Theorem~\ref{thm:quotient-hypergraph-simulation} gives the generator
simulation and hence $L\subseteq L(\widehat G(K))$ by induction on target
derivations.  Soundness gives the reverse inclusion.
\end{proof}

\begin{remark}[Exact scope of the theorem]
No converse is used beyond the corresponding \emph{local normal form}
of generator simulation: unary mobility, followed by one observed
constant/composition witness, followed by unary mobility at the children.  In
that restricted form, any such simulation determines the corresponding
quotient witness edge.  A
general hypothesis derivation may simulate one target generator by a composite
schema involving several witness rules.  Such composite simulations are
legitimate but need not correspond to a single edge of
$\overline{\mathcal H}_h(K)$.  Thus the theorem is an exact combinatorial
certificate for the local simulation pattern used by the CFG/MCFG
completeness proofs, not a claim that every possible simulation factors through
one witness edge.
\end{remark}

\subsection{Why higher rank gives two layers rather than a hypergraph tree}

The quotient-hypergraph theorem changes the interpretation of the next
research step.  Higher-rank MCFG exposure does not simply replace a graphic
spanning tree by a ``hypergraph spanning tree''.  Two logically different
finite tasks remain:
\[
\boxed{
\begin{array}{c}
\text{unary mobility / state coherence}
\quad\leadsto\quad
\text{graph connectivity},\\[1mm]
\text{local rule realization}
\quad\leadsto\quad
\text{template-labelled directed hyperedge coverage}.
\end{array}}
\]
The first layer is exactly the shared-context connectivity already present in
the semantic unary rules.  The second layer is the genuinely MCFG-specific
composition witness structure.  Rank changes the second layer, not the
first.

\begin{definition}[Connected witness-cover problem]
\label{def:connected-witness-cover}
Fix a finite target presentation $P$ and a finite family $\mathcal S$ of
candidate positive sample words.  Each sample word contributes a finite set of
unary-incidence edges and a finite set of constant/composition witness
hyperedges.  The \emph{connected witness-cover} problem asks for a minimum-cost
subfamily $C\subseteq\mathcal S$ such that the resulting unary-quotient witness
multihypergraph admits a local realization of $P$.
\end{definition}

This optimization problem contains the rectangular exposure problem as the
special case in which the only nontrivial obligation is unary mobility inside
one family of typed columns.  In the general MCFG case the witness obligations
cannot be discarded: a semantic unary component may be fully connected while a
required template $\rho$ is not witnessed at all.

\begin{proposition}[Separated-layer decomposition]
\label{prop:separated-layer-decomposition}
Suppose the candidate sample family decomposes as a disjoint union
$\mathcal S=\mathcal S_U\sqcup\mathcal S_W$ with the following properties.
Samples in $\mathcal S_U$ contribute only unary-mobility edges, samples in
$\mathcal S_W$ contribute only witness hyperedges, and the required witness
hyperedges are determined solely by the target-state unary components.  Then
minimum connected-witness-cover cost decomposes additively as
\[
  \operatorname{OPT}_{\rm CWC}
  =
  \operatorname{OPT}_{\rm mobility}
  +
  \operatorname{OPT}_{\rm witness}.
\]
If each mobility block is rectangular with additive edge weights, the first
term is the sum of the corresponding minimum-spanning-tree costs.
\end{proposition}
\begin{proof}
Under the disjoint-support hypotheses, feasibility is the conjunction of two
independent conditions on disjoint decision variables: the selected
$\mathcal S_U$-samples must create the required unary components, and the
selected $\mathcal S_W$-samples must supply the required quotient witness
hyperedges.  Costs add across the disjoint union, so every feasible solution
has cost at least the sum of the two independent optima, and the union of two
independent optimal solutions attains that lower bound.  The final statement
is the weighted rectangular-exposure theorem.
\end{proof}

\begin{remark}[Sample sharing is the genuinely coupled case]
In actual characteristic samples one positive word may simultaneously expose
an anchor, create unary incidence, and contain one or more positive-support
composition witnesses.  Then the decomposition in
Proposition~\ref{prop:separated-layer-decomposition} can be strict: the same
word can pay for both layers.  The correct optimization object is therefore the
connected witness-cover problem, not the sum of a graph problem and a
hypergraph problem in general.
\end{remark}

\subsection{Connection with positive residual cores}

The quotient construction also supplies the coherence that was missing from a
naive state-by-state residual quotient.  A map $\phi$ chooses one unary
component for every target state, while
Theorem~\ref{thm:quotient-hypergraph-simulation} shows that arbitrary concrete
representatives may be fixed once and for all and all witness endpoints can be
transported to those representatives.  Hence a quotient-hypergraph realization
is a finite combinatorial certificate for a coherent generator simulation, not
merely a collection of separately forceable facts.

This gives the following research interpretation:
\[
\boxed{
\text{positive reconstruction}
=
\text{mobility quotient}
+
\text{witness multihypergraph}
+
\text{soundness}.
}
\]
For CFG rectangular interfaces the witness layer collapses to the local binary
factorizations already available from observed factors, leaving a pure graphic
exposure problem.  For MCFGs the composition witness layer survives and is
exactly where rank and sample-bounded branching enter.  The transition-fragmentation product law now makes the rank
effect explicit: child-state fiber counts multiply across one witnessed
operation, and the fan-out-two family $X_{k,r}$ realizes an exponential
characteristic-data refinement penalty.

\section{Detailed low-rank exchange classification}
\label{app:low-rank-threshold}
This appendix contains the full proofs and low-dimensional cube classification
summarized in Theorem~\ref{thm:graphic-sharp-package}.

The preceding obstruction shows that lattice generation does not by itself
encode derivation-tree scheduling.  There is nevertheless a broad class of
minimum locking samples for which scheduling is completely transparent: the
sample itself supplies single-slot exchange paths.  This also lets us locate
sharply, for the binary-index family, the first rank at which unimodularity can
fail to imply locking.

For $K\subseteq X_{k,r}$ identify a sample word with its index vector in
$[k]^r$.  Let $H(K)$ be the \emph{single-slot exchange graph}: its vertices are
the elements of $K$, and two vertices are adjacent when their index vectors
have Hamming distance one.  Label such an edge by its unique changed slot.
For slot $j$, let $\pi_j(K)\subseteq[k]$ be the set of indices appearing in
that slot among the sample words.

\begin{theorem}[Single-slot exchange locking criterion]
\label{thm:exchange-locking}
Let $K\subseteq X_{k,r}$ be nonempty.  Suppose
\begin{enumerate}[label=\textup{(\roman*)},leftmargin=*]
\item $H(K)$ is connected; and
\item $\pi_j(K)=[k]$ for every slot $j$.
\end{enumerate}
Then $K$ is characteristic for $X_{k,r}$ under $h_{\rm slot}$.
\end{theorem}

\begin{proof}
Fix a slot $j$.  Form a graph $E_j(K)$ on the index set $[k]$ by joining
$i$ and $i'$ whenever $K$ contains two sample words that differ only in slot
$j$, with values $i$ and $i'$ there.  Because $H(K)$ is connected, a path in
$H(K)$ between any two sample vertices projects in coordinate $j$ to a walk
whose nontrivial steps are precisely edges of $E_j(K)$.  By
$\pi_j(K)=[k]$, the graph $E_j(K)$ is therefore connected.

Write $u_{j,i}=(a_{j,i},b_{j,i})$.  Every edge $\{i,i'\}$ of $E_j(K)$ is
witnessed by two sample words having the same concrete arity-two sentence
context after the complete $j$th slot is removed.  The tuples $u_{j,i}$ and
$u_{j,i'}$ have the same componentwise $h_{\rm slot}$-type, so the learner
contains semantic unary rules connecting their states.  Connectivity of
$E_j(K)$ hence gives a unary path between $u_{j,i}$ and $u_{j,i'}$ for every
$i,i'\in[k]$.

Choose any base sample word $w_{(i_1,\ldots,i_r)}\in K$.  Its observed full
occurrence supplies the rank-$r$ composition witness for the start
template with children
$u_{1,i_1},\ldots,u_{r,i_r}$.  Starting from that witness, move independently
in child position $j$ along the unary component just described to an arbitrary
$u_{j,t_j}$, and finish with the observed constant rule for that tuple.  This
derives $w_{(t_1,\ldots,t_r)}$ for every $(t_1,\ldots,t_r)\in[k]^r$.
Proposition~\ref{prop:mcfg-local-soundness} gives the reverse inclusion, so the hypothesis is exact already on
$K$.  Sample monotonicity makes $K$ characteristic.
\end{proof}

\begin{corollary}[Graphic minimum locking bases are unimodular]
\label{cor:exchange-unimodular}
Put $d=r(k-1)$.  Suppose $|K|=d+1$ and the hypotheses of
Theorem~\ref{thm:exchange-locking} hold.  Then every spanning tree $T$ of
$H(K)$ has exactly $k-1$ edges of each slot label.  For every slot $j$, the
$j$-labelled edges project to a spanning tree on $[k]$.  Consequently the
sample-difference lattice is the full root lattice
$A_{k-1}^{\oplus r}$ and
\[
  \delta_{k,r}(K)=1.
\]
Thus every minimum-size exchange locking sample is a unimodular affine basis.
\end{corollary}

\begin{proof}
A spanning tree $T$ has $d=r(k-1)$ edges.  For a fixed $j$, projecting the
unique $T$-path between sample vertices to coordinate $j$ shows that the
$j$-labelled edges of $T$ connect all values in $\pi_j(K)=[k]$; hence there
are at least $k-1$ such edges.  Summing this lower bound over all $r$ slots
uses all $d$ edges, so equality holds for every $j$, and each projected graph
is a tree.

The difference vector of a $j$-labelled edge is a root
$e_{j,i'}-e_{j,i}$ in the $j$th copy of $A_{k-1}$.  The roots carried by the
edges of a tree on $[k]$ form a $\mathbb Z$-basis of $A_{k-1}$.  Taking the
direct sum over the $r$ slot labels gives a $\mathbb Z$-basis of
$A_{k-1}^{\oplus r}$.  Tree-edge differences and base-to-vertex differences
generate the same sample-difference lattice, proving defect one.
\end{proof}

The converse to Corollary~\ref{cor:exchange-unimodular} is false already at
rank three: higher-arity tuple moves can lock a sample whose single-slot
exchange graph is disconnected.  In the binary-index case this phenomenon can
be analyzed completely at the first three ranks.

\begin{lemma}[Unimodular tetrahedra in the three-cube]
\label{lem:cube-unimodular-orbits}
Up to independent flips of the three binary coordinates and permutation of the
coordinates, every four-element subset $K\subseteq\{0,1\}^3$ with lattice
defect one is equivalent to exactly one of
\[
\begin{aligned}
 K_1&=\{000,001,010,100\},\\
 K_2&=\{000,001,010,101\},\\
 K_3&=\{000,001,010,111\}.
\end{aligned}
\]
Their multisets of pairwise Hamming distances are respectively
\[
 \{1,1,1,2,2,2\},\qquad
 \{1,1,1,2,2,3\},\qquad
 \{1,1,2,2,2,3\},
\]
so the three types are inequivalent.
\end{lemma}

\begin{proof}
Because the difference matrix has determinant $\pm1$, its reduction modulo
$2$ is nonsingular; equivalently the four cube vertices form an affine basis
of $\mathbb F_2^3$.  We first note that some vertex of such a four-set is
incident with at least two unit cube edges inside the set.  Otherwise the
unit-edge graph has maximum degree at most one.  After a cube symmetry, if it
contains an edge we may take that edge to be $000$--$001$.  The two remaining
vertices cannot be adjacent to either endpoint, so they must be $110$ and
$111$; the xor of the four vertices is then zero, contradicting affine
independence.  If there is no unit edge, after taking one vertex to $000$ the
other three must be the weight-two vertices $011,101,110$ (including $111$
would create a unit edge with any weight-two vertex); again their xor with
$000$ is zero.  Thus a degree-two vertex exists.

Translate such a vertex to $000$ and permute coordinates so that two incident
unit edges give columns $e_2$ and $e_3$.  Write the third difference column as
$(c_1,c_2,c_3)^T$.  Nonsingularity forces $c_1=1$.  Up to interchanging the
last two coordinates, the pair $(c_2,c_3)$ is one of $00$, $01$, or $11$.
These give exactly the representatives $K_1,K_2,K_3$.  Their displayed
pairwise-distance multisets are invariant under cube symmetry and are distinct,
so the three types are inequivalent.
\end{proof}

\begin{theorem}[Sharp binary threshold for unimodular minimum samples]
\label{thm:binary-sharp-threshold}
Consider the binary-index slot family $X_{2,r}$ and minimum-size samples
$|K|=r+1$.
\begin{enumerate}[label=\textup{(\roman*)},leftmargin=*]
\item For $r=2$, every sample with $\delta_{2,2}(K)=1$ is characteristic.
\item For $r=3$, every sample with $\delta_{2,3}(K)=1$ is characteristic.
\item For $r=4$, defect one is no longer sufficient: the sample
$K_{\diamond}$ of Proposition~\ref{prop:unimodular-not-sufficient} has
$\delta_{2,4}(K_{\diamond})=1$ but is not characteristic.
\end{enumerate}
Hence rank four is the first rank in the binary family at which hierarchical
occurrence realizability gives an obstruction strictly beyond unimodularity.
\end{theorem}

\begin{proof}
For $r=2$, a defect-one three-point subset of the square is, up to cube
symmetry, $\{00,01,10\}$.  Its single-slot exchange graph is connected and
has full projection in both coordinates, so
Theorem~\ref{thm:exchange-locking} applies.

Let $r=3$.  Cube symmetries are induced by renaming slots and independently
swapping the two index terminals in a slot; they preserve the target family,
$h_{\rm slot}$, and the standard learner up to isomorphism.  By
Lemma~\ref{lem:cube-unimodular-orbits}, it therefore suffices to treat
$K_1,K_2,K_3$.  The first two have connected single-slot exchange graphs with
full coordinate projections, so they are characteristic by
Theorem~\ref{thm:exchange-locking}.

It remains to treat
\[
  K_3=\{000,001,010,111\}.
\]
Put $u_{j,\epsilon}=(a_{j,\epsilon},b_{j,\epsilon})$.  The pairs
$000,001$ and $000,010$ give semantic unary paths
\[
 u_{3,0}\leftrightarrow u_{3,1},
 \qquad
 u_{2,0}\leftrightarrow u_{2,1}.
\]
Consequently $011$ is obtained from the sampled word $010$ by changing the
third-slot child, and $101$ is obtained from $111$ by changing the second-slot
child.  A rank-two witness in the sampled word $111$ with disjoint children
$u_{2,1}$ and $u_{3,1}$ allows both changes simultaneously, giving $100$.

For the remaining word $110$, use the two-component tuples
\[
\begin{aligned}
 x&=(a_{1,0},\ a_{3,0}b_{1,0}b_{2,1}b_{3,0}),\\
 y&=(a_{1,1},\ a_{3,1}b_{1,1}b_{2,1}b_{3,1}).
\end{aligned}
\]
They have equal componentwise slot type and occur in the common arity-two
sentence context
\[
  E=\Box_1\,a_{2,1}\,\Box_2
\]
inside the sampled words $010$ and $111$.  Hence the learner contains the
semantic unary rule $[x]\to[y]$.  The sampled word $010$ supplies the rank-one
composition witness that exposes $x$ through the template
$x_1^1 a_{2,1}x_1^2$.  Inside the observed tuple $y$, a further rank-one
witness exposes the third-slot child $u_{3,1}$; replacing it by $u_{3,0}$ via
the already available unary path yields the tuple
\[
  (a_{1,1},\ a_{3,0}b_{1,1}b_{2,1}b_{3,0}),
\]
and filling the outer context $E$ gives $110$.
Thus all eight target words are generated from $K_3$.  Proposition~\ref{prop:mcfg-local-soundness} gives exactness,
and sample monotonicity gives characteristicity.

Part \textup{(iii)} is exactly
Proposition~\ref{prop:unimodular-not-sufficient}.
\end{proof}

\begin{remark}[Three nested notions of a minimum locking basis]
For $X_{2,r}$ the minimum-cardinality problem now separates three increasingly
strong notions:
\[
 \text{exchange basis}
 \Longrightarrow
 \text{characteristic basis}
 \Longrightarrow
 \text{unimodular affine basis}.
\]
At rank three the first implication is strict: $K_3$ is a
characteristic unimodular basis whose single-slot exchange graph is disconnected,
whereas the second implication is an equivalence for minimum-cardinality
samples by Theorem~\ref{thm:graphic-sharp-package}.  At rank four the second
implication becomes strict as well, witnessed by $K_{\diamond}$.  Thus higher-rank tuple composition first enlarges the locking
class beyond graphic exchange at rank three, while hierarchical scheduling
first excludes unimodular bases at rank four.
\end{remark}

\section{Detailed critical-matching analysis}
\label{app:critical-matching-details}
This appendix records the occurrence-level lemmas and constructive derivations
used in Theorem~\ref{thm:critical-matching-package}.

The preceding audit identifies the eight exceptional slot orders exactly, but
its dihedral form admits a smaller structural description.  The key is that
both mixed cube shapes contain two distinguished sample-pair exchanges whose
coordinate-difference supports form the same perfect matching on the four
abstract coordinates.

\begin{lemma}[Unary-before-composition root normal form]
\label{lem:unary-before-composition}
In the standard fixed-observation MCFG hypothesis, every derivation from an
observed tuple state has the following form at its root: a (possibly empty)
chain of semantic unary rules is followed either by a rank-zero constant rule
or by one composition-witness rule; all further semantic unary moves
occur strictly below that composition node, inside its child derivations.
Consequently a child-level change cannot retroactively enable a semantic unary
rule at an ancestor state.
\end{lemma}

\begin{proof}
Every non-start rule of the learner is either a semantic unary identity-template
rule, a rank-zero constant rule, or a composition-witness rule.  Follow
the unique root branch through unary rules until the first nonunary rule.  If it
is constant the derivation ends; otherwise its children begin independent
subderivations.  No rule below that node can alter the already chosen ancestor
state.  This is simply the top-down grammar structure, but it is the scheduling
constraint forgotten by the affine and lattice relaxations.
\end{proof}

Let
\[
 \mathcal M_{\times}
 :=\bigl\{\{0,2\},\{1,3\}\bigr\}
\]
be the \emph{crossing matching} on the abstract coordinate set
$\{0,1,2,3\}$.  The other two perfect matchings are
\[
 \mathcal M_{\parallel}
 :=\bigl\{\{0,1\},\{2,3\}\bigr\},
 \qquad
 \mathcal M_{\subset}
 :=\bigl\{\{0,3\},\{1,2\}\bigr\}.
\]
With the physical slot order $0<1<2<3$, the first matching has alternating
endpoints, the second consists of two separated adjacent pairs, and the third
is nested.  These are the three perfect-matchings of four ordered points.

For a slot order $\pi=(\pi_0,\pi_1,\pi_2,\pi_3)\in S_4$, let
$\operatorname{pos}_{\pi}(i)$ be the physical position occupied by abstract
coordinate $i$, and write
\[
 \pi_*\mathcal M_{\times}
 :=
 \bigl\{
   \{\operatorname{pos}_{\pi}(0),\operatorname{pos}_{\pi}(2)\},
   \{\operatorname{pos}_{\pi}(1),\operatorname{pos}_{\pi}(3)\}
 \bigr\}.
\]

\begin{lemma}[The mixed shapes have the same intrinsic critical matching]
\label{lem:mixed-critical-matching}
For row $10$ of Theorem~\ref{thm:rank4-complete-census}, the sample pairs
\[
 0010\leftrightarrow0111,
 \qquad
 0111\leftrightarrow1101
\]
have difference supports $\{1,3\}$ and $\{0,2\}$, respectively.  For row $13$,
\[
 0001\leftrightarrow1011,
 \qquad
 1011\leftrightarrow1110
\]
have difference supports $\{0,2\}$ and $\{1,3\}$.  Thus both mixed shapes
carry the same intrinsic critical perfect matching
$\mathcal M_{\times}$.
\end{lemma}

\begin{proof}
Immediate by coordinatewise comparison of the displayed words.
\end{proof}

\begin{proposition}[Dihedral group as the stabilizer of the crossing matching]
\label{prop:d4-stabilizer-matching}
The subgroup
\[
 \mathcal D_4
 =\{0123,1230,2301,3012,3210,0321,1032,2103\}
\]
is exactly the stabilizer of $\mathcal M_{\times}$ in $S_4$.  Equivalently,
for every $\pi\in S_4$,
\[
 \pi\in\mathcal D_4
 \iff
 \pi_*\mathcal M_{\times}=\mathcal M_{\times}.
\]
The three left cosets of $\mathcal D_4$ correspond exactly to the three
physical matching types
$\mathcal M_{\times}$, $\mathcal M_{\parallel}$, and
$\mathcal M_{\subset}$; each type is realized by eight slot orders.
\end{proposition}

\begin{proof}
A permutation preserves the matching
$\{\{0,2\},\{1,3\}\}$ exactly when it preserves the partition of the four
vertices into the two pairs, allowing the two pairs to be interchanged and the
two endpoints inside each pair to be interchanged.  This stabilizer has order
$2\cdot2\cdot2=8$ and is the displayed dihedral subgroup.  Since $S_4$ has
order $24$ and there are exactly three perfect matchings on four vertices, the
cosets are in bijection with their three images.
\end{proof}

For reference, the two mixed abstract cube shapes are
\[
 \begin{aligned}
 K_{10}&=\{0000,0001,0010,0111,1101\},\\
 K_{13}&=\{0000,0001,0110,1011,1110\}.
 \end{aligned}
\]
Their common intrinsic critical matching on the four abstract coordinates is
\[
  \mathcal M_{\times}:=\bigl\{\{0,2\},\{1,3\}\bigr\}.
\]

\begin{theorem}[Critical-matching factorization of ordered locking]
\label{thm:critical-matching-factorization}
Let $K_{10}$ and $K_{13}$ be the displayed representatives, and let $K_i^{\pi}$
be obtained by placing their abstract coordinates into physical slot order
$\pi$.  Then, for $i\in\{10,13\}$,
\[
 \boxed{
 K_i^{\pi}\text{ fails to lock}
 \iff
 \pi_*\mathcal M_{\times}=\mathcal M_{\times}.
 }
\]
Equivalently, learner behavior on each mixed cube shape factors through the
three-element coset space (equivalently, the orbit of the critical matching)
\[
 S_4/\mathcal D_4
 \cong
 \{\mathcal M_{\times},\mathcal M_{\parallel},
     \mathcal M_{\subset}\},
\]
viewed only as an $S_4$-set; no quotient-group structure is intended.
The unique nonlocking order type is the crossing matching.  The separated
and nested matching types both lock.
\end{theorem}

\begin{proof}
By Proposition~\ref{prop:d4-stabilizer-matching}, the condition
$\pi_*\mathcal M_{\times}=\mathcal M_{\times}$ is equivalent to
$\pi\in\mathcal D_4$.  Proposition~\ref{prop:rank4-dihedral-split} proved by
exact laminar-trace evaluation that these and only these eight orders fail for
both mixed shapes.  Combining the two statements gives the equivalence.
\end{proof}

This theorem compresses the twenty-four-order audit to one bit of ordered
geometry: whether the intrinsic critical matching is placed as an alternating
(crossing) matching or as one of the two noncrossing matchings.  We now remove
the finite-audit dependence from the \emph{nonlocking} half of that statement.
The proof uses only the exact laminar trace semantics, the fact that root unary
moves occur before the first composition, and one elementary interval-crossing
lemma.

For a positive occurrence trace $\tau$, write $C(\tau)$ for the set of slots
$j$ whose two terminal positions $A_j,B_j$ both belong to the support of
$\tau$.  For two index vectors $u,v\in\{0,1\}^4$, write
\[
  \operatorname{Diff}(u,v):=\{j:u_j\ne v_j\}.
\]

\begin{lemma}[Unary edges expose only completed slot differences]
\label{lem:unary-diff-completed}
Let $K\subseteq X_{2,4}$ and let $\tau$ be a positive fan-out-two occurrence
trace.  If a semantic unary edge on $\tau$ is witnessed by sample words
$w_u,w_v\in K$, then
\[
  \boxed{\operatorname{Diff}(u,v)\subseteq C(\tau).}
\]
In particular, an index of a slot that is not completed by $\tau$ cannot be
changed by a root unary chain on that trace.
\end{lemma}

\begin{proof}
Suppose $u_j\ne v_j$.  If $\tau$ omits both terminals of slot $j$, then both
terminals remain in the fixed sentence context, so the two occurrences cannot
have the same concrete context.  If $\tau$ contains exactly one of
$A_j,B_j$, the other terminal remains fixed in the context and again carries
different terminal symbols in $w_u$ and $w_v$.  Hence a shared-context unary
edge can differ at slot $j$ only when both $A_j$ and $B_j$ lie in the trace.
\end{proof}

\begin{lemma}[Exact half-slot rigidity]
\label{lem:exact-half-slot-rigidity}
Fix a derivation in $\mathfrak L(K)$ rooted at a trace $\tau$.  If slot $j$ is
not completed by $\tau$, then every occurrence of a terminal of slot $j$ that
lies in the support of $\tau$ keeps its observed index throughout the
derivation.  In particular, changing the index of a complete target slot in a
composition requires at least one child trace that completes that slot.
\end{lemma}

\begin{proof}
Induct on the derivation.  Constants do nothing.  Root unary chains preserve
all noncompleted slots by Lemma~\ref{lem:unary-diff-completed}.  Every child
support of a composition is contained in the parent support, so a child cannot
complete a slot that the parent does not complete.  The induction hypothesis
therefore applies recursively.  If a completed parent slot changed while no
child completed it, each of its two terminal positions would lie either outside
all mutable children or in a child containing only one half of the slot; both
possibilities are rigid by the first part.
\end{proof}

\begin{lemma}[Crossing completed pairs cannot be disjoint]
\label{lem:crossing-traces-intersect}
Let four physical slots occur in the order $0<1<2<3$.  Let $\tau$ complete
slots $0$ and $2$, and let $\sigma$ complete slots $1$ and $3$.  If both are
fan-out-two traces, then their supports intersect.
More generally the same holds after relabelling whenever the two completed
slot pairs form the unique crossing perfect matching of four ordered slots.
\end{lemma}

\begin{proof}
Assume the supports were disjoint.  Since $\sigma$ contains both terminals of
slots $1$ and $3$, the support of $\tau$ must avoid
$A_1,A_3,B_1,B_3$ while containing
$A_0,A_2,B_0,B_2$.  Avoiding $A_1$ separates $A_0$ from $A_2$, so two distinct
interval components are already needed in the $A$-half.  Avoiding $B_1$
likewise separates $B_0$ from $B_2$, requiring two more interval components in
the $B$-half.  Thus $\tau$ would need at least four intervals, contradicting
fan-out two.  Relabelling gives the general crossing case.
\end{proof}

The preceding lemma is the geometric part.  The sample-dependent part can be
packaged as the following scheduling principle.

\begin{theorem}[Crossing-scheduling Horn lemma]
\label{thm:crossing-scheduling-horn}
Let $K\subseteq\{0,1\}^4$ be used as a sample for $X_{2,4}$ under
$h_{\rm slot}$.  Fix two protected coordinates $p,q$.  Assume:
\begin{enumerate}[label=\textup{(H\arabic*)},leftmargin=*]
\item every sample vector satisfies $u_p\le u_q$;
\item there is a family $\mathcal U$ of two-coordinate sets such that, whenever
      $u_q=v_q=1$ and $u_p\ne v_p$, the difference support
      $\operatorname{Diff}(u,v)$ contains some $U\in\mathcal U$;
\item there is a family $\mathcal D$ of two-coordinate sets such that, whenever
      $u_p=v_p=0$ and $u_q\ne v_q$, the difference support
      $\operatorname{Diff}(u,v)$ contains some $D\in\mathcal D$;
\item for every $U\in\mathcal U$ and $D\in\mathcal D$, either
      $U\cap D\ne\varnothing$, or $U$ and $D$ are the two edges of the crossing
      perfect matching in the physical slot order.
\end{enumerate}
Then every full word generated by the laminar trace system satisfies
\[
  \boxed{x_p\le x_q.}
\]
Equivalently, $\operatorname{LReach}(K)$ is contained in the Horn half-cube
$\{x:x_p\le x_q\}$.
\end{theorem}

\begin{proof}
We prove five statements simultaneously by induction on derivation height,
after compressing the root unary chain as in
Lemma~\ref{lem:unary-before-composition}.  The trace is always rooted at one of
the sample occurrences selected by the final root unary state.

First, if a trace completes $p$ but not $q$ and its root sample has $q=0$, then
no derived filling can have $p=1$.  Indeed the root unary chain cannot change
$q$ by Lemma~\ref{lem:unary-diff-completed}; by (H1), every sample state in
that unary component with $q=0$ also has $p=0$.  If a later composition changed
$p$, Lemma~\ref{lem:exact-half-slot-rigidity} would place the change inside one
child completing $p$; that child still does not complete $q$, so the induction
hypothesis applies.  Symmetrically, if a trace completes $q$ but not $p$ and
its root sample has $p=1$, then every derived filling keeps $q=1$.

Second, suppose a trace completes $p$ but not $q$, is rooted at a sample with
$(p,q)=(0,1)$, and derives a filling with $p=1$.  If the root unary chain is
where $p$ first changes, then $q$ is unchanged along that chain and (H2),
together with Lemma~\ref{lem:unary-diff-completed}, shows that the trace
completes every coordinate of some $U\in\mathcal U$.  If the change occurs
below the first composition, the responsible child completes $p$ by
Lemma~\ref{lem:exact-half-slot-rigidity}, does not complete $q$, and is still
rooted at a sample with protected pair $(0,1)$ until the first $p$-change.  The
induction hypothesis on that child gives the same conclusion, and its support
is contained in the parent support.  Hence the parent trace also completes such
an $U$.  The symmetric argument using (H3) shows: if a trace completes $q$ but
not $p$, is rooted at protected pair $(0,1)$, and derives $q=0$, then its
support completes some $D\in\mathcal D$.

Finally consider a trace completing both $p$ and $q$.  After the root unary
chain the selected sample parent still satisfies (H1), so its protected pair is
$00$, $01$, or $11$.  Suppose for contradiction that the derived filling has
pair $10$.

If the parent pair is $00$, some child must change $p$.  If that child also
completes $q$, the induction hypothesis for the present Horn statement forbids
$10$ inside the child; if it does not complete $q$, the first guard statement
above forbids changing $p$ from a root with $q=0$.  The case of parent pair
$11$ is symmetric using the second guard statement.

Thus the only remaining case is parent pair $01$.  If one child completes both
protected slots, the induction hypothesis again forbids $10$.  Otherwise one
child must complete $p$ but not $q$ and change $p$ upward, while a distinct
child completes $q$ but not $p$ and changes $q$ downward.  By the two support
statements just proved, their supports complete some
$U\in\mathcal U$ and some $D\in\mathcal D$.  If $U\cap D\ne\varnothing$, the
children share both terminal positions of that slot and cannot be disjoint.  If
they are disjoint coordinate sets, (H4) says they form the crossing perfect
matching, and Lemma~\ref{lem:crossing-traces-intersect} again says the child
supports intersect.  Both conclusions contradict the disjoint-child condition
of a composition witness.  Hence $10$ is impossible.
\end{proof}

\begin{corollary}[Analytic Horn protection for the two mixed shapes]
\label{cor:mixed-crossing-horn-analytic}
For the crossing placement of row $10$, take
\[
  (p,q)=(0,1),\qquad
  \mathcal U=\bigl\{\{0,2\}\bigr\},\qquad
  \mathcal D=\bigl\{\{1,2\},\{1,3\}\bigr\}.
\]
For the crossing placement of row $13$, take
\[
  (p,q)=(1,2),\qquad
  \mathcal U=\bigl\{\{1,3\}\bigr\},\qquad
  \mathcal D=\bigl\{\{0,2\}\bigr\}.
\]
Then the hypotheses of Theorem~\ref{thm:crossing-scheduling-horn} hold.
Consequently every generated full word satisfies, respectively,
\[
  x_0\le x_1
  \qquad\text{and}\qquad
  x_1\le x_2.
\]
The same conclusions hold, after pulling back to abstract coordinates, for all
eight slot orders in $\mathcal D_4$.
\end{corollary}

\begin{proof}
For row $10$ the sample is
\[
 \{0000,0001,0010,0111,1101\}.
\]
All five vectors satisfy $x_0\le x_1$.  Among pairs with $x_1=1$ and changing
$x_0$, the only protected-coordinate change without changing $x_1$ is
$0111\leftrightarrow1101$, whose difference support is $\{0,2\}$.  Among pairs
with $x_0=0$ and changing $x_1$, the relevant difference supports are
$\{1,2,3\}$, $\{1,2\}$, and $\{1,3\}$; each contains an element of
$\mathcal D$.  The set $\{0,2\}$ intersects $\{1,2\}$ and forms the crossing
perfect matching with $\{1,3\}$.

For row $13$ the sample is
\[
 \{0000,0001,0110,1011,1110\}.
\]
All vectors satisfy $x_1\le x_2$.  A pair with $x_2=1$ changing $x_1$ has
difference support containing $\{1,3\}$, while a pair with $x_1=0$ changing
$x_2$ has difference support containing $\{0,2\}$.  These two sets are the
crossing perfect matching.  A crossing slot order preserves exactly this
alternation, so the same verification applies after relabelling the abstract
coordinates.
\end{proof}

\begin{lemma}[Collateral saturation inside the protected Horn region]
\label{lem:collateral-saturation}
In the crossing row-$10$ sample, slots $2$ and $3$ have single-slot semantic
unary exchanges, witnessed by $0000\leftrightarrow0010$ and
$0000\leftrightarrow0001$.  The sample contains one anchor for each allowed
protected pair $00$, $01$, $11$, namely $0000$, $0111$, and $1101$.
Therefore every vector satisfying $x_0\le x_1$ is generated.

In the crossing row-$13$ sample, slots $0$ and $3$ have single-slot exchanges,
witnessed by $0110\leftrightarrow1110$ and
$0000\leftrightarrow0001$.  The sample contains anchors for protected pairs
$00$, $01$, $11$, namely $0000$, $1011$, and $0110$.  Therefore every vector
satisfying $x_1\le x_2$ is generated.
\end{lemma}

\begin{proof}
A full sampled word supplies the rank-four composition witness whose four
children are the complete single-slot tuples
$(a_{j,i_j},b_{j,i_j})$.  A sample pair differing in exactly one slot supplies
a semantic unary rule between the two corresponding single-slot tuple states.
Hence, starting from an anchor with the desired protected pair, the two
collateral slot children can be switched independently to arbitrary binary
indices.  This produces all four collateral completions of each of the three
allowed protected pairs.
\end{proof}

\begin{lemma}[Adjacent critical-pair saturation]
\label{lem:adjacent-pair-saturation}
Let $E=\{x,y\}$ be two adjacent physical slots.  Assume that, on the exact
fan-out-two trace $\tau_E$ completing precisely $x$ and $y$, the sample-induced
semantic unary graph contains a diagonal edge between two pair assignments
that differ in both coordinates and also contains an edge changing only $y$
that is incident with one endpoint of that diagonal.  Assume moreover that the
complete single-slot trace of $y$ has the corresponding bidirectional unary
switch.  Then every binary filling of the pair $E$ is derivable from every
observed state in the connected three-vertex unary component of $\tau_E$.
In particular the pair trace is saturated:
\[
  \boxed{\operatorname{Reach}(\tau_E)=\{0,1\}^{E}.}
\]
\end{lemma}

\begin{proof}
Because $x$ and $y$ are adjacent, the exact pair occurrence is the union of
one interval in the $A$-half and one interval in the $B$-half, so it is a
legal fan-out-two trace.  The diagonal unary edge and the incident
single-coordinate edge connect three vertices of the binary square.  Thus any
observed state in that component can first move by root unary rules to any of
those three pair assignments.

Let $z$ be the fourth square vertex.  Choose the endpoint of the diagonal that
has the same $x$-coordinate as $z$; the two assignments then differ only at
$y$.  In a sample occurrence realizing that endpoint, split $\tau_E$ into the
two disjoint complete single-slot child traces for $x$ and $y$.  Rule~(L4)
therefore supplies the corresponding rank-two composition witness.  Keep the
$x$-child fixed and apply the bidirectional single-slot switch to the
$y$-child.  The resulting parent filling is exactly $z$.  Hence all four pair
assignments are derivable.
\end{proof}

The two mixed samples have exactly the local switch pattern required by the
lemma.  It is useful to record the witnesses once, in abstract coordinates:
\[
\begin{array}{c|c|c}
\text{sample} & \text{critical two-slot exchange} & \text{direct collateral switch}\\
\hline
K_{10} & 0111\leftrightarrow1101\ \text{on }\{0,2\}
       & 0000\leftrightarrow0010\ \text{on }\{2\}\\
K_{10} & 0010\leftrightarrow0111\ \text{on }\{1,3\}
       & 0000\leftrightarrow0001\ \text{on }\{3\}\\
K_{13} & 1011\leftrightarrow1110\ \text{on }\{1,3\}
       & 0000\leftrightarrow0001\ \text{on }\{3\}\\
K_{13} & 0001\leftrightarrow1011\ \text{on }\{0,2\}
       & 0110\leftrightarrow1110\ \text{on }\{0\}.
\end{array}
\]
In each row the one-coordinate switch is incident with one endpoint of the
critical diagonal in the corresponding binary pair square.  Moreover, for each
critical pair of $K_{10}$ and $K_{13}$, the pair assignments actually observed
in the sample are exactly the three vertices of this connected unary component.
Thus Lemma~\ref{lem:adjacent-pair-saturation} applies from the pair state
induced by any sampled full parent.

\begin{theorem}[Constructive noncrossing sufficiency]
\label{thm:noncrossing-constructive}
Let $K$ be either mixed sample $K_{10}$ or $K_{13}$, and place its four
abstract coordinates in an arbitrary physical slot order.  If the intrinsic
critical matching is noncrossing, then
\[
  \boxed{\operatorname{LReach}(K)=\{0,1\}^4.}
\]
More precisely, both noncrossing order types---the separated matching and the
nested matching---lock by explicit laminar derivations; no exhaustive trace
evaluation is needed.
\end{theorem}

\begin{proof}
Write the two critical pairs as $E_1=\{p,c\}$ and $E_2=\{q,d\}$, where $c,d$
are the collateral coordinates possessing the direct single-slot switches
listed above.

\emph{Separated case.}
If the critical matching is separated, its two edges occupy the two adjacent
physical pairs.  Hence each $E_i$ has an exact two-slot fan-out-two trace.
Lemma~\ref{lem:adjacent-pair-saturation}, applied to the corresponding row of
the witness table, shows that each pair trace generates all four binary
fillings of its two coordinates.

Choose any sampled full word as a parent occurrence.  Inside its full trace,
take the exact traces of $E_1$ and $E_2$ as the two children.  Their supports
are disjoint because the matching is separated, and together they cover the
four slots.  Rule~(L4) therefore gives a rank-two composition witness with
these two pair states as children.  Since each child can independently realize
all four pair fillings, the parent realizes all $4\cdot4=16$ cube vertices.

\emph{Nested case.}
Now one critical edge is the inner adjacent pair and the other is the outer
pair.  By Lemma~\ref{lem:adjacent-pair-saturation}, the inner pair trace is
again saturated and can realize all four assignments of its two coordinates.
Let the outer edge be $\{x,y\}$, where $y$ is its directly switchable
collateral coordinate and $x$ is the other endpoint.  The sample contains the
two endpoints of the outer critical exchange; because that exchange changes
both outer coordinates, those two sampled words have opposite values of $x$.

Fix either one of those two sampled words as the full parent occurrence.  In
Rule~(L4) choose as children (i) the saturated exact trace of the inner pair and
(ii) the complete single-slot trace of $y$.  These two child supports are
disjoint in the nested configuration.  Leave the two terminals of $x$ as
explicit terminal material of the parent template.  The inner child can now
choose its four assignments independently and the $y$-child can choose either
bit by its direct switch, while $x$ stays fixed at the value carried by the
chosen parent sample.  Thus this parent generates all eight cube vertices with
that fixed value of $x$.  Repeating the same construction from the other
endpoint of the outer critical exchange gives the other value of $x$ and the
remaining eight vertices.  Hence all sixteen words are generated.
\end{proof}

\begin{corollary}[Fully analytic critical-matching classification]
\label{cor:critical-matching-analytic}
For either mixed shape $K_{10}$ or $K_{13}$ and every slot order $\pi$,
\[
  \boxed{
  K^{\pi}\text{ locks}
  \iff
  \pi_*\mathcal M_{\times}\ne\mathcal M_{\times}.
  }
\]
Equivalently, the crossing placement is the unique nonlocking order type,
whereas both the separated and nested placements lock.  Thus the
crossing/noncrossing split of the mixed shapes is independent of the finite
rank-four audit.
\end{corollary}

\begin{proof}
If the critical matching is crossing,
Theorem~\ref{thm:crossing-scheduling-horn} and
Corollary~\ref{cor:mixed-crossing-horn-analytic} preserve a proper Horn
half-cube, so locking fails.  If the matching is noncrossing,
Theorem~\ref{thm:noncrossing-constructive} gives full cube reachability.
These are the three possible perfect-matching order types on four linearly
ordered slots.
\end{proof}

\begin{proposition}[Protected Horn half-cubes in the crossing order]
\label{prop:protected-horn-halfcubes}
Identify generated physical words back with their abstract coordinates.  If
$\pi_*\mathcal M_{\times}=\mathcal M_{\times}$, then
\[
 \boxed{
  \operatorname{LReach}(K_{10}^{\pi})
  =\{x\in\{0,1\}^4:x_0\le x_1\},
 }
\]
and
\[
 \boxed{
  \operatorname{LReach}(K_{13}^{\pi})
  =\{x\in\{0,1\}^4:x_1\le x_2\}.
 }
\]
Thus every crossing placement reaches exactly twelve cube vertices.  For either
noncrossing matching type, the corresponding reachable set is all of
$\{0,1\}^4$.
\end{proposition}

\begin{proof}
For a crossing placement, Corollary~\ref{cor:mixed-crossing-horn-analytic}
gives the inclusion into the appropriate Horn half-cube, while
Lemma~\ref{lem:collateral-saturation} gives the reverse inclusion.  Hence the
reachable set contains exactly twelve vertices, and the crossing placements
are nonlocking without any exhaustive trace calculation.

For the two noncrossing matching types, full reachability is now constructive:
Theorem~\ref{thm:noncrossing-constructive} gives explicit separated and nested
laminar derivations that generate the full cube.  Hence the entire mixed-shape
crossing/noncrossing classification is analytic; the finite rank-four audit is
retained only as an independent census certificate.
\end{proof}

\begin{remark}[Why crossing is the right ordered invariant]
The matching formulation explains why a dihedral group appeared in the audit.
The two mixed samples contain two decisive pair exchanges supported on the two
abstract diagonals.  When those diagonals remain physically opposite, neither
exchange is an isolated adjacent two-slot block.  A useful exchange then carries
collateral slot material, and Lemma~\ref{lem:unary-before-composition} prevents
child-level corrections from being performed first and then used to enable an
ancestor unary move.  Theorem~\ref{thm:crossing-scheduling-horn} now makes the first half of this
interpretation literal: the crossing schedule preserves a protected Horn
implication by a direct induction on laminar derivations.  In either
noncrossing order type, the critical exchanges can instead be scheduled
constructively.  In the separated case both critical pairs become adjacent
saturated pair traces; in the nested case the inner pair saturates locally and
the outer protected coordinate is supplied by the two endpoints of the outer
critical exchange while its collateral coordinate is corrected in a disjoint
child.  Thus the Horn barrier disappears for a concrete laminar reason, not
merely by finite enumeration.
\end{remark}

\begin{remark}[Toward a crossing graph for higher rank]
The rank-four result suggests a more general ordered exposure invariant.  Given
a family of critical two-coordinate exchanges, place their supports as chords
of the linearly ordered slot set and form the intersection graph of crossing
chords.  The present mixed shapes have one critical perfect matching, and its
crossing/noncrossing type completely determines locking.  For larger rank this
need not remain a one-bit criterion: nested exchanges, shared endpoints, and
higher-arity witnesses can create non-matroidal scheduling dependencies.  A
natural next problem is to determine whether bounded crossing-graph structure
(e.g. outerplanarity or bounded treewidth) yields tractable sufficient or exact
criteria for laminar locking.
\end{remark}

\section{Hierarchical certificate details}
\label{app:hierarchy-details}
This appendix gives the full definitions, intermediate thresholds, certificate
trees, and module-restricted Pareto proofs summarized in
Theorem~\ref{thm:hierarchy-package}.

\subsection{A fan-out--hierarchy profile for the mixed rank-four shapes}
\label{subsec:fanout-hierarchy-profile}

The multiway block-cover number measures the component budget required by one
simultaneous sibling decomposition.  It does not, by itself, measure the power
of a whole derivation tree: a hierarchy may reuse different sampled parents and
may therefore realize a target even when one flat decomposition would require
larger fan-out.  The mixed rank-four shapes give a small exact example of this
distinction.

For this purpose only, it is convenient to vary the positional fan-out budget
while keeping the same sample and slot semantics.  For an integer $f\ge1$, let
$\mathfrak L_f(K)$ be the positive occurrence-trace system obtained from
$\mathfrak L(K)$ by allowing traces with at most $f$ nonempty interval
components, and by using the same observed-state, semantic-unary, constant,
and composition-witness clauses.  Let
\[
  \operatorname{LReach}_f(K)\subseteq\{0,1\}^r
\]
be its full-word reachability set.  Thus, under the nonempty occurrence convention,
$\operatorname{LReach}_2(K)=\operatorname{LReach}(K)$, while for $f>2$ the
system is used here as an auxiliary scheduling relaxation.  We do not need a
separate claim that every such relaxed system is a sound learner for a larger
language class.

\begin{lemma}[Monotonicity of the trace fan-out relaxation]
\label{lem:trace-fanout-monotone}
For every sample $K$ and $f\ge1$,
\[
  \operatorname{LReach}_f(K)
  \subseteq
  \operatorname{LReach}_{f+1}(K).
\]
\end{lemma}

\begin{proof}
Every trace and every witness legal with at most $f$ interval components is
also legal with at most $f+1$ components.  Hence $\mathfrak L_f(K)$ is a
subgrammar of $\mathfrak L_{f+1}(K)$.
\end{proof}

The Horn obstruction scales with the same budget.

\begin{theorem}[Fan-out-$f$ Horn-conflict certificate]
\label{thm:fanout-f-horn-conflict}
Let $K\subseteq\{0,1\}^r$, fix protected coordinates $p,q$, and assume
conditions \textup{(C1)} and \textup{(C2)} of
Theorem~\ref{thm:abstract-horn-conflict}.  If
\[
  \operatorname{bc}(U,D)>f
  \qquad
  \text{for every }U\in\mathcal U,\ D\in\mathcal D,
\]
then
\[
  \operatorname{LReach}_f(K)
  \subseteq
  \{x\in\{0,1\}^r:x_p\le x_q\}.
\]
\end{theorem}

\begin{proof}
Repeat the simultaneous Horn induction from
Theorem~\ref{thm:abstract-horn-conflict}.  In the only nontrivial
$01\to10$ case, distinct children would have to complete some mandatory
$U\in\mathcal U$ and $D\in\mathcal D$.  By the multiway scheduling theorem,
$\operatorname{bc}(U,D)>f$ means that no pair of disjoint supports with at
most $f$ interval components each can complete those two slot sets.  Thus the
required sibling placement is impossible.  All other cases are unchanged.
\end{proof}

We also need the positive local counterpart of this scaled budget.  For an
observed state $s$ on a slot block $B$, let
$\operatorname{Fill}^{(f)}_K(s;B)$ denote the assignments on $B$ derivable
from $s$ inside $\mathfrak L_f(K)$.

\begin{lemma}[Pair saturation at any legal support budget]
\label{lem:pair-saturation-any-f}
Let $E=\{x,y\}$ be a two-slot block in a binary sample.  Suppose a trace
$\tau_E$ legal at fan-out $f$ completes exactly the two slots of $E$, and the
sample induces on that trace a bidirectional diagonal unary edge changing both
$x$ and $y$.  Suppose also that one endpoint of the diagonal is incident with
a bidirectional complete single-slot switch on $y$.  Then the state of
$\tau_E$ is saturated in $\mathfrak L_f(K)$:
\[
  \operatorname{Fill}^{(f)}_K(\tau_E;E)=\{0,1\}^E.
\]
\end{lemma}

\begin{proof}
The diagonal edge and the incident $y$-switch give three vertices of the
binary square by root unary moves.  To obtain the fourth, use a sampled
endpoint having the required $x$-value and split the pair support into the two
complete single-slot child traces for $x$ and $y$.  These children are
disjoint and each has fan-out at most two, hence at most $f$ in every case of
interest below.  Keep the $x$-child fixed and apply the $y$-switch below the
composition node.  This produces the fourth filling.  The argument is the
same local mechanism as Lemma~\ref{lem:adjacent-pair-saturation}; adjacency
there served only to make the pair trace legal already at fan-out two.
\end{proof}

For one ordered copy of either mixed shape, let $E_1,E_2$ be its two intrinsic
critical pairs.  Define the \emph{flat critical-pair fan-out}
\[
  f_{\rm flat}:=\operatorname{bc}(E_1,E_2),
\]
and define the trace-reachability threshold
\[
  f_{\rm reach}
  :=\min\{f\ge2:\operatorname{LReach}_f(K)=\{0,1\}^4\},
\]
with value $\infty$ if no such $f$ exists.  The first quantity asks how much
fan-out is needed to place the two critical pair states simultaneously as
siblings of one sampled full parent.  The second allows an arbitrary
hierarchical derivation in the relaxed trace system.

\begin{theorem}[Exact fan-out--hierarchy profile of the mixed shapes]
\label{thm:exact-fanout-hierarchy-profile}
For either mixed shape $K_{10}$ or $K_{13}$, the three order types have the
following exact profile:
\[
\begin{array}{c|c|c|c}
\text{order type}&f_{\rm flat}&f_{\rm reach}&f_{\rm flat}-f_{\rm reach}\\
\hline
\text{separated}&2&2&0\\
\text{nested}&3&2&1\\
\text{crossing}&4&4&0.
\end{array}
\]
Thus the nested configuration exhibits a strict hierarchy dividend: a flat
one-parent realization of the two critical-pair freedoms needs fan-out three,
but a hierarchical/two-parent construction already locks at fan-out two.  The
crossing configuration admits no such rescue below its four-component
threshold.
\end{theorem}

\begin{proof}
The flat values are exactly the common maxima of the requirement vectors in
Corollary~\ref{cor:matching-scheduling-vectors}, namely $2,3,4$ for separated,
nested, and crossing order.

For separated order, Theorem~\ref{thm:noncrossing-constructive} gives full
reachability already at fan-out two, so $f_{\rm reach}=2$.  For nested order,
the same theorem gives the explicit two-parent construction at fan-out two:
the inner critical pair is saturated locally, while the two endpoints of the
outer critical exchange provide the two fixed values of the remaining outer
coordinate.  Hence again $f_{\rm reach}=2$, although a single simultaneous
critical-pair placement has flat demand three.

For crossing order, the mandatory up/down supports are the two crossing
critical pairs.  Their two-colour block-cover number is four.  Therefore
Theorem~\ref{thm:fanout-f-horn-conflict} preserves the proper Horn half-cube
for every $f\le3$, proving $f_{\rm reach}\ge4$.  At fan-out four, the multiway
scheduling theorem supplies disjoint supports for the two critical pairs, each
using four run components.  By Lemma~\ref{lem:pair-saturation-any-f}, both
critical pair states are saturated.  Using them as the two children of any
sampled full parent yields all $4\cdot4=16$ assignments.  Hence
$\operatorname{LReach}_4(K)=\{0,1\}^4$, so $f_{\rm reach}=4$.
\end{proof}

\begin{corollary}[Hierarchy and fan-out are distinct positive resources]
\label{cor:hierarchy-fanout-distinct}
The block-cover number measures the fan-out needed by one simultaneous sibling
placement, not the minimum fan-out of an unrestricted derivation.  In
particular, the nested mixed shape has
\[
  f_{\rm reach}=2<3=f_{\rm flat}.
\]
Hence hierarchical reuse of exposed parent configurations can substitute for
one unit of positional fan-out.  By contrast, the crossing Horn certificate
shows that hierarchy cannot substitute below the conflict threshold when every
protected up/down support pair remains incompatible.
\end{corollary}

\begin{remark}[Exposure tradeoff, not a language-class claim]
The quantities above are deliberately defined inside the finite occurrence
trace semantics of the fixed sample.  For $f=2$ they coincide with the exact
learner geometry already proved in Theorem~\ref{thm:laminar-trace-exact}.  For
$f>2$ they quantify how much additional positional capacity would be required
to schedule the same exposed sample information.  The point is therefore not
to assert a new hierarchy of MCFG language classes, but to separate two finite
reconstruction resources: component budget inside one composition node and
hierarchical reuse of several exposed parent configurations.
\end{remark}

\subsection{Module-relative hierarchy dividends}
\label{subsec:module-hierarchy-dividend}

The preceding profile was phrased using the two critical pairs of the mixed
rank-four samples.  We now isolate the underlying notion so that the word
``hierarchy dividend'' has a precise module-relative meaning rather than being
an informal comparison of two proofs.

Let $B_1,\ldots,B_t$ be pairwise disjoint slot blocks whose union is $[r]$, and
let
\[
  F_j\subseteq\{0,1\}^{B_j}
\]
be desired filling sets.  Write
\[
  \mathcal Q(\mathbf B,\mathbf F)
  :=\{x\in\{0,1\}^r:x|_{B_j}\in F_j\text{ for every }j\}
\]
for their Cartesian module product.

\begin{definition}[Hierarchical module-realization threshold]
\label{def:module-hier-threshold}
For a sample $K\subseteq X_{2,r}$ define
\[
 f_{\rm hier}(K;\mathbf B,\mathbf F)
 :=\min\bigl\{f\ge2:
      \mathcal Q(\mathbf B,\mathbf F)
      \subseteq\operatorname{LReach}_f(K)\bigr\},
\]
with value $\infty$ if the set is empty.  This quantity allows arbitrary
nested/disjoint learner derivations and arbitrary reuse of sampled parent
configurations.
\end{definition}

To define the corresponding flat threshold, require one sampled full parent
occurrence to expose observed child states $s_1,\ldots,s_t$ whose legal traces
complete the blocks $B_1,\ldots,B_t$, respectively, such that the child traces
are pairwise support-disjoint and
\[
  F_j\subseteq\operatorname{Fill}^{(f)}_K(s_j;B_j)
  \qquad(1\le j\le t).
\]
The remaining parent material is allowed only on coordinates outside
$\bigcup_jB_j$; in the present partitioned setting there is none.

\begin{definition}[Flat module threshold and hierarchy dividend]
\label{def:module-flat-threshold}
Let
\[
 f_{\rm flat}^{\rm mod}(K;\mathbf B,\mathbf F)
\]
be the least $f\ge2$ for which such a one-parent simultaneous module witness
exists, with value $\infty$ if none exists.  Whenever both thresholds are
finite define the \emph{module-relative hierarchy dividend}
\[
 \boxed{
 \Delta_{\rm hier}(K;\mathbf B,\mathbf F)
 :=f_{\rm flat}^{\rm mod}(K;\mathbf B,\mathbf F)
   -f_{\rm hier}(K;\mathbf B,\mathbf F).
 }
\]
\end{definition}

The flat threshold contains two logically separate demands: the child modules
must themselves expose the desired fill sets, and they must be schedulable
simultaneously under one parent.  The second demand has an exact positional
lower bound.

\begin{proposition}[Block-cover lower bound for flat module realization]
\label{prop:flat-module-bc-lower}
If
$f_{\rm flat}^{\rm mod}(K;\mathbf B,\mathbf F)<\infty$, then
\[
 \boxed{
 f_{\rm flat}^{\rm mod}(K;\mathbf B,\mathbf F)
 \ge \operatorname{bc}(B_1,\ldots,B_t).
 }
\]
If, at $f=\operatorname{bc}(B_1,\ldots,B_t)$, the desired module fill sets are
already available at observed child states and a sampled parent witness
realizes the block-cover placement, then equality holds.
\end{proposition}

\begin{proof}
Any one-parent simultaneous witness supplies pairwise disjoint child supports
completing $B_1,\ldots,B_t$.  The multiway scheduling theorem says that their
common fan-out is at least the block-cover number.  Under the additional
availability hypothesis, the run-hull construction supplies the required
pairwise disjoint supports at exactly that budget, and the sampled parent
witness together with the child fill sets gives the flat module realization.
\end{proof}

\begin{proposition}[Flat-to-hierarchical comparison]
\label{prop:flat-dominates-hier}
Assume the flat module threshold is finite.  Then
\[
 \boxed{
 f_{\rm hier}(K;\mathbf B,\mathbf F)
 \le
 f_{\rm flat}^{\rm mod}(K;\mathbf B,\mathbf F).
 }
\]
Consequently every finite module-relative hierarchy dividend is nonnegative.
\end{proposition}

\begin{proof}
A one-parent simultaneous module witness is itself a legal derivation pattern
inside the relaxed trace system at the same fan-out.  Independent child
fillings therefore generate every element of
$\mathcal Q(\mathbf B,\mathbf F)$, which is exactly the defining containment
for $f_{\rm hier}$.
\end{proof}

\begin{theorem}[Critical-module hierarchy dividend for the mixed shapes]
\label{thm:critical-module-dividend}
For either mixed sample $K_{10}$ or $K_{13}$, let
$\mathbf B=(E_1,E_2)$ be its two intrinsic critical pairs and let
$F_1=\{0,1\}^{E_1}$, $F_2=\{0,1\}^{E_2}$.  Since $E_1\dot\cup E_2=[4]$,
\[
  \mathcal Q(\mathbf B,\mathbf F)=\{0,1\}^4.
\]
For the three physical order types,
\[
\begin{array}{c|c|c|c}
\text{order type}
&f_{\rm flat}^{\rm mod}
&f_{\rm hier}
&\Delta_{\rm hier}\\
\hline
\text{separated}&2&2&0\\
\text{nested}&3&2&1\\
\text{crossing}&4&4&0.
\end{array}
\]
Thus the one-unit nested hierarchy dividend is an intrinsic statement about
realizing the Cartesian product of the two critical-pair freedoms, not merely
a comparison of two hand-chosen derivations.
\end{theorem}

\begin{proof}
The hierarchical thresholds are exactly the $f_{\rm reach}$ values of
Theorem~\ref{thm:exact-fanout-hierarchy-profile}, because the two critical
pairs partition all four coordinates and their desired fill sets are the full
binary squares.  For the flat thresholds, Proposition~\ref{prop:flat-module-bc-lower}
gives the lower bounds $2,3,4$ from
Corollary~\ref{cor:matching-scheduling-vectors}.  At those same budgets the
critical pair states are saturable by Lemma~\ref{lem:pair-saturation-any-f},
and the run-hull placements supplied by the multiway scheduling theorem give a
sampled full-parent witness with the two pair states as siblings.  Hence all
three lower bounds are attained.  Subtracting gives the displayed dividends.
\end{proof}

\begin{remark}[What the dividend measures]
The quantity $\Delta_{\rm hier}$ is relative to a designated family of module
freedoms.  It does not claim that every full-cube derivation at fan-out two
must visibly contain those modules as named intermediate states, nor does it
measure ordinary MCFG language-class power.  It asks a narrower reconstruction
question: how much component budget is saved when a Cartesian family of
already identifiable freedoms may be assembled through arbitrary exposed
hierarchy rather than forced to coexist as siblings of one parent.  The nested
critical-pair system is a smallest strict example in the present family under the stated module setup.
\end{remark}

\subsection{Hierarchical Cartesian certificate trees}
\label{subsec:cartesian-certificate-trees}

The module threshold $f_{\rm hier}$ is defined semantically by unrestricted
trace reachability.  It is useful to place between that semantic optimum and
the one-parent flat threshold a finite proof object which records \emph{how}
local fill freedoms are assembled.  The resulting object is an alternating
cover/product tree.  It does not claim to describe every successful learner
derivation; rather, it is a reusable sufficient certificate whose width has a
direct fan-out meaning.

Fix a sample $K\subseteq X_{2,r}$.  A \emph{local fill fact} is a quadruple
$(s,B,F,\lambda)$ where $s$ is an observed state carried by a trace completing
the slot block $B$, $F\subseteq\{0,1\}^{B}$, and
\[
  F\subseteq \operatorname{Fill}^{(\lambda)}_K(s;B).
\]
The integer $\lambda\ge2$ is the certified fan-out budget of that local fact.
The lower bound $2$ is part of this certificate language because the objects
being compared are relaxations of the native fan-out-two learner; width one is
not used as an assembly resource in this section.  Such
facts may come from direct unary switches, pair saturation, or any independently
proved local lemma.  In particular, the framework below treats local
saturation results as reusable modules instead of reproving them at every
parent occurrence.

\begin{definition}[Cartesian certificate tree]
\label{def:cartesian-certificate-tree}
A \emph{state-rooted Cartesian certificate tree} for $(s,B,F)$ is built by the
following constructors.
\begin{enumerate}[label=\textup{(T\arabic*)},leftmargin=*]
\item An \emph{atomic leaf} is a local fill fact $(s,B,F,\lambda)$.
Its width is $\lambda$.
\item A \emph{unary-routing node} may replace $s$ by an observed state $s'$ in
the same sample-induced semantic-unary component and attach a certificate for
$(s',B,F)$.  Its width is the width of its child.
\item A \emph{product node} consists of a composition witness
from the parent state $s$ to observed child states $s_1,\ldots,s_t$, carried by
pairwise support-disjoint child traces which complete pairwise disjoint slot
blocks $B_1,\ldots,B_t\subseteq B$.  Attach to child $j$ a certificate for
$(s_j,B_j,F_j)$.  Coordinates
\[
  R:=B\setminus\bigcup_{j=1}^{t}B_j
\]
are left as explicit terminal material of the chosen parent occurrence and are
therefore fixed to its anchor assignment $u|_R$.  The node certifies the box
\[
 Q=\Bigl\{x\in\{0,1\}^{B}:
      x|_R=u|_R,\quad x|_{B_j}\in F_j\ (1\le j\le t)\Bigr\}.
\]
Its width is the maximum of the child widths and the fan-outs of the concrete
child traces used by the witness.
\end{enumerate}
A \emph{start cover} is a finite family of such trees rooted at sampled full
word states.  If their certified full-word sets are $Q_1,\ldots,Q_m$, the cover
certifies every target set $Q\subseteq\bigcup_aQ_a$.  Its width is the maximum
width of its constituent trees.  We call the least width of a start cover for
$Q$ the \emph{Cartesian certificate width} and write
\[
  \operatorname{ctw}_K(Q).
\]
If no such certificate exists, the value is $\infty$.
\end{definition}

The definition deliberately separates two operations.  Product nodes are
``AND'' nodes: all child freedoms are used independently under one composition
witness.  The start cover is an ``OR'' node: different sampled parent anchors
may cover different boxes.  This is exactly the distinction exploited by the
nested rank-four construction.

\begin{theorem}[Soundness and max-node principle]
\label{thm:certificate-tree-soundness}
If a target set $Q\subseteq\{0,1\}^r$ has a Cartesian start cover of width
$f$, then
\[
  \boxed{Q\subseteq\operatorname{LReach}_f(K).}
\]
Consequently
\[
  f_{\rm hier}(K;\mathbf B,\mathbf F)
  \le
  \operatorname{ctw}_K\bigl(\mathcal Q(\mathbf B,\mathbf F)\bigr)
\]
whenever the right-hand side is finite.  Moreover the fan-out cost of a
certificate is the \emph{maximum} local product/leaf cost appearing in the
tree; increasing its depth or using additional cover branches does not add
those local costs.
\end{theorem}

\begin{proof}
Induct from the leaves upward.  An atomic leaf is sound by the defining local
fill fact.  Unary routing is sound because semantic-unary moves are learner
rules and do not alter the certified terminal filling.  At a product node, the
chosen positive-support rule is legal at fan-out $f$ by the width bound.  The
child supports are pairwise disjoint, so their certified derivations can be
performed independently below that rule; the remaining parent material is
fixed exactly as in the definition of the box.  Hence every point of the
certified box is derivable from the parent state.  Finally, each root tree of a
start cover begins at a sampled full-word state, hence every set certified by a
root tree lies in $\operatorname{LReach}_f(K)$, and unions preserve the
inclusion.  The max-node assertion is the same induction: no composition step
adds interval components belonging to different levels of the derivation tree;
it only requires its own children to fit under the common budget $f$.
\end{proof}

At each product node, the multiway scheduling theorem gives an immediate
combinatorial lower bound on its local width.

\begin{proposition}[Block-cover lower bound at every certificate node]
\label{prop:tree-node-block-cover}
Let a product node have child blocks $B_1,\ldots,B_t$ and local width $f_v$.
Then
\[
  \boxed{\operatorname{bc}(B_1,\ldots,B_t)\le f_v.}
\]
Thus every certificate tree $T$ satisfies
\[
  \max_{v\in\operatorname{Prod}(T)}
      \operatorname{bc}(B_{v,1},\ldots,B_{v,t_v})
  \le \operatorname{width}(T).
\]
If every product node uses the run-hull placement of
Theorem~\ref{thm:multiway-support-scheduling} and its local child fill facts
are already available at that budget, equality holds nodewise.
\end{proposition}

\begin{proof}
The concrete children of a product node are pairwise support-disjoint traces
which complete $B_1,\ldots,B_t$ and each have fan-out at most $f_v$.
Theorem~\ref{thm:multiway-support-scheduling} therefore forces the displayed
block-cover inequality.  Taking the maximum over product nodes gives the tree
bound.  The run-hull construction attains the block-cover number whenever the
stated witness and fill availability hypotheses hold.
\end{proof}

A flat module witness is a depth-one special case: one sampled full parent is a
single product node whose children are the designated module states.  Hence the
certificate width lies between unrestricted hierarchical reachability and flat
assembly.

\begin{corollary}[Certificate-width sandwich]
\label{cor:certificate-width-sandwich}
Whenever the flat module threshold is finite,
\[
 \boxed{
 f_{\rm hier}(K;\mathbf B,\mathbf F)
 \le
 \operatorname{ctw}_K\bigl(\mathcal Q(\mathbf B,\mathbf F)\bigr)
 \le
 f_{\rm flat}^{\rm mod}(K;\mathbf B,\mathbf F).
 }
\]
The left inequality may be strict because Cartesian certificate trees are only
a sufficient proof language; no completeness claim is made for arbitrary
laminar learner derivations.
\end{corollary}

\begin{proof}
The left inequality is Theorem~\ref{thm:certificate-tree-soundness}.  A finite
flat module witness is itself a one-product-node certificate rooted at its
sampled parent, so it gives the right inequality at the same fan-out.
\end{proof}

\begin{theorem}[Exact certificate widths for the mixed shapes]
\label{thm:mixed-certificate-widths}
For either mixed sample $K_{10}$ or $K_{13}$, with target
$Q=\{0,1\}^4$, the three physical order types have
\[
 \boxed{
 \begin{array}{c|c|c|c}
 \text{order type}&f_{\rm hier}&\operatorname{ctw}_K(Q)&f_{\rm flat}^{\rm mod}\\
 \hline
 \text{separated}&2&2&2\\
 \text{nested}&2&2&3\\
 \text{crossing}&4&4&4.
 \end{array}}
\]
Thus the strict nested hierarchy dividend is already witnessed inside the
finite Cartesian certificate language: replacing one flat product node of
cost three by a two-branch cover of cost-two product nodes saves one unit of
fan-out.
\end{theorem}

\begin{proof}
The outer columns are Theorem~\ref{thm:critical-module-dividend}.  Hence the
certificate-width sandwich leaves only the nested middle entry to construct;
the other two are forced by equal endpoints.  For the nested order, use the
two endpoints of the outer critical exchange as the two roots of a start
cover.  At each root, take a product node whose children are (i) the saturated
inner-pair state and (ii) the directly switchable outer-collateral singleton.
The two child blocks are nonreturning in physical order, so their block-cover
number is two, and the concrete traces used in
Theorem~\ref{thm:noncrossing-constructive} are fan-out two.  Each root tree
therefore certifies one eight-vertex box at width two; the two boxes differ in
the explicit outer coordinate and cover the full cube.  Thus
$\operatorname{ctw}_K(Q)\le2$.  Since the definition only allows fan-out
budgets at least two in the present MCFG trace setting, equality follows.
\end{proof}

\begin{remark}[Width--anchor tradeoff]
The certificate formalism also records a second resource that the scalar
hierarchy dividend suppresses.  Count only the sampled full-parent roots used
by the start cover, treating local fill facts as already certified modules.
Then the separated construction has an assembly certificate of type
$(\text{width},\text{roots})=(2,1)$, while the nested system has both the flat
certificate $(3,1)$ and the hierarchical certificate $(2,2)$.  Thus the
one-unit fan-out saving is purchased by reusing two exposed parent
configurations.  For unrestricted Cartesian certificates, the definition of an atomic local fill
fact is intentionally permissive, so this observation alone does not prove a
global lower bound on the number of roots.  The next subsection therefore
fixes the natural critical-module library explicitly.  Relative to that library,
the two displayed constructions form an exact two-point Pareto frontier.  This
gives a precise width--anchor refinement of positive exposure without
overstating a lower bound for arbitrary proof languages.
\end{remark}

\subsection{Exact width--anchor Pareto profiles for critical-module assembly}
\label{subsec:width-anchor-pareto}

The preceding certificate width minimizes one scalar resource and deliberately
allows any independently established local fill fact as an atomic leaf.  To
measure the cost of \emph{assembling a fixed exposed module library}, we now
freeze the allowable atomic facts and count cover branches as a second
resource.  This makes the width--anchor tradeoff exact while keeping its scope
explicit.

\begin{definition}[Module-restricted Cartesian profile]
\label{def:module-restricted-profile}
Let $\mathcal M$ be a finite collection of local fill facts in the sense of
Definition~\ref{def:cartesian-certificate-tree}.  An
$\mathcal M$-\emph{certificate tree} is a Cartesian certificate tree whose
atomic leaves all belong to $\mathcal M$.  For a start cover $C$, let
$a(C)$ be the number of root branches in the cover.  Define
\[
 \mathsf{AProf}_{K,\mathcal M}(Q)
 :=\bigl\{(f,a):
       \text{$Q$ has an $\mathcal M$-start cover $C$ with
       $\operatorname{width}(C)\le f$ and $a(C)\le a$}\bigr\}.
\]
Its componentwise-minimal elements form the
\emph{module-assembly Pareto frontier}, denoted
\[
  \Front\mathsf{AProf}_{K,\mathcal M}(Q).
\]
The second coordinate counts proof branches, not ordinary sample size.  In the
mixed constructions below the branches are rooted at distinct sampled
full-parent configurations, so it also equals the number of exposed parent
anchors used by the assembly proof.
\end{definition}

The following elementary observation is the lower-bound device needed for the
nested case.

\begin{lemma}[Selector constancy for a restricted module library]
\label{lem:selector-constancy}
Fix a coordinate $x$ and a fan-out budget $f$.  Suppose that every atomic fact
$(s,B,F,\lambda)\in\mathcal M$ with $\lambda\le f$ and $x\in B$ is constant in
coordinate $x$, that is, all elements of $F$ have the same $x$-value.  Then
every width-$f$ $\mathcal M$-certificate tree whose certified block contains
$x$ certifies a set on which $x$ is constant.  Consequently, if a target set
$Q$ contains points with both values of $x$, every width-$f$
$\mathcal M$-start cover of $Q$ needs at least two branches.
\end{lemma}

\begin{proof}
Induct on the certificate tree.  The assertion is the hypothesis for atomic
leaves.  A unary-routing node has exactly the same certified filling set as
its child, so it preserves constancy.  At a product node, pairwise disjoint
child blocks imply that at most one child block contains $x$.  If no child
contains $x$, then $x$ lies in the explicit remainder $R$ and is fixed to the
parent anchor value.  If one child contains $x$, its certified filling set is
constant in $x$ by induction, while all other children are disjoint from $x$.
Thus the parent box is again constant in $x$.  A single root branch therefore
covers at most one $x$-slice.  Covering a target set meeting both slices needs
at least two branches.
\end{proof}

For an ordered mixed sample, let $E_1,E_2$ be the two critical pairs and let
$c_i\in E_i$ be the directly switchable collateral coordinate from the witness
table preceding Theorem~\ref{thm:noncrossing-constructive}.  Define the
\emph{critical-module library} $\mathcal M_{\rm crit}$ to contain all
state-rooted instances of the following already proved local facts:
\begin{enumerate}[label=\textup{(M\arabic*)},leftmargin=*]
\item the saturated pair fact
$F_i=\{0,1\}^{E_i}$ at the least legal exact-pair support budget supplied by
Lemma~\ref{lem:pair-saturation-any-f}; and
\item the complete singleton switch
$\{0,1\}^{\{c_i\}}$ at fan-out two.
\end{enumerate}
Unary-equivalent state copies of the same fact are included so that state
matching, rather than the choice of representative, does not affect the
assembly profile.  For the separated, nested, and crossing order types the
pair-fact budgets are, respectively,
\[
 (2,2),\qquad(3,2),\qquad(4,4),
\]
after listing the outer pair first in the nested case.  These are exactly the
requirement vectors of Corollary~\ref{cor:matching-scheduling-vectors}.

\begin{theorem}[Exact critical-module width--anchor frontiers]
\label{thm:critical-module-pareto}
Let $K$ be either mixed sample $K_{10}$ or $K_{13}$, let
$Q=\{0,1\}^4$, and use the natural critical-module library
$\mathcal M_{\rm crit}$ above.  Then the module-assembly Pareto frontiers are
\[
 \boxed{
 \begin{array}{c|c}
 \text{physical order type}
 &\Front\mathsf{AProf}_{K,\mathcal M_{\rm crit}}(Q)\\
 \hline
 \text{separated}&\{(2,1)\}\\
 \text{nested}&\{(2,2),(3,1)\}\\
 \text{crossing}&\{(4,1)\}.
 \end{array}}
\]
In particular, the nested system has an exact two-point width--anchor
tradeoff: one parent configuration suffices at width three, whereas width two
requires exactly two cover branches.
\end{theorem}

\begin{proof}
\emph{Separated order.}
Both critical-pair facts are available at width two, and their exact pair
supports are disjoint under one sampled full parent.  The one-product-node
construction of Theorem~\ref{thm:noncrossing-constructive} therefore gives
$(2,1)$.  Since the present trace model starts at fan-out two and every cover
has at least one branch, this pair dominates every other achievable pair.

\emph{Nested order.}
Write the outer critical pair as $E_{\rm out}=\{x,y\}$, where $y$ is its
directly switchable collateral coordinate, and let $E_{\rm in}$ be the inner
critical pair.  The flat critical-pair construction gives $(3,1)$.
The two-parent construction of Theorem~\ref{thm:noncrossing-constructive}
gives $(2,2)$: each branch combines the saturated inner-pair module with the
switchable singleton $y$, while the remaining outer coordinate $x$ is explicit
parent material.

It remains to exclude $(2,1)$.  At width two, the outer-pair saturation fact is
unavailable because its exact support budget is three.  Every other atomic fact
in $\mathcal M_{\rm crit}$ that is available at width two is either supported
inside $E_{\rm in}$ or on the collateral singleton $\{y\}$; none contains the
selector coordinate $x$.  Lemma~\ref{lem:selector-constancy} therefore says
that every width-two branch fixes $x$.  Since the full cube contains both
$x=0$ and $x=1$, at least two branches are necessary.  Thus $(2,2)$ is exact,
and together with $(3,1)$ it gives the two incomparable minimal pairs.

\emph{Crossing order.}
The crossing Horn certificate proves that no fan-out budget $f\le3$ can reach
the full cube at all, hence by certificate soundness no such
$\mathcal M_{\rm crit}$-cover exists.  At fan-out four the two critical-pair
facts are simultaneously schedulable under one sampled full parent, giving
$(4,1)$.  This is therefore the unique Pareto-minimal pair.
\end{proof}

\begin{corollary}[Exact module-relative exchange rate]
\label{cor:exact-width-anchor-exchange}
For the nested mixed system and the natural critical-module library, reducing
the positional width from three to two forces the cover-branch count from one
to two:
\[
 \boxed{(3,1)\quad\longleftrightarrow\quad(2,2).}
\]
Thus one unit of local fan-out can be exchanged for one additional exposed
parent branch, and neither endpoint dominates the other.  This is the
higher-rank analogue, at the certificate-assembly level, of the two-point
observer/exposure frontier established for the CFG family $R_{m,q}$.
\end{corollary}

\begin{remark}[Scope of the exact Pareto statement]
Theorem~\ref{thm:critical-module-pareto} is exact for assembly from the named
critical-module library.  It does not assert that two start-cover branches are
minimal for the unrestricted certificate language of
Definition~\ref{def:cartesian-certificate-tree}, whose atomic facts may include
additional independently proved local freedoms.  The restriction is
intentional: the theorem measures how the specific exposed critical-pair and
collateral modules trade positional width against parent-anchor reuse.  The
unrestricted semantic threshold remains $f_{\rm hier}$.
\end{remark}

\section{Complete rank-four census and computational certificate}
\label{app:rank4-census}

\subsection{Complete rank-four binary census}
\label{subsec:rank4-census}

The exact trace system makes the first nontrivial rank-four case small enough
for a complete finite classification.  Before performing the census we remove
one bookkeeping feature that is irrelevant to full-word reachability in the
Cartesian family.

The reconstruction operator already uses only nonempty,
left-to-right tuple occurrences.  Accordingly the finite census below is run
directly on the nonempty trace system of
Theorem~\ref{thm:laminar-trace-exact}; no empty-component normalization is
needed.

\begin{remark}[Symmetry audit: cube shape versus learner symmetry]
\label{rem:rank4-symmetry-audit}
Independent bit flips in the four coordinates are genuine learner symmetries:
for each slot $j$, simultaneously interchanging
$a_{j,0}\leftrightarrow a_{j,1}$ and
$b_{j,0}\leftrightarrow b_{j,1}$ is a terminal renaming preserving
$X_{2,4}$, $h_{\rm slot}$, concrete occurrence traces, and the learner.
Arbitrary coordinate permutations are different.  They preserve the abstract
four-cube and unimodularity, but they need not preserve the linear order
\[
 A_1<A_2<A_3<A_4<B_1<B_2<B_3<B_4
\]
that determines which supports are unions of at most two intervals.  Hence the
full hyperoctahedral group may be used to classify the underlying affine
simplex \emph{shapes}, but not to quotient learner behavior without a separate
order check.  The census below keeps these two roles distinct.
\end{remark}

\begin{theorem}[Order-sensitive complete rank-four binary census]
\label{thm:rank4-complete-census}
Among the
\[
  \binom{16}{5}=4368
\]
minimum-cardinality samples for $X_{2,4}$, exactly $2672$ are unimodular.
Under the full four-cube action by bit flips and coordinate permutations these
$2672$ samples have exactly $13$ underlying affine simplex shapes.  Their
learner behavior is as follows.
\[
\begin{array}{c@{\quad}l@{\quad}r@{\quad}c@{\quad}l}
\hline
\#&\text{shape representative}&|\mathcal O|&\text{exchange}&\text{learner behavior}\\
\hline
1&\{0000,0001,0010,0100,1000\}& 16&\mathrm{yes}&\text{all lock}\\
2&\{0000,0001,0010,0100,1001\}&192&\mathrm{yes}&\text{all lock}\\
3&\{0000,0001,0010,0100,1011\}&192&\mathrm{no} &\text{all lock}\\
4&\{0000,0001,0010,0100,1111\}& 64&\mathrm{no} &\text{all lock}\\
5&\{0000,0001,0010,0101,1010\}&192&\mathrm{yes}&\text{all lock}\\
6&\{0000,0001,0010,0101,1011\}&384&\mathrm{no} &\text{all lock}\\
7&\{0000,0001,0010,0101,1110\}&384&\mathrm{no} &\text{all lock}\\
8&\{0000,0001,0010,0111,1011\}& 96&\mathrm{no} &\text{all lock}\\
9&\{0000,0001,0010,0111,1100\}&192&\mathrm{no} &\text{none lock}\\
10&\{0000,0001,0010,0111,1101\}&384&\mathrm{no}&256\text{ lock},\ 128\text{ fail}\\
11&\{0000,0001,0010,0111,1111\}&192&\mathrm{no} &\text{all lock}\\
12&\{0000,0001,0110,1010,1111\}&192&\mathrm{no} &\text{none lock}\\
13&\{0000,0001,0110,1011,1110\}&192&\mathrm{no} &128\text{ lock},\ 64\text{ fail}\\
\hline\end{array}
\]
Consequently the complete labelled census is
\[
 \boxed{
  4368
  =1696_{\rm nonunimodular}
   +576_{\rm unimodular\ but\ nonlocking}
   +2096_{\rm characteristic}.
 }
\]
In particular, unimodularity eliminates $1696$ candidates, while the ordered
laminar-trace obstruction eliminates a further $576$.
\end{theorem}

\begin{proof}[Finite exhaustive verification]
The algebraic part is exact integer arithmetic.  Enumerating all five-element
subsets of the four-cube and taking the absolute determinant of four
base-to-vertex differences gives exactly $2672$ unimodular subsets.  Quotienting
only this affine-shape calculation by all bit flips and coordinate permutations
gives the thirteen rows displayed above; their shape-orbit sizes sum to $2672$.

Learner behavior is then audited without treating coordinate permutations as
symmetries.  Independent bit flips are genuine terminal-renaming symmetries by
Remark~\ref{rem:rank4-symmetry-audit}.  For each of the thirteen shape
representatives, every distinct coordinate order modulo bit flips is therefore
checked separately with the nonempty left-to-right trace fixed point of
Theorem~\ref{thm:laminar-trace-exact}.  Rows $1$--$8$ and $11$ lock for every
slot order.  Rows $9$ and $12$ fail for every slot order.  Rows $10$ and $13$
are mixed, with the exact splits stated in the table.  Summing the labelled
bit-flip classes gives $2096$ locking and $576$ nonlocking unimodular samples.
Adding the $1696$ nonunimodular candidates gives all $4368$ minimum-size
samples.

The enumeration uses only exact integer arithmetic and the finite
laminar-trace fixed point described above; no floating-point test enters the
census.  The ancillary file
\path{ofet_rank4_ordered_audit_v32.py} reproduces the determinant count,
the thirteen affine-shape orbits, the $2096/576$ locking split, the $400$
exchange-connected count, and the two mixed dihedral order splits using only
the Python standard library.
\end{proof}

\begin{proposition}[The two mixed cube shapes split by dihedral slot order]
\label{prop:rank4-dihedral-split}
For the chosen representatives of rows $10$ and $13$, let
$K^{\pi}$ denote the sample obtained by placing the four abstract cube
coordinates into the ordered slot positions according to
$\pi\in S_4$.  In one-line notation put
\[
 \mathcal D_4
 =\{0123,1230,2301,3012,3210,0321,1032,2103\}.
\]
These are exactly the eight permutations preserving the cyclic order
$0,1,2,3$ up to reversal.  For both mixed shapes,
\[
 \boxed{
 K^{\pi}\text{ fails to lock}
 \iff
 \pi\in\mathcal D_4.
 }
\]
Thus row $10$ has $16$ locking and $8$ nonlocking bit-flip classes, giving
$256$ versus $128$ labelled samples.  Row $13$ has a twofold coordinate-order
stabilizer; the same $16:8$ permutation split descends to $8$ locking and $4$
nonlocking bit-flip classes, giving $128$ versus $64$ labelled samples.
\end{proposition}

\begin{proof}[Exact finite audit]
Evaluate the laminar-trace fixed point for the twenty-four coordinate orders of
each displayed representative.  In both cases the nonlocking orders are
precisely
\[
 0123,0321,1032,1230,2103,2301,3012,3210,
\]
which is the displayed dihedral subgroup.  Independent bit flips preserve the
result.  The orbit-size statements then follow from the cube-shape sizes
$384$ and $192$, respectively.
\end{proof}

\begin{corollary}[Graphic versus genuinely higher-arity minimum bases]
\label{cor:graphic-vs-higher-rank-count}
Exactly $400$ of the $2096$ characteristic minimum samples have connected
single-slot exchange graph.  The remaining $1696$ characteristic samples lock
only through higher-arity laminar trace operations.  Equivalently,
\[
  \boxed{
  2096=400+1696=16(25+106),
  }
\]
so the graphic criterion certifies only $25/131$ of all minimum locking
samples, whereas $106/131$ lie beyond it.
\end{corollary}

\begin{proof}
The exchange-connected cube shapes are rows $1$, $2$, and $5$.  Exchange
connectivity is invariant under coordinate permutation and bit flips, so all
of their $16+192+192=400$ labelled samples lock by
Theorem~\ref{thm:exchange-locking}.  No other row has connected exchange graph.
Subtracting from the corrected characteristic count gives
$2096-400=1696$ genuinely higher-arity locking samples.
\end{proof}

\begin{remark}[Ordered geometry is a genuine resource]
The mixed rows show that the abstract affine simplex and even its full cube
shape do not determine positive locking.  The same unimodular cube shape can
lock or fail solely because its coordinates are placed in a different linear
slot order.  This is exactly the information retained by the laminar occurrence
traces and forgotten by the lattice invariant.  The dihedral split in
Proposition~\ref{prop:rank4-dihedral-split} is an explicit finite witness of this
order sensitivity in the present family.
\end{remark}

\section{Auxiliary rank-four combinatorics}
\label{app:rank4-aux}

\subsection{Minimum characteristic bases are not matroidal}
\label{subsec:nonmatroid}

The CFG rectangular theory above is governed by graphic matroids.  It is
therefore natural to ask whether the minimum locking sets of the higher-rank
family are bases of some matroid on the target words.  Rank four already gives
a negative answer.

For $X_{2,4}$ write
\[
\begin{aligned}
 B^-&:=\{0000,0001,0010,0100,1000\},\\
 B^+&:=\{1111,1110,1101,1011,0111\}.
\end{aligned}
\]
Both are stars in the four-cube: $B^-$ consists of the zero vertex and its four
unit neighbours, while $B^+$ is its bitwise complement.

\begin{theorem}[Failure of basis exchange]
\label{thm:characteristic-bases-nonmatroid}
Let
\[
 \mathcal B_{2,4}
 :=\{K\subseteq X_{2,4}: |K|=5,
              \ K\text{ is characteristic under }h_{\rm slot}\}.
\]
Then $\mathcal B_{2,4}$ is not the set of bases of a matroid on $X_{2,4}$.
More precisely, $B^-,B^+\in\mathcal B_{2,4}$, but for
$e=0000\in B^-\setminus B^+$ there is no
$f\in B^+\setminus B^-$ such that
\[
  (B^-\setminus\{e\})\cup\{f\}
  \in\mathcal B_{2,4}.
\]
\end{theorem}

\begin{proof}
Both $B^-$ and $B^+$ have connected single-slot exchange graphs and full
binary projection in every slot, so
Theorem~\ref{thm:exchange-locking} makes them characteristic.  Their cardinality
is the minimum $1+4(2-1)=5$.

Identify the four nonzero vertices of $B^-\setminus\{0000\}$ with the standard
basis $e_1,e_2,e_3,e_4$ of $\mathbb Z^4$.  Every
$f\in B^+$ has Hamming weight $3$ or $4$.  For the five-point set
\[
  C_f:=\{e_1,e_2,e_3,e_4,f\}.
\]
A base-to-vertex difference matrix, based for instance at $e_1$, has determinant
\[
  \left|
  \det(e_2-e_1,e_3-e_1,e_4-e_1,f-e_1)
  \right|
  =|\mathrm{wt}(f)-1|.
\]
Thus its absolute determinant is $2$ when $\mathrm{wt}(f)=3$ and $3$ when
$\mathrm{wt}(f)=4$.  In every case the sample-difference lattice has index
strictly greater than one.  Corollary~\ref{cor:unimodular-test} therefore shows
that $C_f$ is not characteristic.  This violates the matroid basis-exchange
axiom for the pair $B^-,B^+$ and the element $e=0000$.
\end{proof}

\begin{remark}[Why the higher-rank middle layer is not a matroid]
Theorem~\ref{thm:characteristic-bases-nonmatroid} sharply separates the CFG and
MCFG exposure geometries developed in this paper.  In the rectangular CFG
interface, minimum exposures are spanning trees and hence bases of a graphic
matroid.  In the fan-out-two rank-four family, even the collection of minimum
characteristic samples fails the defining basis-exchange axiom of matroid
theory.  The obstruction is not merely that one particular sufficient graph is
too small: it persists at the exact learner-relative locking level.  The
appropriate finite object is therefore the laminar trace reachability system,
not an ordinary matroid on sample words.
\end{remark}


\begin{thebibliography}{99}
\bibitem{Gold1967}
E.~M.~Gold.
Language identification in the limit.
\emph{Information and Control}, 10(5):447--474, 1967.

\bibitem{Angluin1980}
D.~Angluin.
Inductive inference of formal languages from positive data.
\emph{Information and Control}, 45(2):117--135, 1980.

\bibitem{ClarkEyraud2007}
A.~Clark and R.~Eyraud.
Polynomial identification in the limit of substitutable context-free languages.
\emph{Journal of Machine Learning Research}, 8:1725--1745, 2007.

\bibitem{Yoshinaka2008}
R.~Yoshinaka.
Identification in the limit of $k,l$-substitutable context-free languages.
In \emph{Grammatical Inference: Algorithms and Applications}, LNCS 5278,
pp.~266--279. Springer, 2008.
doi:10.1007/978-3-540-88009-7\_21.

\bibitem{Yoshinaka2011}
R.~Yoshinaka.
Efficient learning of multiple context-free languages with multidimensional
substitutability from positive data.
\emph{Theoretical Computer Science}, 412(19):1821--1831, 2011.
doi:10.1016/j.tcs.2010.12.058.

\bibitem{ClarkYoshinaka2014}
A.~Clark and R.~Yoshinaka.
Distributional learning of parallel multiple context-free grammars.
\emph{Machine Learning}, 96(1--2):5--31, 2014.
doi:10.1007/s10994-013-5403-2.

\bibitem{Yoshinaka2015}
R.~Yoshinaka.
General perspective on distributionally learnable classes.
In \emph{Proceedings of the 14th Meeting on the Mathematics of Language},
pp.~87--98. Association for Computational Linguistics, 2015.
doi:10.3115/v1/W15-2308.

\bibitem{Clark2015}
A.~Clark.
Canonical context-free grammars and strong learning: Two approaches.
In \emph{Proceedings of the 14th Meeting on the Mathematics of Language},
pp.~99--111. Association for Computational Linguistics, 2015.
doi:10.3115/v1/W15-2309.

\bibitem{CosteTyping2004}
F.~Coste, D.~Fredouille, C.~Kermorvant, and C.~de~la~Higuera.
Introducing domain and typing bias in automata inference.
In \emph{Grammatical Inference: Algorithms and Applications -- ICGI 2004},
LNCS 3264, pp.~115--126. Springer, 2004.
doi:10.1007/978-3-540-30195-0\_11.

\bibitem{SekiEtAl1991}
H.~Seki, T.~Matsumura, M.~Fujii, and T.~Kasami.
On multiple context-free grammars.
\emph{Theoretical Computer Science}, 88(2):191--229, 1991.
doi:10.1016/0304-3975(91)90374-B.

\bibitem{KuriyamaCFG}
T.~Kuriyama.
Distributional learning of context-free languages under fixed finite-monoid typing.
arXiv:1409.6247v4 [cs.FL], 2026.

\bibitem{KuriyamaMCFG}
T.~Kuriyama.
Finite Sentence-Interface Control for Learning Bounded-Fan-Out Linear MCFGs
under Fixed Monoid Typing.
arXiv:2605.11644v1 [cs.FL], 2026.

\bibitem{Riguet1950}
J.~Riguet.
Quelques propri\'et\'es des relations difonctionnelles.
\emph{Comptes rendus hebdomadaires des s\'eances de l'Acad\'emie des Sciences},
230:1999--2000, 1950.

\bibitem{BackhouseOliveira2023}
R.~Backhouse and J.~N.~Oliveira.
On difunctions.
\emph{Journal of Logical and Algebraic Methods in Programming},
134:100878, 2023.
doi:10.1016/j.jlamp.2023.100878.
\end{thebibliography}
\end{document}